\newif\ifstoc
\stocfalse

\documentclass[11pt, letterpaper]{article}
\usepackage{xspace}
\usepackage{amsthm}
\usepackage{todonotes}
\usepackage[ruled,vlined,linesnumbered]{algorithm2e}
\usepackage{fullpage}
\usepackage{dsfont}
\usepackage{amssymb}
\usepackage{amsmath}
\usepackage{graphicx}
\usepackage[hidelinks]{hyperref}
\usepackage[inline]{enumitem}
\usepackage{mathtools}

\ifstoc
\else
\usepackage{authblk}
\fi

\usepackage{tikz}
\usetikzlibrary{positioning,calc}
\usepackage{environ}
\makeatletter
\newsavebox{\box@tikzpicture}
\NewEnviron{tikzpicture*}[2][]{%
	\begin{lrbox}{\box@tikzpicture}%
		\begin{tikzpicture}[#1]
		\BODY
		\end{tikzpicture}%
	\end{lrbox}
	\pgfmathsetmacro\width@scale@picture{#2/\wd\box@tikzpicture}%
	\begin{tikzpicture}[#1,scale=\width@scale@picture]
	\BODY
	\end{tikzpicture}%
}%
\makeatother

\theoremstyle{plain}
\newtheorem{theorem}{Theorem}
\newtheorem{lemma}{Lemma}
\newtheorem{observation}{Observation}

\theoremstyle{definition}

\newtheorem{definition}{Definition}

\newtheorem*{problem}{Problem}

\theoremstyle{remark}

\newcommand{\openC}{\eta}

\newcommand{\dist}{\text{dist}}

\newcommand{\Reals}{{\mathds R}}
\newcommand{\B}{{\mathcal B}}

\newcommand{\E}{{\mathcal E}}
\newcommand{\G}{{\mathcal G}}

\newcommand{\hp}{{\mathit{HP}}}

\newcommand{\M}{{\mathcal M}}
\newcommand{\inte}{{\textnormal{int}}}
\newcommand{\exte}{{\textnormal{ext}}}
\newcommand{\lengt}{{\textnormal{length}}}
\newcommand{\cost}{{\textnormal{cost}}}
\newcommand{\cone}{{\textnormal{cone}}}
\newcommand{\parent}{{\textnormal{par}}}
\newcommand{\children}{{\textnormal{children}}}
\newcommand{\fix}{{\textnormal{fixedcost}}}
\DeclareMathOperator*{\argmin}{arg\,min}

\newcommand{\tCC}{\tilde{\clustering}}

\newcommand{\R}{{\mathcal R}}

\newcommand{\clustering}{\mathcal{C}}
\newcommand{\C}{\clustering}        
\newcommand{\COpt}{\C^{\textsc{Opt}}}

\newcommand{\Part}{\mathrm{Part}}

\newcommand{\Table}{\mathrm{Table}}
\newcommand{\diam}{\mathnormal{diam}}
\newcommand{\length}{\mathnormal{length}}

\def\polylog{\operatorname{polylog}}

\newcommand{\CH}{H}
\newcommand{\per}{h}
\newcommand{\bound}{\partial}
\newcommand{\Opt}{\mbox{\sc Opt}}
\newcommand{\mydef}{\coloneqq}

\newcommand{\outside}{\exte} 
\newcommand{\inter}{\inte} 

\newcommand{\pol}{\pi} 
\newcommand{\cl}{\mathcal{C}} 

\newcommand{\facilityname}{unit disk fencing problem}
\newcommand{\fixedkname}{$k$-cluster fencing problem}
\newcommand{\bestc}[1]{\ensuremath{\sigma_{\min}( #1 )}}

\newcommand{\calL}{\mathcal{L}}

\usepackage{color}
\definecolor{blueblack}{rgb}{0,0,.7}

\newcounter{sideremark}

\newcommand{\vincent}[1]{}

\newcommand{\eva}[1]{}

\title{Fast Fencing}
  
\author[1]{Mikkel Abrahamsen}
\author[1]{Anna Adamaszek}
\author[2]{Karl Bringmann}
\author[3]{Vincent Cohen-Addad}
\author[4]{Mehran Mehr}
\author[5]{Eva Rotenberg}
\author[1]{Alan Roytman}
\author[1]{Mikkel Thorup}
\affil[1]{Department of Computer Science, University of Copenhagen, Denmark}
\affil[2]{Saarland Informatics Campus, Max Planck Institute for Informatics, Germany}
\affil[3]{Sorbonne Universit\'e, CNRS, Laboratoire d'informatique de Paris 6, LIP6, F-75252 Paris, France}
\affil[4]{Department of Computer Science, TU Eindhoven, Netherlands}
\affil[5]{Department of Applied Mathematics and Computer Science, Technical University of Denmark, Denmark}

\date{}
\begin{document}
\pagenumbering{gobble}
\thispagestyle{empty}
\maketitle
\begin{abstract}
We consider very natural ``fence enclosure'' problems studied by Capoyleas, Rote, and Woeginger and Arkin, Khuller, and Mitchell in the early 90’s. Given a set $S$ of $n$
points in the plane, we aim at finding a set of closed curves such
that (1) each point is enclosed by a curve and (2) the total length of
the curves is minimized.  We consider two main variants. In the first variant,
we pay a unit cost per
curve in addition to the total length of the curves. An equivalent formulation of this version is that
we have to enclose $n$ unit disks, paying only the total length of the
enclosing curves.
In the other variant, we are allowed to use at most $k$ closed curves and pay no cost per curve.

For the variant with at most $k$ closed curves,	we present
an algorithm that is polynomial	in both	$n$ and	$k$. For
the variant with unit cost per curve, or unit disks, we	
present	a near-linear time algorithm.

Capoyleas, Rote, and Woeginger solved the problem with at most
$k$ curves in $n^{O(k)}$ time.
Arkin, Khuller, and Mitchell used this to solve the
unit cost per curve version in exponential time. At
the time, they conjectured that the problem with $k$ curves
is NP-hard for general $k$. Our polynomial time algorithm
refutes this unless P equals NP.
\end{abstract}

\newpage
\pagenumbering{arabic}
\section{Introduction}
We consider some very natural ``fence enclosure'' problems studied by
Capoyleas, Rote, and Woeginger~\cite{CRW91} and
Arkin, Khuller, and Mitchell~\cite{AKM93} in the early 90s. Given a
set $S$ of $n$ points in the plane, we aim at finding a set of closed
curves such that (1) each point is enclosed by a curve and (2) the
total length of the curves is minimized.  We consider two main
variants.
In the first variant, we pay an opening cost $\openC>0$ per curve, which is part of the input.
An equivalent formulation is that a circle with radius $\openC/2\pi$ is centered at each point, and we need to enclose these circles with curves of minimal total length, paying no opening cost.
The equivalence is illustrated and explained in Figure~\ref{fig:protection}.
By a suitable scaling, we may assume that the circles are unit circles.
We thus refer to this variant as the {\em \facilityname{}}.
In the other variant, we are allowed to use at most $k$ closed curves and pay no cost per curve. We can
think of this as dividing the points into $k$ clusters and then viewing the
closed curves as perimeters of the convex hulls of
the clusters. For this reason, we refer to the variant as the {\em \fixedkname{}} (also referred to as the {\em minimum perimeter sum
problem} in the literature).

Capoyleas, Rote, and Woeginger~\cite{CRW91} presented an algorithm for the \fixedkname{} that runs in time $\rho(k)n^{O(k)}$, where $\rho(k)$ is the number of nonisomorphic
planar graphs on $k$ nodes.  This yields a polynomial running time
when $k$ is fixed.
Arkin et al.~\cite{AKM93} conjectured the problem to be
NP-hard when $k$ is part of the input and neither an NP-hardness proof nor a polynomial time algorithm has been found so far.
To solve the unit disk
version, which is equivalent to having a unit cost per cluster, Arkin et al.~suggested running their algorithm for the $k$-cluster fencing problem for all $k\leq n$, adding $k$ to the
total perimeter length. These were the best known bounds
except in the special case of the $k$-cluster fencing problem for $k\in\{1,2\}$ (see the discussion on related work
for more details).

\subsection{Our Results}
We present polynomial time algorithms for both problems.
More specifically, for the \facilityname{}, we present an efficient near-linear time algorithm (Theorem~\ref{thm:facility}).
For the \fixedkname{}, we present an algorithm that is
polynomial in both $n$ and $k$ (Theorem~\ref{thm:fixedk}).
In particular, this refutes
the conjectured hardness unless P $=$ NP.

Our algorithm for the \facilityname{} can be generalized to the case
where the input consists of objects that are allowed to be disks with different
diameters or polygonal objects that have to be fenced. For this
variant,  our
running time increases by a factor that is logarithmic in
the ratio between the maximum and the minimum object diameter.

In order to achieve near-linear bounds for the \facilityname{}, we
introduce new techniques that we believe can have other applications
in computational geometry. We give a detailed overview of the techniques
below.

Throughout our paper, we assume that it is possible to compare the costs of two different
clusterings efficiently. Note that this is a standard assumption in computational geometry.

\subsection{Applications}
The problem of fencing in disks or objects appears very commonly
in the real world. A good example is the protection of trees,
either at construction sites to protect the roots, or in the
wild to protect rare trees from deer and other animals. When
trees are planted by nature, we have no control over their
location. In this context, each disk should have a sufficient
diameter to protect rare trees from wildlife (see Figure~\ref{fig:protection}).

There are many standards that specify how far fences should be from trees,
and even discussions on different advantages of grouping trees beyond the fence cost (e.g., see~\cite{TreeProtection}).

\begin{figure}
\centering
\includegraphics[width=0.5\textwidth, trim={1.5cm 0.5cm 1.4cm 0.9cm},clip]{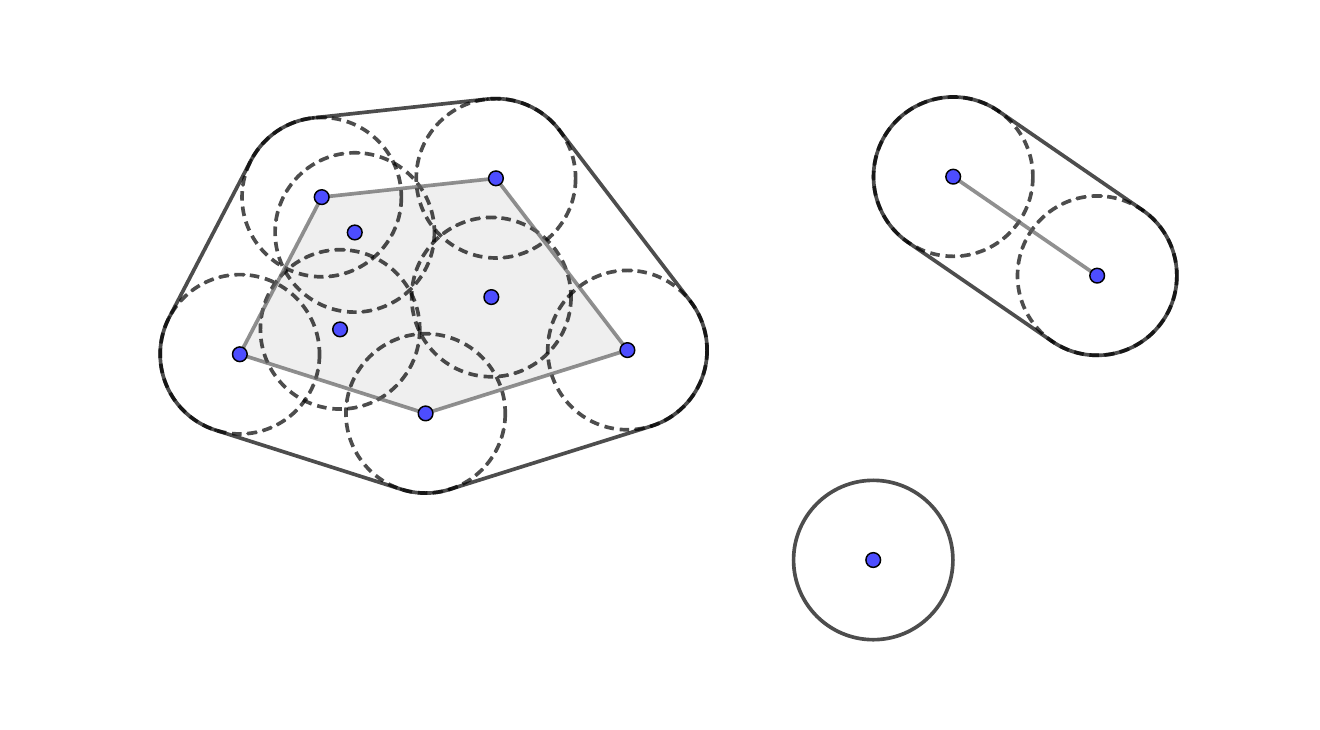}
\caption{A set of $11$ points and the set of enclosing curves minimizing the total length plus the opening cost $\openC$ per curve.
There are three curves, enclosing $1$, $2$, and $8$ points, drawn in gray (one cannot be seen, as it consists of just one input point).
A dashed circle centered at each point with radius $\openC/2\pi$ is drawn.
For each cluster, the curve enclosing the respective circles is drawn in black.
We note that the perimeter of each black curve is exactly $\openC$ larger than the perimeter of the gray curve, as the linear pieces sum to the perimeter of the gray curve and the circular arcs sum to $\openC$.
Hence, the problem of enclosing points with curves, minimizing the total length plus an opening cost $\openC$ per curve, is equivalent to that of enclosing circles with curves of minimal total length.}\label{fig:protection}
\end{figure}



\subsection{Our Techniques}
\paragraph{Approach for the \facilityname{}.}
Our approach is as follows.
We first partition the plane and recursively subdivide it into four quadrants, using a quadtree dissection-type approach.
This divides the plane into \emph{cells} of geometrically decreasing sizes.
Each cell of width $w$ consists of four cells of width $w/2$.
We find optimal partitions for the smallest cells first (i.e., the lowest level of the quadtree).
We then work with increasing cell sizes, solving the problem layer by layer in the quadtree.
To obtain a solution at a given cell size, we rely on solutions to smaller cell sizes.
For this to work, we need to have precomputed the solutions for various \emph{polyominoes} made of some constant number of cells of smaller sizes.

For a given polyomino, we show how to obtain an optimal clustering of the points within the polyomino by merging the solutions for smaller polyominoes.
In some cases, merging is not enough as there can be one large cluster intersecting all the cells of the polyomino, which we do not find by merging solutions of smaller polyominoes.

Hence, we need an efficient algorithm for finding the best cluster $C$ intersecting all the cells of the polyomino.
In order to do this, we give a subroutine running in time $O(n\log^3 n)$ that finds the best cluster that intersects all cells.
The subroutine works in two steps: it first finds a point $x_0$ that belongs to the cluster
we are looking for together with a point $p$ on the boundary of the cluster.
%
Once we have $x_0$ and $p$, the idea is to make an angular sweep of a ray from $x_0$, and consider the points in the order the ray sweeps over them.
For each point $v$, we calculate the ``best'' path from $p$ to $v$, in terms of both its length and how many clusters' opening costs it may save.
We prove that the ``best'' path consists of line segments between points, and save for each $v$ information about the last line segment on the path to $v$
(see the red lines in Figure~\ref{fig:curvetree}).
This information allows us to finally retrieve the boundary of the convex hull of $C$ recursively.

\begin{figure}[ht!]
\begin{center}
\begin{tikzpicture*}{0.33\textwidth}
\begin{scope}[
tinyvertex/.style={
draw,
circle,
minimum size=1mm,
inner sep=0pt,
outer sep=0pt%
},
every label/.append style={
font=\small,
}
]
\node[tinyvertex,label={left:$x_0$}] (x) at (0,0) {};
\node[tinyvertex,label={left:$p$}] (p) at (0,5) {};
\node[tinyvertex,label={below right:$t$}] (t) at (4.5,2.7) {};

\draw plot coordinates {(1,5) (1.6,5.2) (1.1,2.3) (1,5)};

\draw plot coordinates {(2.5,3) (2.75,3.5) (3,3) (2.75,2.5) (2.5,3)};

\draw plot coordinates {(4,1) (2.2,0.8) (4.5,2.7) (4,1)};

\draw [dashed] plot coordinates {(0,0) (5,3)};

\draw [<-] (6,0) arc (0:90:7.4);

\draw [red] plot coordinates {(0,5) (1.6,5.2) (2.75,3.5)};
\draw [red] plot coordinates {(0,5) (1,5)};
\draw [red] plot coordinates {(0,5) (1.1,2.3)};
\draw [red] plot coordinates {(1.6,5.2) (2.5,3)};
\draw [red,dashed] plot coordinates {(1.6,5.2) (4.5,2.7)};

\end{scope} 
\end{tikzpicture*}
\end{center}
\caption{Given the advice that $p$ lies on the perimeter of the cluster containing $x_0$ in an optimal clustering, we can compute the cluster containing $x_0$ with an angular sweep. Here, the angular sweep has reached the point $t$.}
\label{fig:curvetree}
\end{figure}
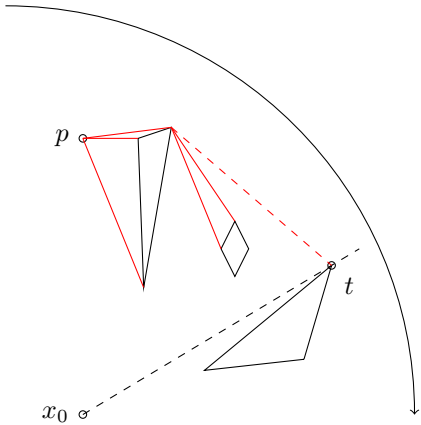


\paragraph{Approach for the \fixedkname{}.}
Our algorithm for the \fixedkname{} shares some similarities with the work of Gibson
et~al.~\cite{GKKPV12}, in which a dynamic programming approach was
used to solve the minimum radius sum problem.  One main difference
is that the running time complexity they obtain is $O(n^{881})$, whereas our approach
works in $O(n^{27})$ time.
Their technique is to
divide a given instance $(S,k)$ of the minimum radius sum problem into
subproblems, where $S$ is the set of points to be clustered and $k$ is
a constraint on the number of clusters we can use.  The problem
$(S,k)$ is solvable if we have a solution to all subproblems.
However, the number of subproblems is exponential.  To get a
polynomial time algorithm, they showed that a solution can be found
after considering only a polynomial number of subproblems.

We also use a dynamic programming approach for our problem, although we need new
techniques since the number of candidate clusters for the minimum radius sum problem
is only $O(n^2)$ (as each disk is determined by two points in $S$, one determining the center
and the other determining the radius).  In contrast, our problem has
an exponential number of candidate clusters (dictated by all subsets of $S$).
We define subproblems based on \emph{boxes}, which are rectangles that cover
some portion of the plane and some number of input points from $S$.  Our key observation
is that there is some separator of each box (i.e., a vertical line segment or a horizontal line
segment) that splits the box into two strictly smaller boxes such that an optimal solution
only has a constant number of line segments that intersect this separator (in fact, we give a
bound of two on the number of such intersecting line segments).

A dynamic programming approach naturally follows: we simply guess the position of such a separator
(for which we argue there are $O(n)$ choices), and then guess which segments crossing this separator
belong to an optimal solution.  We first obtain solutions for smaller boxes, and then glue
together solutions for smaller boxes to obtain solutions for larger boxes.

We note that there is an unpublished solution~\cite{M17} (i.e., polynomial time algorithm)
for the rectilinear version of the problem,
where we must enclose $n$ points using $k$ axis-parallel rectangles rather than
convex hulls (as in our setting).  The solution to the axis-parallel version uses similar ideas.
In particular, it is possible to argue the existence of separators that do not cut through any clusters of
an optimal solution in this setting.
This is not possible for our problem.
For the \fixedkname{}, it is possible that any such vertical or horizontal separator cuts at least one cluster, and allowing skew separators would
result in subproblems of high complexity.
In particular, consider an example with $n$ sufficiently large and assume $k<n/3$.
We have $n-k+1$ points spread evenly on a circle as the corners of a regular $(n-k+1)$-gon,
and $k-1$ points spread evenly on a surrounding circle with the same center.
The surrounding points are fairly close yet sufficiently far enough away that the optimal
solution is to cluster the inner $n-k+1$ points
together, and open a cluster for each of the $k-1$ points on the surrounding circle.
Cutting away the $k-1$ points on the outside (with not necessarily axis-aligned separators)
creates subproblems defined on polygonal regions with $k-1$ sides, resulting in high complexity subproblems.

\subsection{Related Work}
The literature on geometric clustering is vast~\cite{agarwal1998efficient}, and thus we focus on the
most relevant prior works.
Arkin, Khuller, and Mitchell~\cite{AKM93} considered many clustering variants related to the problems studied in the present paper.  For the variant
where points have a value associated with them, they showed that the problem of maximizing
profit (i.e., sum of values of points enclosed minus the total perimeter) is NP-hard when values are
unrestricted in sign.  When values are strictly positive, they gave an
$O(n^3)$ time algorithm.  For the version in which there is a budget on the total perimeter
we can use, the problem of maximizing profit is also NP-hard, even when
values are strictly positive (they provided a pseudo-polynomial algorithm when the values
are integers).

The $k$-cluster fencing problem for $k=1$ is the very well-known problem of computing the convex hull of a set of points in the plane~\cite{de2000computational}.
There has also been some work for the special case of $k=2$ clusters.
The work of Mitchell and Wynters~\cite{MW91} studied four flavors
of the problem: minimizing the sum of perimeters, the maximum of the perimeters,
the sum of the areas enclosed by the fences, and the maximum of the areas.
They gave polynomial time solutions for all four flavors, running in time
$O(n^3)$ (for some of them, they gave improved running time bounds of $O(n^2)$).
More recently, the work of Abrahamsen et al.~\cite{ABBMM17} gave an algorithm
running in time $O(n \log^4(n))$ that solves the case of $k=2$ clusters,
yielding the first subquadratic time algorithm for this setting.

There have been many other papers studying related geometric clustering problems.
Capoyleas, Rote, and Woeginger~\cite{CRW91} studied a general
geometric $k$-clustering framework in which the cost of a solution is
determined by some weight function that assigns real weights to any
subset of points in the plane (i.e., each cluster), after which a symmetric $k$-ary
function over $k$-tuples is applied (e.g., the sum function).
For the case when the weight function is the diameter, radius, or perimeter and
the symmetric $k$-ary function is an arbitrary monotone increasing
function (such as the sum or the maximum), they gave an algorithm running in time
$\rho(k)n^{O(k)}$, where $\rho(k)$ is the number of nonisomorphic
planar graphs on $k$ nodes.
This is polynomial if $k$ is fixed and not given as input.

In addition, the work of Behsaz and Salavatipour~\cite{BS15} studied objectives
such as minimizing the sum of radii and minimizing the sum of diameters subject
to the constraint of having at most $k$ clusters.  For
minimizing the sum of radii, they gave a polynomial time algorithm for clustering
points in metric spaces that are induced by unweighted graphs, assuming no singleton
clusters are allowed.  They also showed that finding the best single cluster for each connected
component of the graph yields a $\frac{3}{2}$-approximation algorithm, assuming
no singleton clusters are allowed.  For the problem of minimizing the sum of
diameters, they gave a polynomial time approximation scheme when points
lie in the plane with Euclidean distances, along with a polynomial time
exact algorithm when $k$ is constant (for the metric setting).

Many classical clustering problems are NP-hard when $k$ is given as part of the input,
though there are some notable exceptions.  In 2012, Gibson et~al.~\cite{GKKPV12} devised a
polynomial time algorithm for finding $k$ disks, each centered at a point in $S$, such that the sum of the
radii of the disks is minimized subject to the constraint that their union must cover~$S$.
In their paper, they used a dynamic programming approach to get a running time of $O(n^{881}T(n))$,
where $T(n)$ is the time needed to compare two candidate solutions.

\section{The Unit Disk Fencing Problem}

Given a set of points $A$ in the plane, we denote by $H(A)$ the convex hull of $A$, and by $h(A)$ the perimeter of $H(A)$.

  Let $A$ be a finite set of points in $\mathbb{R}^2$ and let $\openC>0$ be the \emph{opening cost}.
Consider a partition $\cl\mydef\{C_1,\ldots,C_{\ell}\}$ of $A$.
We refer to each set $C_i$ as a \emph{cluster}.
  The \emph{cost} of $\cl$ with respect to $\openC$ is
  $$\openC \cdot \ell + \sum_{i=1}^{\ell} h(C_i).$$
  The partition $\cl$ is \emph{optimal} for $(A,\openC)$ if no partition of $A$ has a lower cost.
  We denote the cost of an optimal partition for $(A,\openC)$ as $\Opt(A,\openC)$.
  When the opening cost is clear from the context, we might omit it.
In the \emph{\facilityname{}}, we are given a set of points $A$ and an opening cost $\openC$, and the goal is to find an optimal partition for $(A,\openC)$.

\begin{observation}
In an optimal partition, the clusters have pairwise disjoint convex hulls.
\end{observation}

We say that $A$ is \emph{indivisible} if $\{A\}$ is an optimal partition for $(A,\openC)$.

\begin{observation}
Each cluster of an optimal partition is indivisible.
\end{observation}

We say an optimal partition $\{C_1,\ldots,C_{\ell}\}$ for $(A,\openC)$ is \emph{maximal} if there is no optimal partition $\{C'_1,\ldots,C'_{\ell'}\}$ of $(A,\openC)$, where $\ell' < \ell$, such that for each $i \in \{1,\ldots,\ell\}$, there is some $j\in \{1,\ldots,\ell'\}$ such that $C_i \subseteq C'_j$.

\subsection{Structural Results}

\begin{figure}
\centering
\includegraphics[scale=0.5]{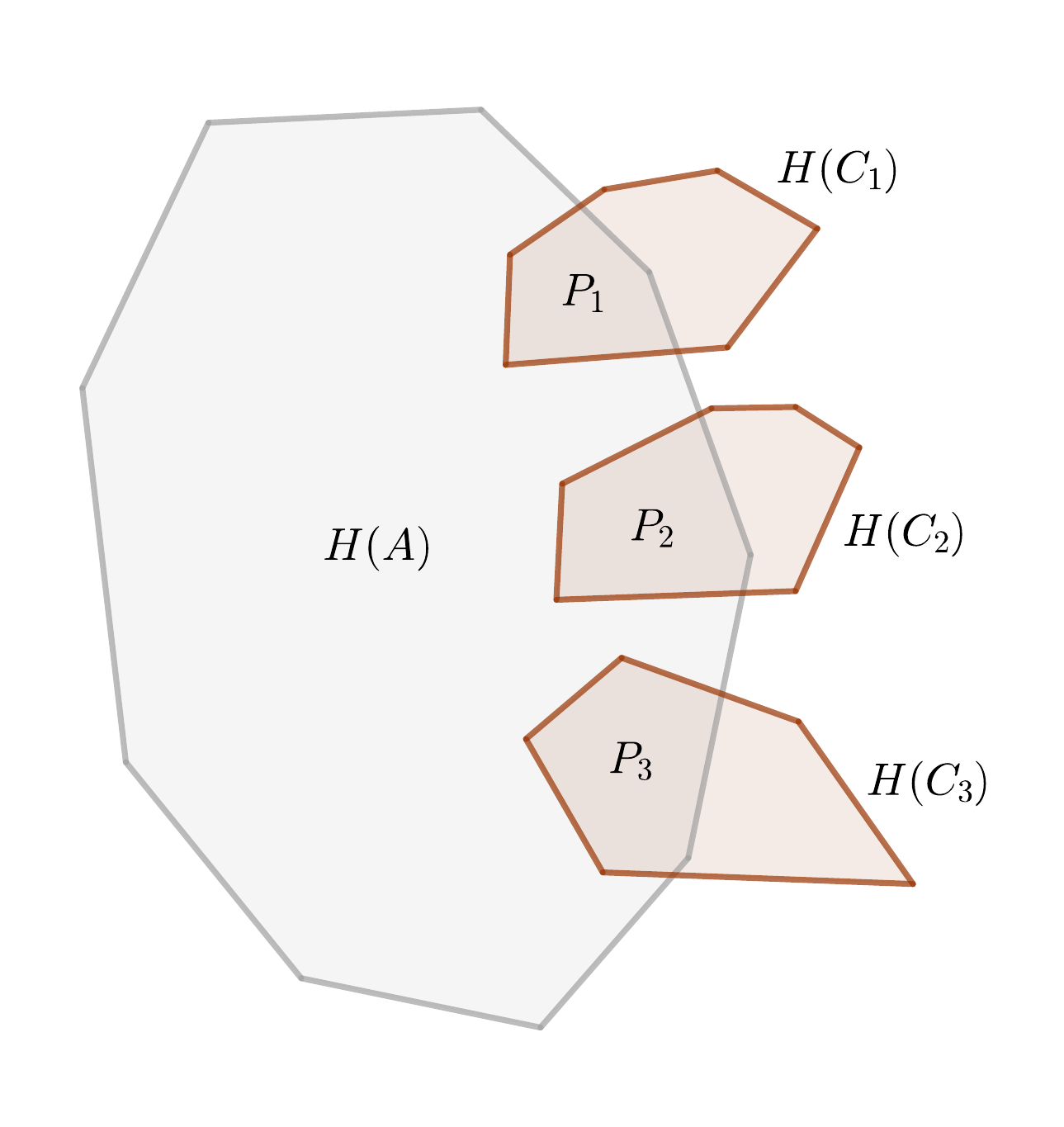}
\caption{The polygon $Q$ from the proof of Lemma~\ref{lem:sum-of-indivisible-clusters} is the set $H(A)\cup\bigcup_{i=1}^{3} H(C_i)$.
Note that the perimeter is $h(Q)=h(A) + \sum_{i=1}^{3} (h(C_i)-h(P_i))$.}
\label{fig:lem:sum-of-indivisible-clusters}
\end{figure}

\begin{lemma}\label{lem:sum-of-indivisible-clusters}
  Let $A,B$ be two indivisible sets of points in $\mathbb R^2$ under the opening cost $\openC$.
  If $H(A)$ and $H(B)$ intersect, then the set $A \cup B$ is indivisible.
\end{lemma}

\begin{proof}
  Let $\cl$ be a maximal optimal partition of $(A \cup B,\openC)$.
  Let $\{C_1,\ldots,C_m\}$ be the clusters from $\cl$ such that $A \cap C_i \neq \emptyset$ and $B \cap C_i \neq \emptyset$ for $1 \le i \le m$.
  Assume for contradiction that $m=0$.
  This means that $A\cap B=\emptyset$ and each cluster of $\cl$ is a subset of either $A$ or $B$, and hence that $|\cl|\geq 2$.
  If $|\cl|=2$, then $\cl=\{A,B\}$, and $\cl$ cannot be optimal as $H(A)$ and $H(B)$ intersect.
  If $|\cl|>2$, then $\cl$ cannot be maximal as $A$ and $B$ are indivisible.
  Hence $m\geq 1$.

  We want to show that the set $A'  \mydef A \cup \bigcup_{i=1}^m C_i$ is indivisible.
  Let $\{A_1, \ldots, A_j\}$, where possibly $j=0$, be the clusters of $\cl$ that are contained in $A$.
  Then the set $\mathcal A\mydef\{A_1,\ldots,A_j, C_1, \ldots, C_m\}$ is a partition of $A'$.
  Note that $\mathcal A$ must be an optimal partition of $A'$, as otherwise $\cl$ would not be optimal for $A\cup B$.

  The cost of $\mathcal A$ is thus
  $$\Opt(A')=(j+m) \openC + \sum_{i=1}^j h(A_i) + \sum_{i=1}^{m} h(C_i).$$
  Now, for each cluster $C_i$, we define a polygon $P_i = H(A) \cap H(C_i)$.
  As $H(A)$ and $H(C_i)$ are both convex, so is $P_i$.

  Note that all points in $A'$ are contained in the polygon $Q\mydef H(A)\cup\bigcup_{i=1}^{m} H(C_i)$, see Figure~\ref{fig:lem:sum-of-indivisible-clusters}.
  It hence follows that
  $$h(A') \le h(Q)= h(A) + \sum_{i=1}^{m} (h(C_i) - h(P_i)).$$
  Note that $\{A_1,\ldots,A_j,P_1\cap A,\ldots, P_m\cap A\}$ is a partition of $A$.
  Hence, by the indivisibility of $A$, we have
  $$\openC + h(A) \le (j+m) \openC + \sum_{i=1}^j h(A_j) + \sum_{i=1}^{m} h(P_i).$$
  Combining these two inequalities yields
  $$h(A') +\openC \le (j+m) \openC + \sum_{i=1}^j h(A_j)
  + \sum_{i=1}^{m} h(C_i)=\Opt(A'),$$
  and so $A'$ is indivisible.

 As $A'$ is indivisible, $A'$ is the union of clusters of $\cl$, and $\cl$ is maximal, it follows that $A'$ is itself a cluster of $\cl$, i.e., $m=1$ and $j=0$.
 Recall furthermore that $A'$ contains $A$ and intersects $B$.
 In a similar way, it can be shown that $B'  \mydef B \cup \bigcup_{i=1}^m C_i$ is a cluster of $\cl$ that contains $B$ and intersects $A$.
 Since $A'\cap B'\neq\emptyset$ and $\cl$ is optimal, we must have $A'=B'=A\cup B$, so $\cl=\{A\cup B\}$ and $A\cup B$ is indivisible.
\end{proof}

\begin{lemma}\label{lem:cluster-contained}
Let $A \subseteq B$ be sets of points in $\mathbb{R}^2$, and let $A$ be indivisible.
Let $\cl = \{C_1,\ldots,C_{\ell}\}$ be a maximal optimal partition of $B$.
Then $A\subseteq C_i$ for some $i\in\{1,\ldots,\ell\}$.
\end{lemma}

\begin{proof}
Let $\mathcal S\subseteq\cl$ be the set of clusters of $\cl$ that intersect $A$.
By Lemma~\ref{lem:sum-of-indivisible-clusters}, each of the sets $S\cup A$, where $S\in\mathcal S$, is indivisible.
It thus follows that $A\cup\bigcup_{S\in\mathcal S} S=\bigcup_{S\in\mathcal S} S$ is also indivisible.
Since $\cl$ is maximal, it must then be the case that $\mathcal S$ consists of a single cluster $C_i$ that contains $A$.
\end{proof}

\begin{lemma}\label{lem:max-opt-part}
Each instance $(A,\openC)$ of the \facilityname{} has a unique maximal optimal partition.
\end{lemma}

\begin{proof}
Consider two maximal optimal partitions $\cl = \{C_1,\ldots,C_{\ell}\}$ and $\cl' = \{C'_1,\ldots,C'_{\ell'}\}$ of $A$.
Lemma~\ref{lem:cluster-contained} gives that each cluster $C_i$ of $\cl$ is contained in some cluster $C'_j$ of $\cl'$.
Likewise, each cluster $C'_j$ of $\cl'$ is contained in some cluster $C_i$ of $\cl$.
Since the clusters of an optimal partition are disjoint, it now follows that there is a one-to-one correspondence between the partitions $\cl$ and $\cl'$, i.e., they are the same partition of $A$.
\end{proof}

\begin{lemma}\label{lem:big-cluster-property}
Let $A = \bigcup_{i \in I} A_i$ be a set of points such that $A = \bigcup_{i \in I} A^\complement_i$, where $A^\complement_i \coloneqq A\setminus A_i$ for $i \in I$.
Let $\cl_i$ be the maximal optimal partition of $A^\complement_i$, and $\cl$ the maximal optimal partition of $A$.
Any cluster of $\cl_i$ that intersects a cluster $C\in\cl$ is also contained in $C$, and it follows in particular that each cluster $C\in\cl$ is the union of clusters of the partitions $\cl_i$.
Furthermore, each cluster $C \in \cl$ is either a cluster of some partition $\cl_i$ or has a non-empty intersection with each set $A_i$.
\end{lemma}

\begin{proof}
In order to prove that first part, consider some cluster $C_i\in\cl_i$ that intersects $C\in\cl$.
Since $C_i$ is indivisble, Lemma~\ref{lem:cluster-contained} gives that there is a cluster of $\cl$ containing $C_i$.
That cluster must be $C$, as the clusters of $\cl$ are disjoint.
It hence also follows that each cluster $C\in\cl$ is the union of some clusters of the partitions $\cl_i$.

In order to prove the second part, consider a cluster $C \in \cl$ that does not intersect some set $A_i$.
Then, $C \subseteq A^\complement_i$.
Since $C$ is indivisible, Lemma~\ref{lem:cluster-contained} gives that there is a cluster $C_i$ of $\cl_i$ such that $C\subseteq C_i$.
By the above, we also have that $C_i\subseteq C$.
Hence $C=C_i$, i.e., $C$ is a cluster of $\cl_i$.
\end{proof}

\subsection{Partitioning into Independent Instances}\label{sec:partitioning}

Consider an instance $(A,\openC)$ of the \facilityname{}.
Observe that any two points $p_1,p_2 \in A$ such that $\|p_1p_2\| < \openC /2$ must be in the same cluster in any optimal partition of $A$.
We will prove that we can efficiently decompose the problem instance $(A,\openC)$ into a collection of independent subinstances $(A_i,\openC)$, such that for each subinstance $\diam(A_i)=\textrm{poly}(|A_i|) \cdot \openC$.

\begin{lemma}
\label{lem:aspectratio}
Any $n$-point instance $(A,\openC)$ of the \facilityname{} can be reduced in time $O(n \log^2 n)$ into a disjoint collection of subinstances $(A_i,\openC)$,
where $A = \bigcup_{i} A_i$ and each subinstance is bounded by a box of side lengths at most  $2 |A_i|^2 \cdot \openC$.
The subinstances are independent in the sense that an optimal partition for $(A,\openC)$ is the union of the optimal partitions for $(A_i,\openC)$.
\end{lemma}

\begin{proof}
Clearly, the optimal partition for $(A,\openC)$ has cost of at most $n \cdot \openC$, as a partition of cost $n \cdot \openC$ can be obtained by opening $n$ singleton clusters. Therefore, if any two points are at a distance greater than $n \cdot \openC / 2$, they must be in separate clusters of any optimal partition (as the perimeter of any cluster containing such two points would be greater than $\Opt(A)$).

We first sort all points from $A$ with respect to their $x$-coordinate. Denote the sorted points by $p_1, \ldots, p_n$. Whenever for two consecutive points $p_i, p_{i+1}$ the difference in their $x$-coordinate is greater than $n \cdot \openC / 2$, we know that the sets of points $\{p_1, \ldots, p_i\}$ and $\{p_{i+1}, \ldots, p_n\}$ can be treated separately, i.e., each cluster of any optimal partition will be contained in either $\{p_1, \ldots, p_i\}$ or $\{p_{i+1}, \ldots, p_n\}$. That gives us a partition of $\{p_1, \ldots, p_n\}$ into subinstances, where each subinstance is contained in a vertical slab of width at most $n^2 \cdot \openC / 2$. Now, for each subinstance we sort the points according to their $y$-coordinate, and perform a similar operation. Therefore, in time $O(n \log n)$ we partitioned $(A,\openC)$ into subinstances $(A_i,\openC)$, such that $A = \bigcup_{i} A_i$, each subinstance has a bounding box of size at most $n^2 \cdot \openC/2 \times n^2 \cdot \openC/2$, and an optimal partition for $(A,\openC)$ is the union of the optimal partitions for $(A_i,\openC)$.

Note that if $|A_i| \ge |A|/2$, then $n^2 \cdot \openC/2 = |A|^2 \cdot \openC/2 \le 2 |A_i|^2 \cdot \openC$, and the size of the bounding box is as required.
It therefore remains to consider subinstances $(A_i,\openC)$ for which $|A_i| < |A|/2$ and the smallest bounding box has size at least $2|A_i|^2\cdot\openC$.
In such an instance, there must be two consecutive points in the order of $x$- or $y$-coordinates where the respective coordinates differ by at least $\frac{2|A_i|^2\cdot\openC}{|A_i|-1}>|A_i|\cdot \openC/2$.
Thus, the instance $(A_i,\openC)$ can be recursively partitioned into yet smaller subinstances.
Since at each recursive level, at most half of the points from the previous level remain, the depth of the recursive tree is at most $O(\log n)$.
The total running time is therefore $O(n \log^2 n)$.
\end{proof}

\subsection{Cells and Polyominoes}

Consider an instance $(A,\openC)$ of the \facilityname{} with bounding box $S$ (which is an axis-parallel square of side length $\textrm{poly}(|A|) \cdot \openC$).
We will recursively subdivide $S$ into \emph{cells}, starting with a single cell $S$, and recursively partitioning each cell into four smaller squares, ending when the side length of the cells is at most $\openC / 8$.
This happens after some $L = O(\log n)$ recursive operations due to the partitioning described in Section~\ref{sec:partitioning}.
We call a cell a \emph{level $i$ cell} if it has been obtained after $i-1$ levels of subdivision.
The square $S$ is thus a level $1$ cell.

We consider level $i$ cells and define \emph{polyominoes} to be simple polygons consisting of some level $i$ cells.
Two cells are \emph{neighbouring} if their boundaries share an edge or a corner.
We will be particularly interested in \emph{monominoes}, which consist of a single cell, \emph{dominoes}, which are the union of two cells sharing an edge, \emph{L-trominoes}, which are the union of three pairwise neighbouring cells (i.e., which are in the shape of the letter L), and \emph{square-tetrominoes}, which are a $2\times 2$-square of cells.
A \emph{basic polyomino} is a monomino, a domino, an L-tromino, or a square-tetromino.
See Figure~\ref{fig:polyominoes} for all the basic polyominoes.
Note that any non-basic polyomino contains two non-neighbouring cells.
Polyominoes consisting of level $i$ cells are called level $i$ polyominoes.
We say that a polyomino $\pol$ is \emph{convex} if the intersection of $\pol$ with any horizontal or vertical line has at most one connected component.

\begin{figure}
\centering
\includegraphics[width=.6\textwidth]{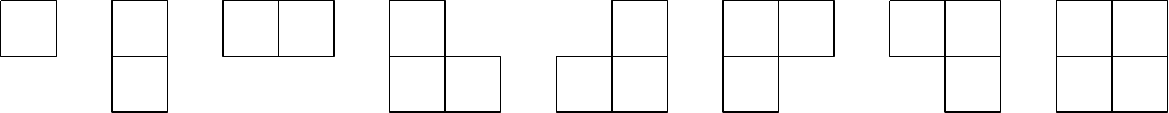}
\caption{Basic polyominoes.}\label{fig:polyominoes}
\end{figure}

Note that each level $i$ monomino, domino, tromino, or tetromino, for $i<L$, contains $4, 8$, $12$, or $16$ cells at level $i+1$, respectively.
We define a \emph{subpolyomino} at level $i+1$ to be a convex polyomino at level $i+1$ that is contained in a basic polyomino at level $i$.
Note that each basic polyomino at level $i+1$ and at level $i$ is also a subpolyomino at level $i+1$.
For all subpolyominoes, see Figure~\ref{fig:sub-polyominoes}.


\begin{figure}
\centering
\includegraphics[width=0.9\textwidth]{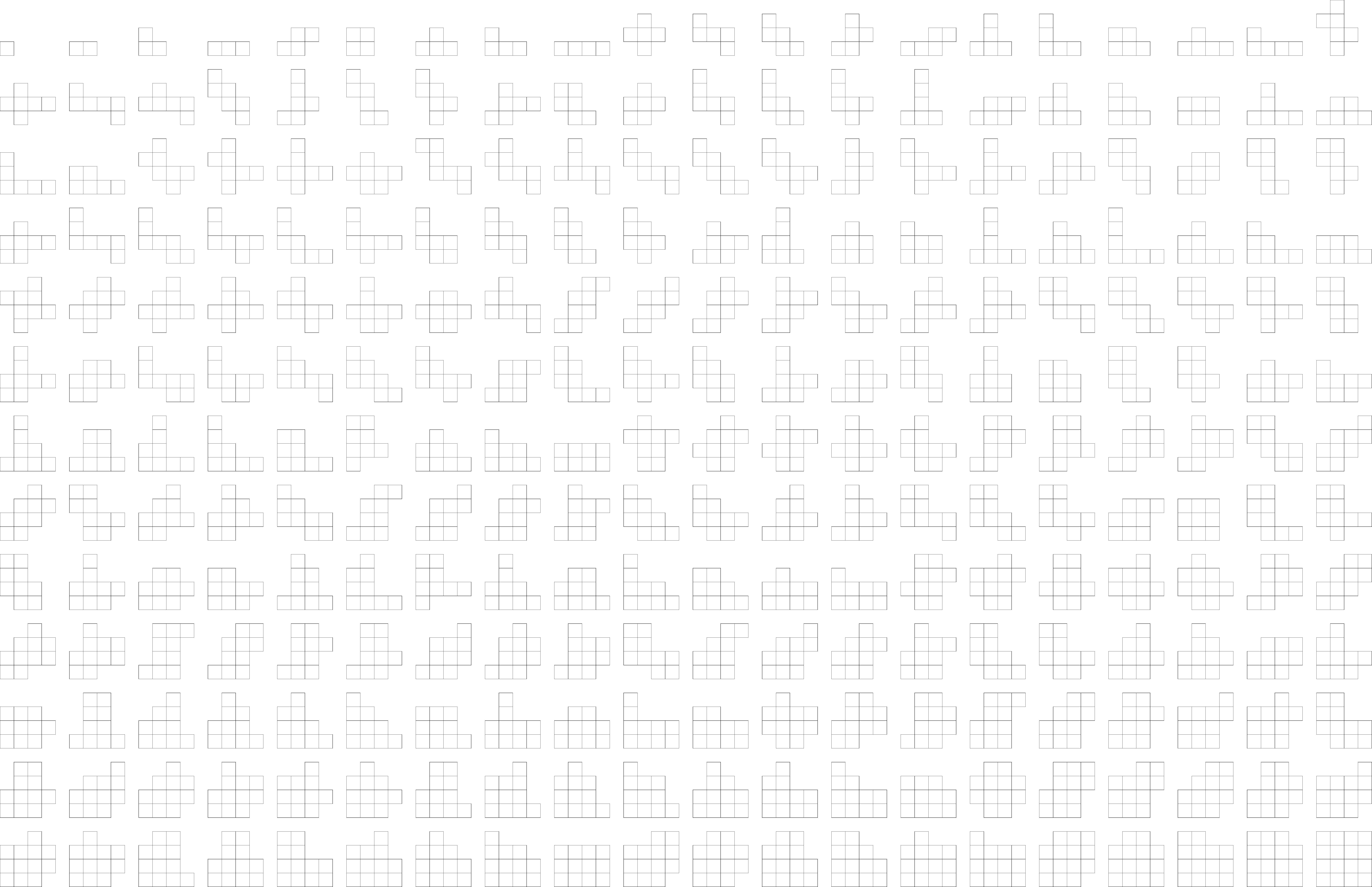}
\caption{The $260$ different subpolyominoes (up to rotations and reflections).}
\label{fig:sub-polyominoes}
\end{figure}

As we want each input point to belong to exactly one cell of a given level, we define a cell to include its right and bottom edge and the bottom-right corner.
The other edges and corners then belong to the neighbouring cells.
For any collection of cells $\pol$ (not necessarily a polyomino), we define $A_\pol\mydef A\cap\pol$ to be the input points in $\pol$.
We say that a polyomino $\pol$ is \emph{empty} if $A_\pol=\emptyset$.
We will consider subproblems of the original \facilityname{} instance $(A,\openC)$, each subproblem corresponding to the input points $A_\pol$ lying in a non-empty basic polyomino $\pol$ of some level.
Note that the number of non-empty basic polyominoes of a given level $i$ is $O(n)$, as each point belongs to a constant number of such polyominoes.
Therefore, the total number of non-empty basic polyominoes is at most $O(n \log n)$.

We compute an optimal partition for each subproblem $(A_\pol, \openC)$, where $\pol$ is a non-empty basic polyomino.
We start with the polyominoes at level $L$.
At level $L$, any two points in the same basic polyomino $\pol$ have distance less than $\openC/2$ and therefore $A_\pol$ is indivisible.
Suppose now that we have already computed the optimal partition for each non-empty basic polyomino at some level $i\leq L$.
As we will see, this makes it possible to compute the optimal partitions for all non-empty subpolyominoes at level $i$.
Since each basic polyomino at level $i-1$ is a subpolyomino at level $i$, we thus also know the optimal partitions of each basic polyomino at level $i-1$.
It follows that the process can be repeated until we reach level $1$, i.e., we have computed an optimal partition of $A$.

To compute and process the polyominoes efficiently, we will use a quadtree construction as described in the next section.

\subsection{Quadtree Construction}

A \emph{quadtree} is a geometric data structure for objects in the Euclidean plane. The root of the quadtree corresponds to a square containing all the input objects, while the children of each node correspond to the four subsquares of the given square. The leaves correspond to subsquares that have small enough side length, or where the input objects have small complexity, e.g., subsquares containing at most some constant number of input objects. See \cite[Chapter 2]{H11} for more information on quadtrees, and for their applications.

In our case, the root of the quadtree corresponds to the level $1$ cell $S$. Each node corresponding to a level $i$ cell $C$ has at most $4$ children, which correspond to level $i+1$ cells contained in $C$. We do not create nodes corresponding to empty cells, i.e., with no points from $A$. The leaves correspond to the highest level cells, i.e., the side length of the leaf cells is at most $\openC / 8$. As there are at most $n$ non-empty cells at each level, the number of nodes of the quadtree is at most $n \cdot L = O(n \log n)$. The quadtree can be constructed in time $O(n \log n)$, as at each of the $L$ levels, we have to compute the subsquares for the $n$ points.

While constructing the quadtree, for each node corresponding to each level $i$ cell, we can compute the nodes corresponding to the eight \emph{neighbouring cells} (if such nodes exist, i.e., the corresponding cells are non-empty). Remembering this information will allow us to construct the polyominoes easier.

From the quadtree, we can construct the set of all non-empty basic polyominoes in time $O(n \log n)$, assigning each basic polyomino $\pol$ to the nodes of the quadtree corresponding to the cells of $\pol$. We do that by considering all nodes of the tree, for each node corresponding to a cell $c$ considering the $21$ basic polyominoes containing $c$, and either constructing a new polyomino, or assigning an existing polyomino to the currently considered node.

\subsection{Finding an Optimal Partition for Each Basic Polyomino}

We now describe an algorithm for finding an optimal partition for each basic polyomino.

\paragraph*{Leaf polyominoes.}
Consider a basic polyomino $\pol$ at level $L$. As $\pol$ consists of level $L$ cells, and the side length of such cells is at most $\openC / 8$, the distance between any two points in $\pol$ is smaller than $\openC / 2$. Therefore, an optimal partition for $(A_\pol,\openC)$ consists of one indivisible set of points $A_\pol$. Therefore, optimal partitions for leaf polyominoes can be computed efficiently.

\paragraph*{At most one big cluster.}
Suppose that we have already computed the optimal partitions for all basic polyominoes at level $i$.
In order to compute the optimal partitions for the basic polyominoes at level $i-1$, we first compute the optimal partitions for all subpolyominoes at level $i$.
This suffices as the basic polyominoes at level $i-1$ are also subpolyominoes at level $i$.
To find an optimal partition for each subpolyomino $\pol$ efficiently, we make use the following property.



\begin{figure}
\includegraphics[scale=0.2]{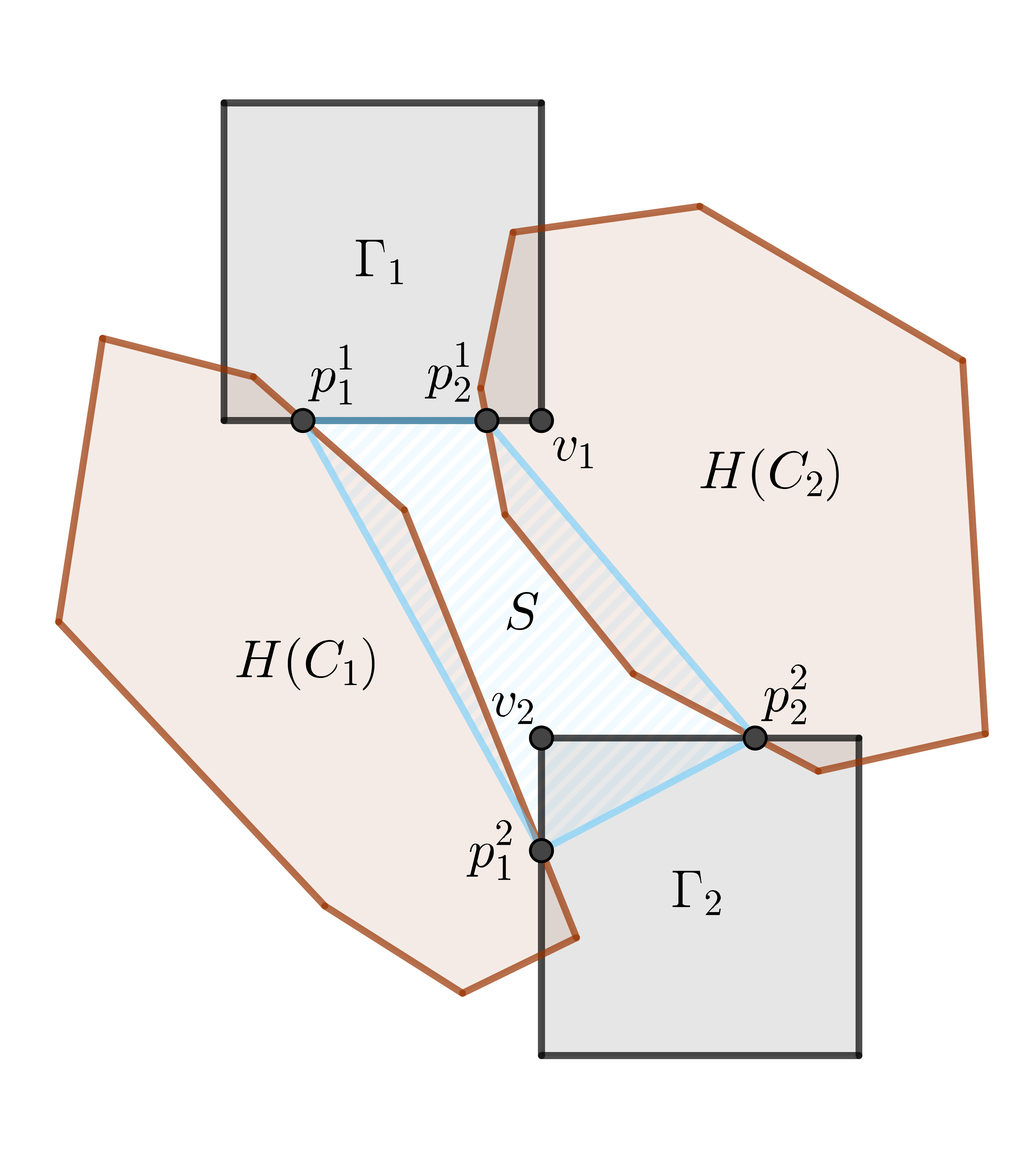}
\centering
\caption{A situation from the proof of Lemma~\ref{lem:only-one-big-cluster}.}
\label{fig:only-one-big-cluster}
\end{figure}

\begin{lemma}\label{lem:only-one-big-cluster}
Let $\pol$ be a non-basic polyomino at some level $i$.
Let  $\cl = \{C_1,\ldots,C_{\ell}\}$ be a maximal optimal partition of $A_\pol$.
For any pair $\Gamma_1, \Gamma_2$ of non-neighbouring cells of $\pol$, there is at most one cluster $C\in\cl$ such that $H(C)$ intersects both $\Gamma_1$ and $\Gamma_2$.
In particular, $\cl$ has at most one cluster $C$ such that $H(C)$ intersects all cells of $\pol$.
\end{lemma}

\begin{proof}
Assume that there are two clusters $C_1, C_2 \in \cl$ such that $H(C_1)$ and $H(C_2)$ both intersect $\Gamma_1$ and $\Gamma_2$.
The situation is depicted in Figure~\ref{fig:only-one-big-cluster}.
Since $\cl$ is optimal, $H(C_1)$ and $H(C_2)$ are disjoint, and it follows that both boundaries $\partial\Gamma_1$ and $\partial\Gamma_2$ intersect both boundaries $\partial H(C_1)$ and $\partial H(C_2)$.
We may then define the point $p^i_j$, for each choice $i,j \in \{1,2\}$, to be the intersection point of $\partial\Gamma_i$ with $\partial H(C_j)$ such that (i) $p^i_1 p^i_2\cap(H(C_1)\cup H(C_2))=\{p^i_1, p^i_2\}$, and (ii) $p^1_j p^2_j\cap (\Gamma_1\cup \Gamma_2)\subseteq \{p^1_j, p^2_j\}$.
Let $S$ be the quadrilateral with vertices at $p^i_j$ for $i,j \in \{1,2\}$, and consider the polygon $C_{12} = H(C_1) \cup H(C_2) \cup S$. We will show that $h(C_{12})$ is not larger than $h(C_1)+h(C_2)$.
As $C_{12}$ contains all points of $C_1 \cup C_2$, this shows that $\cl$ is not optimal or maximal.

In the following, we show the inequality
\begin{align}\label{lem:only-one-big-cluster:ineq}
\|p^1_1 p^1_2\| + \|p^2_1 p^2_2\| \le \|p^1_1 p^2_1\| + \|p^1_2 p^2_2\|.
\end{align}
The desired result then follows as the latter sum is at most the length of the perimeter of $H(C_1)$ and $H(C_2)$ contained in $S$.
Let $\alpha$ be the side length of the cells.
If $\Gamma_1$ and $\Gamma_2$ are in the same row or column of cells, then clearly $\| p^i_1 p^i_2\| \le \alpha\leq \|p^1_j p^2_j\|$, and the inequality~\ref{lem:only-one-big-cluster:ineq} holds.
Otherwise, let $v_1, v_2$ be the corners of $\Gamma_1, \Gamma_2$ minimizing the distance between the cells.
By considering cases of where the points $p^i_j$ are on $\partial\Gamma_i$, one can observe that it always holds that
$$
\|p^i_1p^i_2\| =
\sqrt{\|p^i_1v_i\|^2+\|p^i_2v_i\|^2} \leq
\sqrt{\alpha^2+\|p^i_jv_i\|^2} \leq
\sqrt{\alpha^2+\|p^1_jv_1\|^2+\|p^2_jv_2\|^2} \leq
\|p^1_j p^2_j\|,
$$
and inequality~\eqref{lem:only-one-big-cluster:ineq} follows.
\end{proof}


\paragraph*{Cluster unions.}
Consider a given subpolyomino $\pol$ at some level $i$ for which we want to compute the maximal optimal partition of the input points $A_\pol$.
The overall approach is the following.
We use maximal optimal partitions $\cl_j$ for $A_{\pol_j}$ for various smaller collections $\pol_j\subsetneq\pol$ of level $i$ cells.
We then construct the \emph{merger} of the partitions $\cl_j$, which is the partition of $A_\pol$ we get when we unite the partitions $\cl_j$ and keep merging clusters with overlapping convex hulls.
The merger $\cl$ thus consists of clusters with pairwise disjoint convex hulls, but is in general not optimal.
As we shall see, a maximal optimal partition of $A_\pol$ can be obtained by uniting one subset of the clusters of $\cl$ into one big cluster.
This is the motivation for the following definitions.

Let $\cl=\{C_1,\ldots,C_\ell\}$ be a partition of a set $A$ of points such that $H(C_i) \cap H(C_j) = \emptyset$ for any $i \neq j$.
For any subset $\mathcal S \subseteq \cl$ consider the set $U_{\mathcal S} = \bigcup_{C_i \in \mathcal S} C_i$.
Let $\cl[\mathcal S]$ be the partition consisting of the cluster $U_{\mathcal S}$ and every $C_i\notin\mathcal S$.
Consider the set $\mathcal S^*\subseteq\cl$ minimizing the cost of the partition $\cl[\mathcal S^*]$.
If there is more than one such partition, we are interested in a maximal one.
We say that $\mathcal S^*$ is an \emph{optimal cluster union} for the clustering $\cl$.

\begin{lemma}\label{lem:merger-works}
Let $A = \bigcup_{i \in I} A_i$ be a set of points such that $A = \bigcup_{i \in I} A^\complement_i$, where $A^\complement _i \coloneqq A\setminus A_i$ for $i \in I$.
Let $\cl^*$ be the maximal optimal partition of $A$ and suppose that there is at most one cluster in $\cl^*$ that intersects each set $A_i$.
Let $\cl_i$ be the maximal optimal partition of $A^\complement_i$ and let $\cl$ be the merger of the partitions $\cl_i$.
Let $\mathcal S^*$ be an optimal cluster union for $\cl$.
Then $\cl[\mathcal S^*]$ and $\cl^*$ are the same partition.
\end{lemma}

\begin{proof}
By Lemma~\ref{lem:big-cluster-property}, each cluster of $\cl^*$ is either a cluster of some $\cl_i$, or has non-empty intersection with each set $A_i$.
Since there is at most one cluster in $\cl^*$ of the latter kind, it follows that $\cl^*$ has the form $\{C_1,\ldots,C_k\}$, where each $C_i$, $i\geq 2$, is contained in a cluster of the partition $\cl$.
Each cluster $C$ of $\cl$ is indivisible by Lemma~\ref{lem:sum-of-indivisible-clusters}, so it is contained in some cluster $C_i\in\cl^*$ by Lemma~\ref{lem:cluster-contained}.
For each cluster $C_i$ where $i\geq 2$, there must be a cluster $C\in\cl$ contained in $C_i$, and it follows that $C=C_i$.
Hence, $\cl$ has the form $\{D_1,\ldots,D_\ell,C_2,\ldots,C_k\}$, where $C_1=\bigcup_{i=1}^\ell D_i$.
Therefore, the optimal cluster union for $\cl$ is $\mathcal S^*\mydef\{D_1,\ldots,D_\ell\}$, and $\cl[\mathcal S^*]$ is the partition $\cl^*$.
\end{proof}

\begin{lemma}\label{lem:pol-finding-two-cells}
Let $\pol$ be a non-basic convex polyomino.
There are two cells $\Gamma_1, \Gamma_2$ of $\pol$ with the following properties:
\begin{enumerate}
\item[(i)] $\Gamma_1, \Gamma_2$ are non-neighbouring,
\item[(ii)] $\pol \setminus \Gamma_1$ and $\pol \setminus \Gamma_2$ are convex, and either
\item[(iii.a)] the horizontal distance between $\Gamma_1$ and $\Gamma_2$ is at least as large as the vertical distance, $\Gamma_1$ is leftmost and $\Gamma_2$ is rightmost in $\pol$, or
\item[(iii.b)] the vertical distance between $\Gamma_1$ and $\Gamma_2$ is at least as large as the horizontal distance, $\Gamma_1$ is topmost and $\Gamma_2$ is bottommost in $\pol$.
\end{enumerate}
\end{lemma}

\begin{proof}
Since $\pol$ is non-basic, it has either width of at least $3$ cells or height of at least $3$ cells.
Assume without loss of generality that the width of $\pol$ is at least as large as the height of $\pol$.
We will choose $\Gamma_1$ to be one of the leftmost cells of $\pol$, and $\Gamma_2$ to be one of the rightmost cells of $\pol$.
As we want $\pol \setminus \Gamma_1$ and $\pol \setminus \Gamma_2$ to be convex, $\Gamma_1$ and $\Gamma_2$ have to be topmost or bottommost in their columns.
If there is only one cell in the leftmost (rightmost) column, we take it as $\Gamma_1$ (respectively, $\Gamma_2$), as then clearly $\pol \setminus \Gamma_1$ (respectively, $\pol \setminus \Gamma_2$) remains a convex polyomino.
If there are at least two cells in the column, then, by convexity of $\pol$, at least one of them can be removed without disconnecting the polyomino.
\end{proof}

The following lemma states that we can find optimal cluster unions efficiently.
The proof is in sections \ref{S:findcurve}--\ref{S:findClusterCenter}.

\begin{lemma}\label{lem:finding-cluster-without-center}
Let $\pol$ be a collection of cells and $\Gamma_1, \Gamma_2$ two non-neighbouring cells of $\pol$ such that either
\begin{itemize}
\item[(i)] the horizontal distance between $\Gamma_1$ and $\Gamma_2$ is at least as large as the vertical distance, $\Gamma_1$ is leftmost and $\Gamma_2$ is rightmost in $\pol$, or

\item[(ii)] the vertical distance between $\Gamma_1$ and $\Gamma_2$ is at least as large as the horizontal distance, $\Gamma_1$ is topmost and $\Gamma_2$ is bottommost in $\pol$.
\end{itemize}
Let $\cl$ be a partition of the points $A_\pol$ such that for $i \in \{1,2\}$, $\cl$ restricted to the points of $A_{\pol \setminus \Gamma_i}$ is the maximal optimal partition of $A_{\pol \setminus \Gamma_i}$.
An optimal cluster union for $\cl$ can be found in time $O(n \cdot \polylog n)$, where $n$ is the number of points in $A_\pol$.
\end{lemma}

\paragraph*{Solving non-basic subpolyominoes.}
We can now describe how to find maximal optimal partitions of the points in non-basic subpolyominoes.

\begin{figure}
\includegraphics[width=\textwidth, trim={4cm 1cm 4cm 1cm},clip]{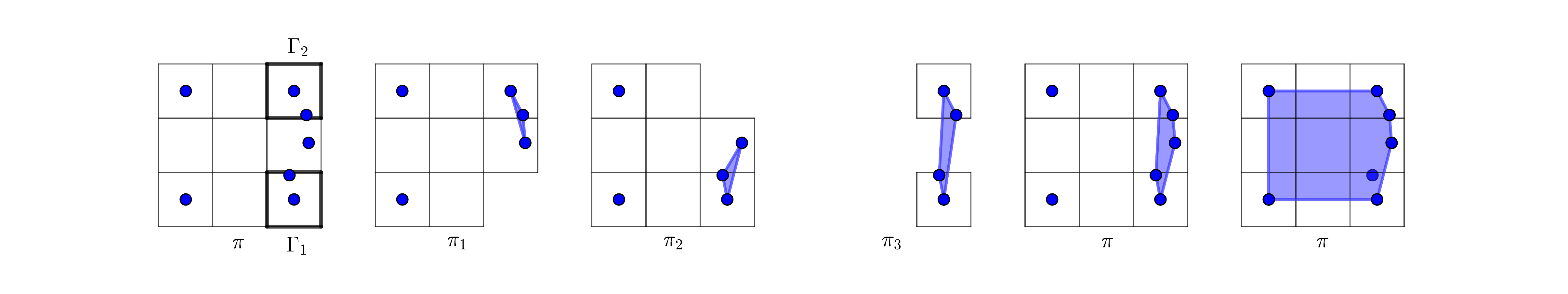}
\caption{A demonstration of how the optimal partition $\cl^*$ of the set $A_\pol$ of points in the non-basic $3\times 3$-cell polyomino $\pi$ is obtained from optimal partitions $\cl^*_1, \cl^*_2, \cl^*_3$ of points $A_{\pol_1}, A_{\pol_2}, A_{\pol_3}$ in smaller collections of cells $\pol_1,\pol_2,\pol_3$, as described the proof of Lemma~\ref{lem:solving-not-2-by-2}.
The cells have width $1$ and the opening cost is $\openC\mydef 5/4$.
The second, third and fourth figure show the optimal clusterings $\cl^*_1, \cl^*_2, \cl^*_3$.
The fifth figure shows the merger $\cl$ of these.
The final figure shows $\cl^*$ (the optimal cluster union consists of all the clusters of $\cl$).}
\label{fig:lem:solving-not-2-by-2}
\end{figure}

\begin{lemma}\label{lem:solving-not-2-by-2}
Let $\pol$ be a non-basic subpolyomino at level $i$.
Given maximal optimal partitions for basic polyominoes at level $i$, we can compute an optimal partition of $A_\pol$ in time $O(n \cdot \polylog n)$, where $n$ is the number of points in $A_\pol$.
\end{lemma}

\begin{proof}
We first find cells $\Gamma_1, \Gamma_2$ of $\pol$ as in Lemma~\ref{lem:pol-finding-two-cells}.
Let $\pol_1 \mydef \pol \setminus \Gamma_1$, $\pol_2 \mydef \pol \setminus \Gamma_2$, and $\pol_3 \mydef \Gamma_1 \cup \Gamma_2$.
As each monomino $\Gamma_i$ is a basic polyomino at level $i$, we know the maximal optimal partition of $A_{\Gamma_i}$ by assumption.
Let $\cl_3$ be the merger of the maximal optimal partitions of $A_{\Gamma_1}$ and $A_{\Gamma_2}$ (which is in fact just the union of the partitions, as the cells are disjoint).
Then, $\Gamma_1 \cup \Gamma_2$ together with the partition $\cl_3$ satisfy the conditions of Lemma~\ref{lem:finding-cluster-without-center}, and we can find an optimal cluster union $\mathcal S_3$ for $\cl_3$ in time $O(n \cdot \polylog n)$.
Define $\cl^*_3\mydef\cl_3[\mathcal S_3]$.
By Lemma~\ref{lem:only-one-big-cluster}, the maximal optimal partition of $A_{\Gamma_1\cup\Gamma_2}$ contains at most one cluster the intersects $\Gamma_1$ and $\Gamma_2$.
Hence, by Lemma~\ref{lem:merger-works}, $\cl^*_3$ is the maximal optimal partition of $A_{\Gamma_1\cup\Gamma_2}$.

See Figure~\ref{fig:lem:solving-not-2-by-2}.
Denote the maximal optimal partition of $A_\pol$ as $\cl^*$.
Consider the sets $A_{\pol_1}, A_{\pol_2}, A_{\pol_3}$, and suppose for now that we know their optimal partitions $\cl^*_1, \cl^*_2, \cl^*_3$.
By Lemma~\ref{lem:big-cluster-property} (taking $A\mydef A_\pol$, $A_1 \mydef A_{\Gamma_1}$, $A_2 \mydef A_{\Gamma_2}$, and $A_3 \mydef A_{\pol \setminus (\Gamma_1 \cup \Gamma_2)}$), we get that a cluster of $\cl^*$ that is not a cluster of any $\cl^*_i$ must intersect both $\Gamma_1$ and $\Gamma_2$.
Due to Lemma~\ref{lem:only-one-big-cluster}, there is at most one such cluster in $\cl^*$.
Let $\cl$ be the merger of the partitions $\cl^*_1, \cl^*_2, \cl^*_3$.
Applying Lemma~\ref{lem:finding-cluster-without-center} for $\pol$, $\Gamma_1$, $\Gamma_2$, and $\cl$, we obtain an optimal cluster union $\mathcal S$ for $\cl$.
By Lemma~\ref{lem:merger-works}, we get that $\cl[\mathcal S]=\cl^*$.

As the polyomino $\pol$ has at most $16$ cells, we need to find optimal partitions for a constant number of subpolyominoes $\pol_i$ before we get down to the basic polyominoes.
That gives the total running time of $O(n \cdot \polylog n)$.
\end{proof}

\paragraph*{Summing up.}
Consider the subpolyominoes $\Pi_i$ at some level $i$.
By Lemma~\ref{lem:solving-not-2-by-2}, the total computation for level $i$ takes time
$$
\sum_{\pol\in\Pi_i} O(|A_\pi|\cdot\polylog |A_\pi|)\leq
\polylog n \cdot\sum_{\pol\in\Pi_i} O(|A_\pi|)=
O(n\cdot \polylog n),
$$
where the equality follows since each point belongs to $O(1)$ level $i$ subpolyominoes.
Due to the preprocessing described in Lemma~\ref{lem:aspectratio}, the number of levels is $O(\log n)$, so the total running time of the algorithm is $O(n \cdot \polylog n)$.
We have thus proven the following theorem:

\begin{theorem}[The \facilityname{}]
\label{thm:facility}
There is an algorithm running in $O(n \polylog(n))$ time that, given any set $A$ of $n$ points in the plane and an opening cost $\openC$, finds a set of $\ell$ closed curves such that each point in $S$ is enclosed by a curve and the total length of the curves plus $\openC \cdot \ell$ is minimized.
\end{theorem}

\subsection{Finding Optimal Cluster Unions}
\label{S:findcurve}
In order to find optimal cluster unions, we first solve a more specialized problem where we require a special point $x_0$ to be contained in the cluster union.
To be more precise, the optimal cluster union $\mathcal S^*$ for the pair $(\cl,x_0)$, where $\cl$ is a partition, is a maximal subset of the clusters of the partition $\cl_{x_0}\mydef\cl\cup\{\{x_0\}\}$ such that $\{x_0\}\in\mathcal S^*$ and the cost of $\cl_{x_0}[\mathcal S^*]$ is minimal.  Note that we use the terms point and vertex interchangeably in the following.

In this section, we prove the following result.

\begin{lemma}\label{lem:finding-cluster-with-center}
Let $\cl=\{C_1,\ldots,C_\ell\}$ be a partition of a given set $A$ of points such that $H(C_i) \cap H(C_j) = \emptyset$ for any $i \neq j$, and let $x_0$ be an arbitrary point. A maximal optimal cluster union for $(\cl,x_0)$ can be found in time $O(n \cdot \polylog n)$, where $n$ is the number of points in $A$.
\end{lemma}

Our first goal is to solve the following more specialized problem.
Given an ``internal'' point $x_0$ and a ``perimeter'' point $p\neq x_0$, and given a set of input points and pre-clustered input points, the goal is to find the optimal cluster containing $x_0$ with an angle-monotone perimeter seen from $x_0$ and with $p$ on its perimeter.
We present an algorithm to solve this problem.
The algorithm can take into account that the cluster must be contained within some delimiting outer boundary.

The idea is to make an angular sweep of a ray from $x_0$, and consider the points in the order the ray sweeps over them. For each point $v$, we calculate the \emph{best} path (see Section~\ref{sec:subprob}) from $p$ to $v$, in terms of both its length and how many clusters' opening costs it may save. In this process, we only store for each vertex its \emph{parent} $\parent(v)$, which is an input point with the property that a best path to $v$ ends with the line segment from $\parent(v)$ to $v$.
Finally, we calculate the parent of $p$, and have thus specified the entire cluster and may generate its convex hull by recursively outputting the parent until we end back at $p$.

\subsubsection{Preliminaries} 
\begin{definition} Any closed simple curve $\sigma$ divides the plane $\mathbb{R}^2$ into two regions: the bounded \emph{interior}, denoted $\inte(\sigma)$, and the unbounded \emph{exterior}, denoted $\exte(\sigma)$. We write $\overline{\inte(\sigma)}=\sigma \cup \inte(\sigma)$ and $\overline{\exte(\sigma)}=\sigma \cup \exte(\sigma)$.
\end{definition}

\begin{definition} A region $R$ of Euclidean space is \emph{star shaped} if there exists a point $s\in R$ such that for all $r\in R$, the line segment $sr$ is contained in $R$. 
We say that $R$ is star shaped \emph{seen from $s$}.\end{definition}

\begin{definition} Given a point $s$, a curve $\pi\colon [0,1]\longrightarrow\mathbb R^2$ is \emph{angle-monotone} with respect to $s$ if
$\measuredangle(\pi(0),s,\pi(x))<\measuredangle(\pi(0),s,\pi(x'))$ for $x<x'$.
Here, $\measuredangle(a,b,c)$ is the counterclockwise angle from the unit vector $\frac{a-b}{\|a-b\|}$ to the unit vector $\frac{c-b}{\|c-b\|}$ on the unit circle.
\end{definition}

When a cluster only contains one point, we call it \emph{trivial}, otherwise, it is \emph{non-trivial}.
If a curve intersects the interior of a cluster, we say that it \emph{dissects} the cluster.
Given a partition $\cl$ of some points, we denote by $V(\cl)$ the set of points, $V(\cl)\mydef \bigcup_{C\in \cl} C$.
%

We denote by $B_r(x)$ the ball of radius $r$ around the point $x$.

\subsubsection{Problem formalization}\label{sec:formal}
We are given: 
\begin{itemize}
	\item special points $x_0$ and $p$,
	\item clusters $\cl\mydef \{C_0,\ldots,C_{l-1}\}$ with costs $\cost(C_i)$,
	\item an \emph{outer limit}, corresponding to the perimeter of an unbounded cluster $C_l$, which the cluster containing $x_0$ is not allowed to intersect.
\end{itemize}
We may assume that no two clusters touch or intersect, and that each cluster is convex. Unclustered points $v$ are treated as trivial clusters $C_i=\{v\}$.

\begin{definition}\label{def:wholecurve}
$Q_{\cl}(p,x_0)$ is the set of all closed simple curves $\sigma$ with $p\in \sigma$, not dissecting any cluster in $\cl$, and such that $\sigma$ is angle-monotone from $x_0$.
	
The cost of a curve $\sigma\in Q_{\cl}(p,x_0)$ is:
\[\lengt(\sigma) + \sum_{C\in \mathcal{C}, C\subseteq \exte(\sigma)} \cost(C).\]
\end{definition}

Note that we sometimes omit the subscript $\cl$ when it is clear from context.

\begin{problem}\label{prob:bestcurve}
	 Given $\cl$, $C_l$, $p$, and $x_0$ as described above, compute $\inf_{\sigma \in Q_{p,x_0}} \cost(\sigma)$.\end{problem}

Note here that even if no outer limit is given as part of the construction, we may take any bounding box around $x_0$ containing all of $V(\cl)$ as an outer limit. 

\begin{lemma}\label{lem:findcurve}
We can compute $\inf_{\sigma \in Q_{p,x_0}} \cost(\sigma)$ and also output $\argmin_{\sigma \in Q_{p,x_0}} \cost(\sigma)$ in time $O(n \log n)$, where $n$ is the total number of vertices  in $V(\cl )$.
\end{lemma}

The rest of this section is dedicated to a proof of the lemma above.

\subsubsection{Reduction to the case where every cluster is non-trivial.}\label{sec:reduction}

We reduce the original $\mathcal{C}$ to an instance $\mathcal{C'}$ where every cluster has a non-empty interior. 

This reduction follows the framework of symbolic perturbation \cite{EdelsbrunnerM90,Yap90,Seidel1998}. We introduce an infinitesimal $\varepsilon$, and replace each vertex $v\in V(\cl )\cup\{x_0\}$ by three vertices $v', v'', v'''$ in an equilateral triangle centered at $v$. Each precluster $C' \in \cl'$ will consist of the set $\{v',v'',v''' \ | \ v\in C\}$ for some $C\in \cl$.
We thus have that every cluster in $\cl '$ has a non-trivial interior.

Finally, we replace every vertex $v\in V(\cl')$ by $v+\varepsilon^2\cdot (\mathcal{N}(0,1),0)$, that is, we perturb each point by a very small random number, such that no three points lie on the same line. Therefore, in the following we can assume all the vertices are in general position.

Note that we may disregard any vertex that does not lie on the convex hull of its cluster.

\subsubsection{Subproblem structure.} \label{sec:subprob}
For any cluster $C$ not containing $x_0$, note that the set of angles  $\{\measuredangle(p,x_0,v) \ | \ v\in H(C)\}$ is either an interval $\left[a,b\right]$ with $0<a<b<2\pi$, or the union of intervals $\left[a,2\pi\right)\cup \left[0,b\right]$, with $0<b<a<2\pi$.
Because of the symbolic perturbation introduced in
Section~\ref{sec:reduction},
these values $a$ and $b$ are realized by unique
vertices on the convex hull of $C$. Denote by $l(C)$ the vertex realizing $a$, and by $r(C)$ the vertex realizing $b$.

\begin{definition}\label{def:legalcurve}
	For $t\notin\{p,x_0\}$, $Q_{\cl}(x_0,p,t)$ is the set of all simple angle-monotone curves $\pi$ from $p$ to $t$, not dissecting any cluster $C$.
	
	Denote by $\cone (p,t)$ the cone with apex $x_0$ through $p$ and $t$. Denote by $\inte(\pi)$ and $\exte(\pi)$ the bounded, and unbounded, region of $\cone(p,t)\setminus \pi$, respectively.

	The cost of a curve $\pi$ is: 
	\[\cost(\pi) = \lengt(\pi)+\sum_{C\in \cl,l(C)\in \overline{\exte(\pi)}}\cost(C).\]
\end{definition}

\begin{observation}Let $\cl$ and $\cl'$ be as in the reduction (Section~\ref{sec:reduction}). Then, since $\varepsilon$ was chosen to be infinitesimal, the difference between
$\inf_{\pi \in \cl} \cost (\pi)$ and $\inf_{\pi' \in \cl'} \cost(\pi')$ is also infinitesimal.
\end{observation}

\begin{definition}
$P_{\cl}(x_0,p,t)\subseteq Q_{\cl}(x_0,p,t)$ is the subset of \emph{polygonal curves} consisting of line segments between points of $V(\cl)$.
\end{definition}

\begin{lemma}
\[\inf_{\pi\in P_{\cl}(x_0,p,t)}\cost(\pi)= \inf_{\pi'\in  Q_{\cl}(x_0,p,t)}\cost(\pi').\]
\end{lemma}
\begin{proof}
For any curve $\pi$ in $Q_{\cl}(x_0,p,t)$, let $V^{\inte }$ denote the set of internal points, and $V^{\exte}$ denote the set of external points. Then, the shortest curve separating $V^{\inte}$ and $V^{\exte}$ will consist of line segments between vertices of $V^{\inte} \cup V^{\exte}$. 
Since $\lengt(\pi') \le \lengt(\pi)$, and since they cover the same clusters, $\cost(\pi') \le \cost(\pi)$.
\end{proof}
  
\begin{definition}
For 
$p\neq t\neq x_0$, 
$P_{\cl}(x_0,p,t)\neq \emptyset$,
let $\pi(t) = \argmin_{\pi\in P_{\cl}(x_0,p,t)}\cost(\pi)$. Since $\pi(t)$ is polygonal, we may write it as the polygonal curve on the vertex set $p=\pi_1,\ldots ,\pi_z = t$.
We say that $\pi_{z-1}$ is the \emph{parent} $\parent(t)$ of $t$.
\end{definition}

The cost of $\pi(t)$ may be rewritten in the following way:
\begin{align*}
\cost & (\pi(t))\\
&=\lengt(\pi_1,\ldots,\pi_{z-1}) + \sum_{C\in\cl, l(C)\in \overline{\exte(\pi)}}\cost(C) + d(\pi_{z-1}, t).
\end{align*}
We denote by $\fix(t)$ the expression $\lengt(\pi_1,\ldots,\parent(t)) + \sum_{C\in\cl, l(C)\in \overline{\exte(\pi)}}\cost(C)$.

When $P_{\cl}(x_0,p,t) = \emptyset$ for some point not contained in any cluster $C\in\cl$, we set $\parent(t)=\bot$ and $\fix(t)=\infty$. When a point lies in the interior of some cluster, $t\in \inte(C)$, or when $t$ lies outside the outer limit, 
we write $\parent(t)=\emptyset$ and $\fix(t)=\infty$.

\subsubsection{Calculating parents via a data structure of intervals.}

For each point $v\in V(\cl)$, we may calculate its angle $\measuredangle(p,x_0,v)$. We may sort the vertices of $V(\cl)$ increasing according to their angle, $p=v_0,v_1,\ldots,v_l=p'$, where $p'$ represents a copy of $p$ with $\measuredangle(p,x_0,p')=2\pi$, for convenience.

We rotate a sweep-ray $\R$ from $x_0$ through $v_i$ for $i=0,1,\ldots,l$. While doing so, we maintain a data structure that encodes for each $r\in \R$ the values $\fix(r)$ and $\parent(r)$.
Here, if $\parent(v_i)=v_j$, then $j<i$.

The data structure maintains a partitioning of $\R$ into maximal intervals $\{I_0,\ldots, I_l\}$ with the same values of $(\fix(t),\parent(t))$ for all $t\in I_i$. 
The first and last intervals are special in the sense that the first, $I_0$, is dedicated to the cluster containing $x_0$, and the last, $I_l$, is dedicated to ensuring that we do not cross the outer limit specified in the problem instance.

We observe that the boundaries between intervals only change continuously, and we thus only store an implicit representation of them. We only consider the sweep-ray at discrete times when either
\begin{itemize}
	\item we reach a vertex $v_i$, corresponding to a vertex insertion in the data structure, or,
	\item an interval shrinks to $\emptyset$, corresponding to the interval being deleted from the data structure.
\end{itemize}

For any vertex $v\in V(\cl)$, it dominates a region of 
$\mathbb{R}^2$ in 
the sense that every point in that region would have $v$ as its parent. We 
call these points the \emph{children} of $v$:
\[\children(v)=\{t\in \mathbb{R}^2 \ | \ v=\parent(t)\}.\]

\begin{lemma}\label{lem:starchildren}
	The region $\children(v)$ is star shaped seen from $v$.
\end{lemma}

\begin{proof}
	Let $t\in \children(v)$, 
	denote by $\pi$ the optimal curve from $p$ to $t$, 
	and assume for contradiction $t'\notin \children(v)$ for some $t'$ on the line segment between $t$ and $v$. 
	Then, $\parent(t')=v'\neq v$. 
	Let $\pi'$ be the optimal path from $p$ to $t'$. Note that $\pi'$ passes through $\parent(t') = v' \ne v$. But then, since both $\pi$ and $\pi'$ pass through $t'$, concatenating $\pi'$ with the line segment $t't$ is better than $\pi$.
\end{proof}

Denote by $P_\R$ the partitioning of the ray $\R$ into intervals consisting of points that have the same parent. For an interval $I\in P_\R$, denote by $\parent(I)$ the shared parent $p(t)$ of all points $t\in I$.

\begin{lemma}
	$|P_\R|\leq 4n+4$ for all angles.
\end{lemma}

\begin{proof}
Enumerate the intervals of $\R$ consecutively by $I_0,I_1,\ldots$, such that $I_0$ has $x_0$ as an endpoint.

For each cluster in $\cl$ with a non-trivial intersection with $\R$, we have an unreachable interval $I$ with $\parent(I) = \emptyset$. A cluster intersecting with $\R$ can contribute with at most $1$ unreachable interval $I$ with $\parent(I)=\bot$, namely, the interval above or below the cluster. For any other interval $I$ we have that $\parent(I) \in \{v_0,\ldots,v_l\}$. It follows from Lemma~\ref{lem:starchildren} that the intervals form a hierarchy
, that is, if $i<j<k$, and $\parent(I_i) = \parent(I_k)\neq \parent(I_j)$, then for any $I_l$ with $\parent(I_l)=\parent(I_j)$, it must hold that $i<l<k$. It follows that there are at most $2 n$ such intervals.
\end{proof}

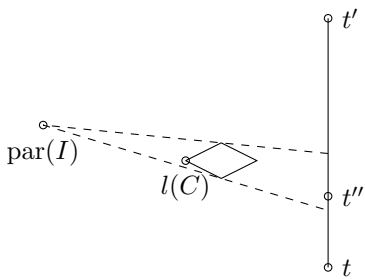
\begin{figure}[ht!]
\begin{center}
\begin{tikzpicture*}{0.3\textwidth}
\begin{scope}[
tinyvertex/.style={
draw,
circle,
minimum size=1mm,
inner sep=0pt,
outer sep=0pt%
},
every label/.append style={
rectangle,
font=\small,
}
]

\node[tinyvertex,label={below:$\parent(I)$}] (par) at (0,2) {};
\node[tinyvertex,label={below:$l(C)$}] (left) at (2,1.5) {};
\node[tinyvertex,label={right:$t''$}] (notin) at (4,1) {};
\node[tinyvertex,label={right:$t$}] (bot) at (4,0) {};
\node[tinyvertex,label={right:$t'$}] (top) at (4,3.5) {};

\draw plot coordinates {(2, 1.5) (2.5, 1.25) (3, 1.5) (2.5, 1.75) (2, 1.5)};

\draw [dashed] plot coordinates {(4,1.6) (0,2) (4,0.8)}; 
  
 \draw plot coordinates {(4,0) (4,3.5)}; 
\end{scope}
\end{tikzpicture*}
\end{center}
\caption{The value of $\fix(t)$ is fixed within an interval.}
\label{fig:fixedcost}
\end{figure}

\begin{lemma}
	For any $I\in P_\R$, for $t,t'\in I$, we have $\fix(t) = \fix(t')$.
\end{lemma}
\begin{proof}
Let $t,t'\in I$. Then, $\parent(t)=\parent(t')=\parent(I)$. But then, the only way $\fix(t)$ could be different from $\fix(t')$ would be if the set of clusters $C\in \cl, l(C)\in \overline{\exte(\parent(t)t)}$ differed from the set of clusters $C\in \cl, l(C)\in \overline{\exte(\parent(t')t')}$.

Assume for contradiction that there exists a cluster $C$ with $l(C) \in \overline{\exte(\parent(I)t')}$ and $l(C) \notin \overline{\exte(\parent(I)t)}$ (see Figure~\ref{fig:fixedcost}). Then, there would exist a point $t''\in \R$ between $t$ and $t'$ with $\parent(t'')\neq \parent(I)$, since any line from $t''$ dissects $C$. But then, $t''\in I$ with $\parent(t'')\neq \parent(I)$; contradiction.
\end{proof}

We denote by $\fix(I)$ the constant value $\fix(t)$ for any $t\in I$.

\paragraph{Determining the boundaries of regions.}
We may categorize boundaries based on how they arise.

Some are lines due to visibility (see Figure~\ref{fig:boundary}, left), while
some are found by equating two costs:
\[\fix(I)+d(\parent(I),t) = \fix(I')+d(\parent(I'),t).\]
One may easily check that this yields a quadratic equation, and thus describes a parabola or a line (see Figure~\ref{fig:boundary}, right). 

\ifstoc
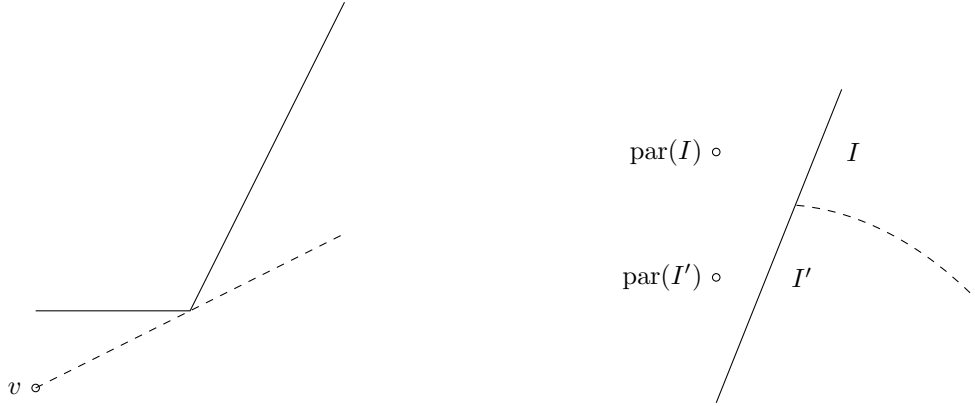
\begin{figure*}[ht!]
\else
\begin{figure}[ht!]
\fi
\begin{center}
\begin{tikzpicture*}{0.28\textwidth}
\begin{scope}[
tinyvertex/.style={
draw,
circle,
minimum size=1mm,
inner sep=0pt,
outer sep=0pt%
},
every label/.append style={
rectangle,
font=\small,
}
]

\node[tinyvertex,label={left:$v$}] (s) at (0,0) {};

\draw plot coordinates {(0, 1) (2, 1) (4, 5)};

\draw [dashed] plot coordinates {(0,0) (4,2)}; 
\end{scope}
\end{tikzpicture*}
\hspace{0.2\textwidth}
\begin{tikzpicture*}{0.28\textwidth}
	\begin{scope}[
		tinyvertex/.style={
			draw,
			circle,
			minimum size=1mm,
			inner sep=0pt,
			outer sep=0pt%
		},
		every label/.append style={
			rectangle,
			font=\small,
		}
		]
		
		\node[tinyvertex,label={left:$\parent(I)$}] (v) at (0,4) {};
		\node[tinyvertex,label={left:$\parent(I')$}] (vp) at (0,2) {};
		\node[label={right:$I'$}] (I) at (0.875,2) {};
		\node[label={right:$I$}] (Ip) at (1.75,4) {};

		\draw plot coordinates {(0,0) (2, 5)};
		
		\draw [dashed, rotate=-5] (1,3.25) parabola (4,2); 
	\end{scope}
\end{tikzpicture*}
\end{center}
\caption{Boundaries between intervals are either a line due to visibility (left) or described by a quadratic equation (right).}
\label{fig:boundary}
\ifstoc
\end{figure*}
\else
\end{figure}
\fi

Hence, we may store the intervals and the functions describing the boundaries between the intervals on the ray. We may build a balanced binary search tree over these intervals, such that each leaf is an interval, and each non-leaf node stores the function describing the boundary between its children. We call this the \emph{succinct representation} of $P_\R$. We note that this succinct representation is enough to determine the parent of a point on the sweep-ray:

\begin{lemma}\label{lem:findparent}
	Given a succinct representation of $P_{\R_{\alpha}}$ at some angle $\alpha$, and assuming no boundaries of intervals intersect between angles $\alpha$ and $\beta$, 
	then for any $t\in \R_{\beta}$ we can find $(\fix (t), \parent(t))$ in $O(\log n)$ time. 
\end{lemma}
\begin{proof}
This can be done by searching the balanced binary search tree of $\R_{\alpha}$. At each node, starting with the root, we may evaluate the function describing the boundary between its children at the angle $\beta$, and recurse into the child whose range contains $t$. Since the tree is a balanced binary search tree over the $O(n)$ intervals, it has height $O(\log n)$, and thus the total time is $O(\log n)$. 
\end{proof}

\paragraph{Initializing the data structure.}

The sweep-ray $\R$ is initialized at angle $0$ (that is, with $p\in \R$). It is initialized with an interval $I_0$ containing $x_0$. For each cluster crossing the line segment between $x_0$ and $p$, we have two consecutive intervals: an unreachable interval $I_{2m-1}$ with $\parent({I_{2m-1}})=\bot$, and an interval corresponding to the inside of the cluster $I_{2m}$ with $\parent(I_{2m})=\emptyset$.
Then, there will be an interval $I_p$ with $\parent(I_p)=p$, and then, for each cluster intersecting the ray after $p$, there will be two consecutive intervals $I_{2m}$ and $I_{2m+1}$ with $\parent(I_{2m})=\emptyset$ and $\parent({I_{2m+1}})=\bot$.
Finally, the last interval will be delimited by the outer limit given as part of the input.

For any pair of consecutive intervals, exactly one of them corresponds to a cluster, and the boundary between them follows the line segment of the convex hull of the cluster.

For all intervals $I$, $\fix(I)$ is initialized to $0$.
 
\subsubsection{Updating the data structure.}
We now describe how we update the interval structure of the sweep-ray.

There are two types of updates. Vertex insertions, and interval annihilations. Whenever we update the data structure, we can foresee when the next interval annihilation will occur.  Thus, when we encounter a point, we know that the topology of $P_\R$ has not changed.
\paragraph{Interval annihilations.}
An interval $I_i$ shrinks to $\emptyset$ if and only if its boundaries to its two neighbors, $I_{i-1}$ and $I_{i+1}$, intersect (see Figure~\ref{fig:annihilation}).

For any triplet of consecutive intervals, $I_{i-1},I_{i},I_{i+1}$, we store the predicted annihilation-angle for $I_i$, that is, when $I_{i-1}$ and $I_{i+1}$ have dominated it. Whenever an interval disappears, only at most two such triplets are affected, namely, we must make a new predicted annihilation-angle for $I_{i-1}$ and for $I_{i+1}$, where they take each other into account. 	

Note that by star shapedness (Lemma~\ref{lem:starchildren}), $\parent(I_{i-1})\neq \parent(I_{i+1})$, and thus we never have to merge intervals.

\ifstoc
\begin{figure*}[ht!]
\else
\begin{figure}[ht!]
\fi
\begin{center}
\begin{tikzpicture*}{0.5\textwidth}
	\begin{scope}[
		tinyvertex/.style={
			draw,
			circle,
			minimum size=1mm,
			inner sep=0pt,
			outer sep=0pt%
		},
		every label/.append style={
			rectangle,
			font=\small,
		}
		]
	
		\node[tinyvertex,label={left:$\parent(I)$}] (v) at (-2,5) {};
		\node[tinyvertex,label={left:$\parent(I')$}] (vp) at (-2,3.2) {};
		\node[tinyvertex,label={left:$l(C)$}] (lC) at (1,1) {};
		\node[label={right:$I'$}] (Ip) at (1.5,2.4) {};
		\node[label={right:$I$}] (I) at (1.5,4) {};
		\node[label={right:$I^{\emptyset}$}] (Ie) at (1.5,1) {};

		\draw [dotted] plot coordinates {(1.5,0) (1.5,5)};
		
		\draw plot coordinates {(2,0) (1,1) (3, 3) (5,1) (4,0)};
		
	\draw[dashed, rotate=-6] (-0.1,3.35) parabola (2.9,2.1);
	\end{scope}
\end{tikzpicture*}
\hspace{0.1\textwidth}
\begin{tikzpicture*}{0.25\textwidth}
\begin{scope}[
tinyvertex/.style={
draw,
circle,
minimum size=1mm,
inner sep=0pt,
outer sep=0pt%
},
every label/.append style={
rectangle,
font=\small,
}
]

\node[label={left:$I_{i}$}] (It) at (2,2.2) {};
\node[label={left:$I_{i+1}$}] (I) at (2,3.5) {};
\node[label={left:$I_{i-1}$}] (Ib) at (2,0.5) {};
	
\draw [dotted] plot coordinates {(2,0) (2,4)};

\draw[dashed] plot [smooth] coordinates {(1,1) (2,1.4) (3,2)};
\draw[dashed] plot [smooth] coordinates {(1,3.5) (2,2.6) (3,2)};
\draw[dashed] plot [smooth]  coordinates {(3,2) (3.6,2.1) (4.2,2.3)};
\end{scope} 
\end{tikzpicture*}

\end{center}
\caption{An interval is deleted because it is dominated by its neighbors.}
\label{fig:annihilation}
\ifstoc
\end{figure*}
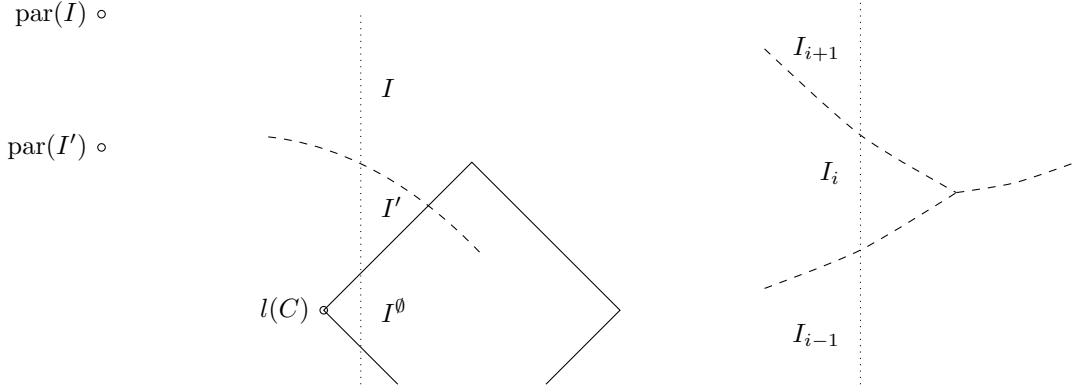
\else
\end{figure}
\fi

\paragraph{Vertex insertions.}
Assume we reach a vertex $t$ and update our sweep-ray such that $t\in \R$.
Due to Lemma~\ref{lem:findparent}, we can determine the interval $I$ containing $t$, and thus, determine $\parent(t)$ and $\fix(t)$.

Now, $t$ gives rise to at most three intervals in $\R$, which we add to the data structure. Recall that every vertex we consider lies on the convex hull of some non-trivial cluster $C$. In the following, denote by $e^-$ and $e^+$ the edges of the perimeter of $C$ incident to $t$ such that $e^-$ comes just before $e^+$ in a clockwise order. We divide our description into four cases, that depend on how the edges of the convex hull incident to $t$ relate to $\R$ (see Figure~\ref{fig:vertexcases}).
\begin{enumerate}
	\item $t=l(C)$,\label{it:left}
	\item $t$ lies on the outer perimeter of $C$, that is, $\exists t'\in C\cap \R$ with $t'$ closer to $x_0$ than $t$,\label{it:top}
	\item $t$ lies on the inner perimeter of $C$, that is, $\exists t'\in C\cap \R$ with $t'$ further from $x_0$ than $t$, \label{it:bot}
	\item $t=r(C)$.\label{it:right}
\end{enumerate}
Since $C$ is convex, these are the only cases.

\ifstoc
\begin{figure*}[ht!]
\else
\begin{figure}[ht!]
\fi
\begin{center}
\begin{tikzpicture*}{0.1\textwidth}
\begin{scope}[
tinyvertex/.style={
draw,
circle,
minimum size=1mm,
inner sep=0pt,
outer sep=0pt%
},
every label/.append style={
rectangle,
font=\small,
}
]
\node[tinyvertex,label={left:$t$}] (v) at (0,0) {};
\draw plot coordinates {(2,-1) (0,0) (2,1)};
\end{scope} 
\end{tikzpicture*}
\hspace{0.1\textwidth}
\begin{tikzpicture*}{0.2\textwidth}
	\begin{scope}[
		tinyvertex/.style={
			draw,
			circle,
			minimum size=1mm,
			inner sep=0pt,
			outer sep=0pt%
		},
		every label/.append style={
			rectangle,
			font=\small,
		}
		]
\node[tinyvertex,label={above:$t$}] (v) at (0,0) {};
\draw plot coordinates {(2,-1) (0,0) (-2,-1)};
	\end{scope} 
\end{tikzpicture*}
\hspace{0.1\textwidth}
\begin{tikzpicture*}{0.2\textwidth}
	\begin{scope}[
		tinyvertex/.style={
			draw,
			circle,
			minimum size=1mm,
			inner sep=0pt,
			outer sep=0pt%
		},
		every label/.append style={
			rectangle,
			font=\small,
		}
		]
		\node[tinyvertex,label={below:$t$}] (v) at (0,0) {};
		\draw plot coordinates {(2,1) (0,0) (-2,1)};
	\end{scope} 
\end{tikzpicture*}\hspace{0.1\textwidth}
\begin{tikzpicture*}{0.1\textwidth}
	\begin{scope}[
		tinyvertex/.style={
			draw,
			circle,
			minimum size=1mm,
			inner sep=0pt,
			outer sep=0pt%
		},
		every label/.append style={
			rectangle,
			font=\small,
		}
		]
		\node[tinyvertex,label={right:$t$}] (v) at (0,0) {};
		\draw plot coordinates {(-2,1) (0,0) (-2,-1)};
	\end{scope} 
\end{tikzpicture*}
\hspace{0.1\textwidth}
\end{center}
\caption{Cases \ref{it:left}, \ref{it:top}, \ref{it:bot}, and \ref{it:right}.}
\label{fig:vertexcases}
\ifstoc
\end{figure*}
\else
\end{figure}
\fi

In Case~\ref{it:left}, a new interval $I^{t,\emptyset}$ is created, with $\parent(I^{t,\emptyset})=\emptyset$ and $\fix(I^{t,\emptyset})=\infty$. The succinct representation of the boundaries of $I^{t,\emptyset}$ are just the straight lines prolonging $e^-$ and $e^+$ (see Figure~\ref{fig:creation}, left and middle).

In Cases~\ref{it:left}, \ref{it:top}, and~\ref{it:bot}, everything that is visible from $\parent(t)$ belongs to $\children(\parent(t))$, and everything else belongs to $\children(t)$ unless it belongs to $I^{t,\emptyset}$ (see Figure~\ref{fig:creation}).

\begin{lemma}\label{lem:firstcases}
	The intervals described above are the correct intervals for Cases \ref{it:left}, \ref{it:top}, and~\ref{it:bot}.
\end{lemma}
\begin{proof}
Because of the symbolic perturbation of Section~\ref{sec:reduction}, we may assume that $t$ does not lie on the border between intervals. That is, there exists some $\delta > 0$ such that all of $B_{\delta}(t)\cap \R$ belongs to the interval $I$ containing $t$. But then, there exists some $\delta'$ with $\delta\geq \delta' >0$, such that $B_{\delta'}(t) \subseteq \children(\parent(t))\cup \children(t)\cup \children(\emptyset)$. 

Then, for any reachable point in $v\in B_{\delta'}(t)$ (i.e., $v\in B_{\delta'}(t)$ such that $P_{\cl}(x_0,p,v)\neq \emptyset$), the optimal path goes through $\parent(t)$ either as the second to last, or as the third to last vertex. Thus, if $\parent(t)$ is visible from $v$, then $\parent(v)=\parent(t)$. Otherwise, necessarily, $\parent(v)=t$.
\end{proof}
\ifstoc
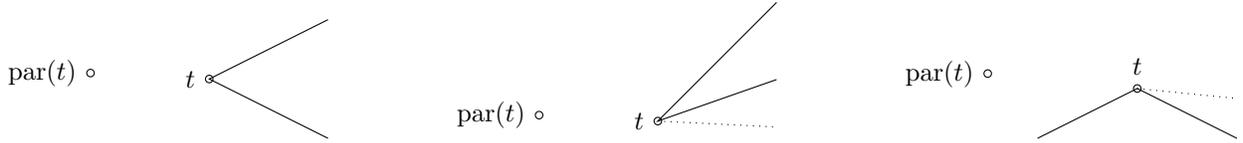
\begin{figure*}[ht!]
\else
\begin{figure}[ht!]
\fi
\begin{center}
\begin{tikzpicture*}{0.25\textwidth}
\begin{scope}[
tinyvertex/.style={
draw,
circle,
minimum size=1mm,
inner sep=0pt,
outer sep=0pt%
},
every label/.append style={
font=\small,
}
]
\node[tinyvertex,label={left:$\parent(t)$}] (pv) at (-2,0.1) {};
\node[tinyvertex,label={left:$t$}] (v) at (0,0) {};
\draw plot coordinates {(2,-1) (0,0) (2,1)};
\end{scope} 
\end{tikzpicture*}
\hspace{0.08\textwidth}
\begin{tikzpicture*}{0.25\textwidth}
	\begin{scope}[
		tinyvertex/.style={
			draw,
			circle,
			minimum size=1mm,
			inner sep=0pt,
			outer sep=0pt%
		},
		every label/.append style={
			font=\small,
		}
		]
\node[tinyvertex,label={left:$\parent(t)$}] (pv) at (-2,0.1) {};
\node[tinyvertex,label={left:$t$}] (v) at (0,0) {};
\draw plot coordinates {(2,2) (0,0) (2,0.7)};
\draw[dotted] plot coordinates {(0,0) (2,-0.1)};
	\end{scope} 
\end{tikzpicture*}
\hspace{0.08\textwidth}
\begin{tikzpicture*}{0.25\textwidth}
	\begin{scope}[
		tinyvertex/.style={
			draw,
			circle,
			minimum size=1mm,
			inner sep=0pt,
			outer sep=0pt%
		},
		every label/.append style={
			rectangle,
			font=\small,
		}
		]
\node[tinyvertex,label={left:$\parent(t)$}] (pv) at (-3,0.3) {};
\node[tinyvertex,label={above:$t$}] (v) at (0,0) {};
\draw plot coordinates {(2,-1) (0,0) (-2,-1)};
\draw[dotted] plot coordinates {(0,0) (2,-0.2)};
	\end{scope} 
\end{tikzpicture*}
\end{center}
\caption{The vertex $t$ gives rise to at most $3$ new intervals.}
\label{fig:creation}
\ifstoc
\end{figure*}
\else
\end{figure}
\fi

Finally, in Case~\ref{it:right}, $t$ lies on the border between the intervals $I,I'$. We may in constant time calculate and compare the two costs and determine whether $\parent(t) = \parent(I)$ or $\parent(t) = \parent(I')$. It follows from the symbolic perturbation in Section~\ref{sec:reduction} that there are no ties; one is strictly better than the other. Given $\parent(t)$, we may calculate $\fix(t)$.  

We now introduce a new interval $\widetilde{I}$ that has $\parent(\widetilde{I})=t$ and $\fix(\widetilde{I})=\fix(t)$. We need to initialize the boundaries of $\widetilde{I}$ correctly: note that an arbitrary number of intervals may be deleted because they are dominated by $\widetilde{I}$, however, they are consecutive intervals, and we may spend constant time on each deletion, since deleted intervals will never resurrect (this follows from the star shapedness property in Lemma~\ref{lem:starchildren}). 

\begin{lemma}
	The interval described above is the correct interval for Case~\ref{it:right}.
\end{lemma}
\begin{proof}
	Assume $\parent(t)=\parent(I)$, which is strictly better than $\parent(I')$. Then, following the same lines as the proof of Lemma~\ref{lem:firstcases}, there exists a small $\delta'>0$ such that $B_{\delta'}(t)\subseteq \children(\parent(t))\cup \children(t)\cup \children(\emptyset)$. Thus, it is necessary to introduce the interval $\widetilde{I}$. It follows from star shapedness (Lemma~\ref{lem:starchildren}) that introducing $\widetilde{I}$ is also sufficient.
\end{proof}
\ifstoc
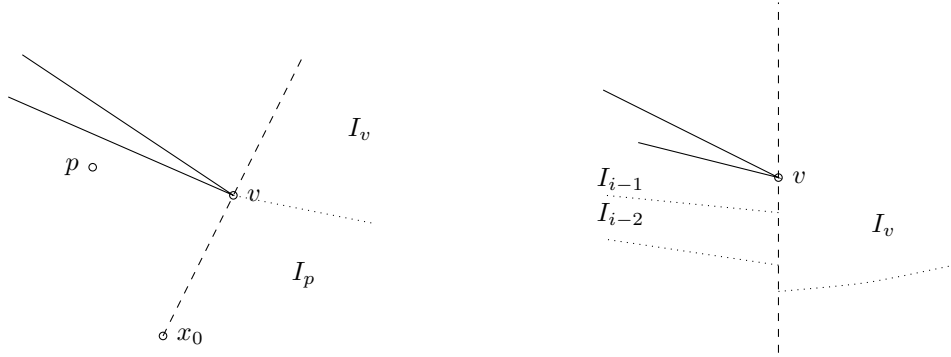
\begin{figure*}[ht!]
\else
\begin{figure}[ht!]
\fi
\begin{center}
\begin{tikzpicture*}{0.3\textwidth}
\begin{scope}[
tinyvertex/.style={
draw,
circle,
minimum size=1mm,
inner sep=0pt,
outer sep=0pt%
},
every label/.append style={
font=\small,
}
]
\node[label={$I_v$}] (cv) at (2.8,2.5) {};
\node[label={$I_p$}] (cp) at (2,0.4) {};
\node[tinyvertex,label={left:$p$}] (p) at (-1,2.4) {};
\node[tinyvertex,label={right:$x_0$}] (x) at (0,0) {};
\node[tinyvertex,label={right:$v$}] (v) at (1,2) {};
\draw plot coordinates {(-2,4) (1,2) (-2.2,3.4)};
\draw [dashed] plot coordinates {(0,0) (2,4)};
\draw [dotted] plot coordinates {(1,2) (3,1.6)};
\end{scope} 
\end{tikzpicture*}
\hspace{0.15\textwidth}
\begin{tikzpicture*}{0.3\textwidth}
	\begin{scope}[
		tinyvertex/.style={
			draw,
			circle,
			minimum size=1mm,
			inner sep=0pt,
			outer sep=0pt%
		},
		every label/.append style={
			font=\small,
		}
		]
	\node[label={$I_{i-1}$}] (i1) at (0.2,1.6) {};
	\node[label={$I_{i-2}$}] (i2) at (0.2,1.2) {};
	\node[label={$I_{v}$}] (i2) at (3.2,1.1) {};
	\node[tinyvertex,label={right:$v$}] (v) at (2,2) {};
	\draw plot coordinates {(0,3) (2,2) (0.4,2.4)};
	\draw [dashed] plot coordinates {(2,0) (2,4)};
	\draw [dotted] plot [smooth] coordinates {(2,1.6) (0,1.8)};
	\draw [dotted] plot [smooth] coordinates {(2,1) (0,1.3)};
	\draw [dotted] plot [smooth] coordinates {(2,0.7) (3,0.8) (4,1)};
\end{scope} 
\end{tikzpicture*}
\end{center}
\caption{The insertion of the vertex $v=r(C)$ may cause the deletion of several intervals.}
\label{fig:case4}
\ifstoc
\end{figure*}
\else
\end{figure}
\fi

What remains is to handle the changes to $\fix$ efficiently and correctly. 
The only change to $\fix$ of any cluster happens in Case~\ref{it:left}: the cost of $C$ needs to be added to $\fix(I)$ for every interval $I$ between $x_0$ and $t$ on $\R$. 

Note that no boundaries change because of this: if $I'$ and $I''$ are intervals that are not separated by $t$, then their $\fix$ either stays the same or is increased by the same amount, that is, $\cost(C)$. The only boundary that could be affected is one incident to $I^{t,\emptyset}$. However, the boundaries of $I^{t,\emptyset}$ are lines determined by the convex hull of $C$. 

Finally, we argue that $\fix$ can be added to all intervals between $x_0$ and $t$ in $O(\log n)$ time: since we store $P_\R$ in a balanced binary search tree, we simply update every vertex of the root path of $I^{t,\emptyset}$ with a lazy delta value that they should push down to their children upon inspection.




\subsection{Proof of Lemma~\ref{lem:finding-cluster-with-center}}
\label{S:findpoint}


In this section, we assume that we are given a polyomino,
a point $x_0$ belonging to an optimal cluster union, and a set of clusters
of points within the polyomino. We will be concerned with Jordan curves that
are not self-intersecting. When we consider a \emph{curve}
we only refer to such a curve. 

We say that a ray $\R$ is \emph{good} if it satisfies the following conditions:
\begin{itemize}
\item it starts at a point $x_0$ of an optimal cluster union, and
\item any input point is at distance at least $\openC/n^2$ from the ray.
\end{itemize}

\subsubsection{Finding a good ray.}
\begin{lemma}
  \label{lem:goodray}
  Given a point $x_0$ that belongs to an optimal cluster union, 
  there exists an $O(n\log n)$ time algorithm that finds a good ray.
\end{lemma}
\begin{proof}
  We know by definition of the point set that $x_0$ is at distance
  at least $\openC/2$ from any other point.
  Thus, we can perform 
  a clockwise scan around $x_0$ and, by the pigeonhole
  principle, there must be at least one ray out of $x_0$
  that is at distance at least $\openC/n^2$ from all the 
  other points. This ray is a good ray.
\end{proof}
%

\subsubsection{Nested curves.}
For any point $p\in \mathbb{R}^2$, denote by $\bestc p$ the optimal curve found in Lemma~\ref{lem:findcurve}.

We want to show that we can find a suitable $p$ such that the convex hull $H(C)$ of the cluster $C$ containing $x_0$ in an optimal clustering is the same curve as $\bestc p$.

To show this, we need some structural lemmas, stating that the curves for different values $p_1,p_2,\ldots$ along a ray are nested nicely around $x_0$ without crossings.

\begin{lemma}
  \label{lem:nestedgeneral}
  \label{lem:nestedray}
  Consider two points $p_1,p_2$ on the same ray from $x_0$, and assume that $d(p_2,x_0) > d(p_1,x_0)$. 
%
  Then, $\bestc{p_1}$ is completely contained in $\overline{\inte(\bestc{p_2})}$.	
%
%
\end{lemma}
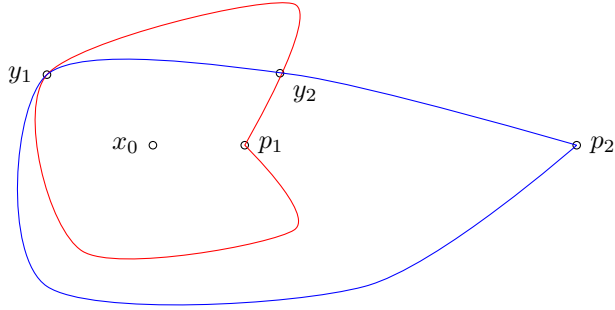
\begin{figure}[ht!]
\begin{center}
\begin{tikzpicture*}{\ifstoc 0.47 \else 0.5 \fi \textwidth}
\begin{scope}[
tinyvertex/.style={
draw,
circle,
minimum size=1mm,
inner sep=0pt,
outer sep=0pt%
},
every label/.append style={
font=\small,
}
]
\node[tinyvertex,label={left:$x_0$}] (x) at (0,0) {};
\node[tinyvertex,label={right:$p_1$}] (p1) at (1.3,0) {};
\node[tinyvertex,label={right:$p_2$}] (p2) at (6,0) {};
\node[tinyvertex,label={left:$y_1$}] (y1) at (-1.5,1) {};
\node[tinyvertex,label={below right:$y_2$}] (y2) at (1.8,1.02) {};

\draw [red] plot [smooth] coordinates {(1.3,0) (2,-1.2) (-1,-1.5) (-1.5,1) (2,2) (1.3,0)};
\draw [blue] plot [smooth] coordinates {(6,0) (3,-2) (-1.5,-2) (-1.5,1) (2,1) (6,0)};

\end{scope} 
\end{tikzpicture*}
\end{center}
\caption{When $p_1$ and $p_2$ are on the same ray (here, aligned with the $x$-axis), their best curves do not cross.}
\label{fig:crossings}
\end{figure}
\begin{proof}
Assume for contradiction that $C_1 = \bestc{p_1}$ intersects the exterior of $C_2 = \bestc{p_2}$ (see Figure~\ref{fig:crossings}). Since $p_1\in C_1$ and $p_1\in \inte(C_2)$, we have by Jordan's theorem that the curves cross at least twice. Consider two crossing points $y_1,y_2$.
Denote by $C_i^{\left|y_2,y_1\right|}$ the part of $C_i\setminus\{y_1,y_2\}$ containing $p_i$ (for $i=1,2$), and let $C_i^{\left|y_1,y_2\right|}$ denote the part not containing $p_i$.

Denote by $\bestc{y_1,y_2} = \argmin_{\pi \in Q_{\cl }(x_0,y_1,y_2)}$ the best curve from $y_1$ to $y_2$ (see Definition~\ref{def:legalcurve}). Since there was a crossing, we must have that $\bestc{y_1,y_2} \neq C_i^{\left|y_1,y_2\right|}$ for some $i\in \{1,2\}$, and thus  $\cost(\bestc{y_1,y_2})<\cost(C_i^{\left|y_1,y_2\right|})$. 

But then, $\cost(\bestc{p_i}) 
> 
\cost(C_i^{\left|y_2,y_1\right|} \circ \bestc{y_1,y_2})$, which is a contradiction since $C_i^{\left|y_2,y_1\right|} \circ \bestc{y_1,y_2}\in Q_{\cl }(p_i,x_0)$.
\end{proof}


  Denote by $\sigma ^{\ast}$ the best curve containing $x_0$, that is, the curve describing the convex hull of the cluster $C_{x_0}$ containing $x_0$ in the optimal clustering.
\begin{lemma}
  \label{lem:nestedopt}
  For any point $p$, $\bestc{p}$ and $\sigma ^\ast$ do not cross.
\end{lemma}
\begin{proof}
  Denote by $p'$ the intersection between $\sigma ^{\ast}$ and the ray from $x_0$ through $p$. It follows from convexity and boundedness of $C_{x_0}$ that $p'$ is uniquely defined. 
  
  Clearly, $\sigma ^{\ast} = \bestc{p'}$. But now, $p$ and $p'$ are points on the same ray from $x_0$, and thus, it follows directly from Lemma~\ref{lem:nestedray} that $\bestc p$ and $\bestc{p'}=\sigma ^{\ast}$ do not cross. In fact, we have that $\sigma ^\ast$ is completely contained in either $\overline{\inte (\bestc p)}$ or $\overline{\exte (\bestc p)}$. 
%
%
%
\end{proof}



\subsubsection{Discretizing the ray.}
Consider a good ray $\R$ with origin~$x_0$. For ease of 
exposition, we rotate the point set and assume that $\R$ is 
parallel to the $x$-axis.
Define the \emph{landmarks} of $\R$ to be
the points of $\R$ at distance $i \cdot \openC/n$ from
$x_0$, for any non-negative integer $0 \le i \le n^7$.
In particular, we call the point
on $\R$ at distance
$i \cdot \openC/n$ from $x_0 $ the $i^{th}$ landmark and denote it by
$x_i$.

Denote by $H_i$ the region of the plane containing all the
points having $x$-coordinate greater than the $x$-coordinate
of $x_{i}$. Note that $H_1 \supset H_2 \supset H_3 \supset
\ldots$.
Denote by $L$ (respectively $R$) the region of the plane that
contains all the points with $y$-coordinate greater
(respectively smaller) than the $y$-coordinate of $x_0$.
Furthermore, we define the $i^{th}$ curve $C_i$ to be the best curve
in $\Reals^2 \setminus H_{i}$ that intersects $\R$ at the $i^{th}$ landmark.

For any curve $C$, let $V^{\inter}(C)$ and $V^{\outside}(C)$
be the set of points in $\inter(C)$ and $\outside(C)$,
respectively.

\begin{lemma}
  \label{lem:depthrecurs1}
  There exists an index $i^*$ satisfying the property 
  $n/2 \ge |V^{\inter}(C_{i^*-1})|, |V^{\outside}(C_{i^*})|$.
\end{lemma}
\begin{proof}
  The proof follows immediately from Lemmas~\ref{lem:nestedray}
  and~\ref{lem:aspectratio}.
  Since the curves are nested, it must be that there is such an index~$i^*$.
\end{proof}

\begin{lemma}
  \label{lem:correctness1}
  Consider the best curve $C^*$.
  We have that $C^*$ is contained in
  \begin{enumerate}
  \item $\overline{\inter(C_{i^*-1})}$
, or
  \item $\overline{\outside(C_{i^*})}$
, or
  \item $\overline{\inter(C_{i^*})} \cap \overline{\outside(C_{i^*-1})}$
.
  \end{enumerate}
\end{lemma}

Define $u$ to be the first input 
point met by a clockwise walk on $C_{i^*}$ 
starting at $x_{i^*}$ and similarly define $v$ to be the first input
point met by a counterclockwise walk on $C_{i^*}$ starting at $x_{i^*}$.
Consider the angle $\measuredangle(u, x_{i^*},v)$ and the bisector of this
angle. Let $s$ 
be the line perpendicular to the bisector that goes through $x_{i^*-1}$. 
We now define the \emph{tube} $T$ to be the set of points 
that are in the interior of $C_{i^*}$ and at distance less than 
$\openC/2$ from $s$. 

We say that a point $(w_x,w_y)$ is \emph{above} $s$ if 
$w_y \ge g(w_x)$, where $g$ is the equation of line $s$.
Define $r$ to be the point of $C_{i^*} \cap R$ 
that is above $s$ and the farthest away from $s$, and 
$\ell$ the point of  $ C_{i^*} \cap L$  that
is above $s$ and the farthest away from $s$.

\begin{lemma}
  \label{lem:technical}
  Consider the best curve $C^*$ that crosses $\R$ in the
  interval $[x_{i^*-1},x_{i^*})$. 
  We have that
  either both $\ell$ and $r$ belong to the tube or 
  there exists a curve with a better cost
  that lies in $\outside(C_{i^*}) \cup C_{i^*}$.
\end{lemma}
\begin{proof}
  Assume that $C^*$ is the best curve overall.
  By Lemma~\ref{lem:nestedray}, $C^*$ lies in
  $\inter(C_{i^*}) \cup C_{i^*}$.
  Consider $C_{i^*}$ and $C^*$. Observe that the distance
  from $x_{i^*}$ to $C^*$ is at most $\openC/n^7$. Therefore,
  the cost of $C_{i^*}$ is at most the cost of $C^*$ plus 
  $2\openC/n^7$.

  Now, observe that if the angle at $x_{i^*}$ is at least
  $\pi$, then the lemma holds immediately as
  the interval $[x_{i^*-1},x_{i^*})$ has size $\openC/n^7$ and
  the tube has size $\openC/2$ and $C^*$ crosses 
  $[x_{i^*-1},x_{i^*})$ and remains in
  $\inter(C_{i^*}) \cup C_{i^*}$.

  Thus, assume that the angle is at most $\pi$, and let
  $\Theta$ denote this angle.
  Observe that since the maximum length of an edge is bounded
  by $n^2 \openC$, we have that the angle has to be such that
  $n^2 \cdot \openC \cdot \sin(\Theta) > \openC/2$ in order for
  $r$ to be outside of the tube. Thus
  $\sin(\Theta) > 1/(2n^2)$.

  We first prove some useful facts regarding $\sin(x)$ and $\cos(x)$.
  First, we have $\sin(x) \leq x$
  for all $x \geq 0$, which holds as the function $f(x) = x - \sin(x)$
  satisfies $f(0) = 0$ and $f'(x) = 1 - \cos(x) \geq 0$.
  Moreover, for all $0 \leq x \leq \pi/2$, we have the property
  that $\cos(x) \leq 1 - x^2/\pi \leq 1 - x^2/4$.
  To see this, consider the function $g(x) = 1 - x^2/\pi - \cos(x)$.
  This function satisfies $g(0) = 0$ and $g'(x) =  -2x/\pi + \sin(x) \geq 0$,
  where the inequality holds since $\sin(x) \geq 2x/\pi$ for all
  $0 \leq x \leq \pi/2$ (since $\sin(x)$ is concave on the interval $[0,\pi/2]$,
  $\sin(0) = 0 = 2\cdot 0/\pi$, and $\sin(\pi/2) = 1 = 2 \cdot (\pi/2)/\pi$).

  Now, since $\sin(\Theta) > 1/(2n^2)$, we know that $\Theta \leq \pi - 1/(2n^2)$.
  Otherwise, we would get a contradiction:
  $1/(2n^2) < \sin(\Theta) = \sin(\pi - \Theta) \leq \pi - \Theta < \pi - (\pi - 1/(2n^2)) = 1/(2n^2)$,
  where the second to last inequality uses the fact that $\sin(x) \leq x$.
  We define $C'$ to be the convex hull of $C_{i^*}$. Since
  the angle is smaller than $\pi$, we have that
  $C'$ does not intersect $x_{i^*}$. We now compute the 
  cost difference between $C'$ and $C_{i^*}$, namely, we
  define $\delta = C_{i^*} - C' > 0$. 
  We have that the gain $\delta$ is at least 
  $\dist(x_{i^*}, r) + \dist(x_{i^*},\ell) - \dist(r,\ell)$.
  We claim that if the angle $\Theta$ is such that 
  $r$ is outside of the tube, then $\delta > 2\openC/n^7$, and so
  we have $C' + 2\openC/n^7 < C' + \delta \le C_{i^*} \le C^* + 2\openC / n^7$, and 
  therefore $C' <  C^*$. This means that $C^*$ is not 
  optimal and so we get a contradiction.

  Thus, we aim at showing that $\delta > 2\openC/n^7$.
  Now, consider the triangle formed by the three points
  $x_{i^*}$, $\ell$, and $r$.  Let $\beta$ denote the angle
  formed by the line segment between the points $x_{i^*}$ and $\ell$,
  along with the line segment between the points $\ell$ and $r$.
  Similarly, let $\gamma$ denote the angle formed by the line segment between
  the points $x_{i^*}$ and $r$, along with the line segment between $r$ and $\ell$.
  Since $\Theta \leq \pi - 1/(2n^2)$, and we have $\Theta + \beta + \gamma = \pi$,
  we must have $\beta + \gamma = \pi - \Theta \geq \pi - (\pi - 1/(2n^2)) = 1/(2n^2)$.
  Thus, we know that $\max\{\beta,\gamma\} \geq 1/(4n^2)$.
  
  For ease of notation, let
  $a = \dist(x_{i^*},\ell)$, $b = \dist(x_{i^*},r)$, and $c = \dist(r,\ell)$.
  Observe that $c = \dist(r,\ell)$ is precisely
  $\dist(x_{i^*},\ell)\cos(\beta) + \dist(x_{i^*},r)\cos(\gamma) = a\cos(\beta) + b\cos(\gamma)$.
  If $\Theta \leq \pi/2$, then we know $\max\{\beta,\gamma\} \geq \pi/4$,
  and hence $c = a\cos(\beta) + b\cos(\gamma) \leq \max\{a+b/\sqrt{2},a/\sqrt{2}+b\}$.
  This yields $\delta \geq a+b-c \geq a+b-\max\{a+b/\sqrt{2},a/\sqrt{2}+b\} \geq (1-1/\sqrt{2})\min\{a,b\}$.
  Hence, consider the case when $\Theta > \pi/2$, which implies
  $0 \leq \beta,\gamma \leq \pi/2$.  Now, since we have $\max\{\beta,\gamma\} \geq 1/(4n^2)$, we obtain
  $c = a\cos(\beta) + b\cos(\gamma) \leq \max\{a + b(1-1/(64n^4)),a(1-1/(64n^4)) + b\}$
  (the inequality follows from $\cos(x) \leq 1 - x^2/4$, which we
  can apply due to the fact that $\beta,\gamma \leq \pi/2$).
  This yields $\delta \geq a+b-c \geq \min\{a,b\}/(64n^4)$.

  In all, as long as
  $\min\{\dist(x_{i^*},r),\dist(x_{i^*},\ell)\}/(64n^4) > 2\openC/n^7$,
  we are done.  This holds as long as we ensure that the minimum distance
  of any point to the ray is sufficiently far (by definition of a good
  ray), as then every input
  point $p$ has the property that $\dist(x_{i^*},p)$ is sufficiently large.


\end{proof}


We define the 
\emph{right tube} as
$T^R = T \cap R \cap \inter(C_{i^*}) \cap \outside(C_{i^*-1})$. 
Consider the points of the right tube and
order them from their distance to the interval
$[x_{i^*-1},x_{i^*})$.
Let $ d^R_1,\ldots,d^R_{|T^R|}$ be the points of
the right tube in this ordering.
Define $S^R_j$ to be the best curve in $\inter(C_{i^*}) \cup C_{i^*}$
that has $d^R_j$ on its boundary, contains $x_0$, and 
crosses $\R$ in the interval $[x_{i^*-1}, x_{i^*})$.


We have the following lemma.
\begin{lemma}
  \label{lem:nestedfinal}
  Consider two integers $i,j$, $i < j$.
  For any points $d^{R}_i,d^{R}_j \in T^{R}$, 
  we have $S^R_i \subseteq S^R_j \cup \inter(S^R_j)$.
  
\end{lemma}
\begin{proof}
  Assume towards a contradiction that the lemma is not correct.
  By definition, $S^R_j$ 
  crosses $\R$ in $[x_{i^*-1},x_{i^*})$; let $p^*$ be the 
  crossing point. 
  By optimality, the segment $\calL$ from $d^R_i$ to $p^*$ 
  lies in~$T^R$.

  Now, observe that since $d^{R}_i$ is closer to 
  $[x_{i^*-1},x_{i^*})$ than $d^{R}_j$, the distance from $d^{R}_i$
  to $\calL$ is
  less than $\openC/2$.
  Thus, including $p^*$
  yields a better curve that lies in
  $C_{i^*} \cup \inter(C_{i^*})$, that has $d^{R}_j$ on its
  boundary, crosses
  $\R$ in $[x_{i^*-1},x_{i^*})$, and
  contains $x_0$, a contradiction.

  
\end{proof}

Define $S^R_j$ to be the best curve in
$(\inter(C_{i^*}) \cap \outside(C_{i^*-1})) \cup C_{i^*-1} \cup C_{i^*}$
starting at $d^R_j$ and
crossing $\R$ in $[x_{i-1},x_i)$.


The lemma below follows immediately by
combining
Lemma~\ref{lem:nestedgeneral} and
Lemma~\ref{lem:nestedfinal} with 
Lemma~\ref{lem:depthrecurs1}.
\begin{lemma}
  \label{lem:depthrecurs}
  There exists an index $j^R$ satisfying the property
  $|V^{\inter}(S_{j^R-1})|, |V^{\outside}(S_{j^R})| \le n/2$.
\end{lemma}


We proceed similarly with the left part of the tube.
Namely, we define 
$T^L = T \cap L \cap \inter(S_{j^R}) \cap \outside(S_{j^R-1})$.
Let $d^L_1,\ldots,d^L_{|T^L|}$ be the points in $T^L$ in 
order of non-decreasing distances to 
the interval $[x_{i^*-1},x_{i^*})$.

We denote by $S^L_j$ the best
curve starting at $d^L_j$ and staying in 
$\inter(S_{j^R}) \cap \outside(S_{j^R-1})$.
Mimicking the proof of Lemma~\ref{lem:nestedfinal}, we
have the following.
\begin{lemma}
  \label{lem:nestedfinal2}
  Consider two integers $i,j$, $i < j$.
  For any points $d^{L}_i,d^{L}_j \in T^{L}$, 
  we have $S^L_i \subseteq S^L_j \cup \inter(S^L_j)$.
  Moreover, there exists an index $j^L$ such that
  $|V^{\inter}(S_{j^L-1})|, |V^{\outside}(S_{j^L})| \le n/2$.
  
\end{lemma}

Before concluding the proof, we define $d^L_{j^L},d^L_{j^L-1},d^R_{j^R-1},d^R_{j^R}$
to be the \emph{dominating} points.

\begin{lemma}
  \label{lem:correctrecurs}
  Consider the best curve $C^*$.
  We have that $C^*$ is contained in
  \begin{enumerate}
  \item $\inter(S_{j^L-1}) \cup S_{j^L-1}$, or
  \item $\outside(S_{j^L}) \cup S_{j^L}$.
  \end{enumerate}
\end{lemma}

\begin{proof}
  To show the lemma we need to prove that $C^*$ does not intersect
  $\inte(S_{j^L}) \cap \exte(S_{j^L-1})$. 
  Observe that 
  $\exte(C_{i^*}) \cup C_{i^*} \subseteq \exte(S_{j^R}) \cup S_{j^R} \subseteq \exte(S_{j^L}) \cup S_{j^L}$
  and 
  $\inte(C_{i^*-1}) \cup C_{i^*-1} \subseteq \inte(S_{j^R-1}) \cup S_{j^R-1} \subseteq \inte(S_{j^L-1}) \cup S_{j^L-1}$.
  Thus, for $C^*$ to intersect
  $\inte(S_{j^L}) \cap \exte(S_{j^L-1})$ we need that $C^*$ intersects
  $\exte(C_{i^*-1}) \cap \inte(C_{i^*})$ and 
  $\exte(S_{j^R-1}) \cap \inte(S_{j^R})$.
  Hence, by Lemma~\ref{lem:nestedopt}, for $C^*$ to intersect 
  $\exte(C_{i^*-1}) \cap \inte(C_{i^*})$, we need that $C^*$ lies in 
  $\exte(C_{i^*-1}) \cap \inte(C_{i^*}) \cup C_{i^*} \cup C_{i^*-1}$.
  Applying Lemma~\ref{lem:nestedopt} and the definition of $S_{j^R-1}, S_{j^R}$, 
  for $C^*$ to intersect $\inte(S_{j^R}) \cap \exte(S_{j^R-1})$, we have that
  $C^*$ lies in $\inte(S_{j^R}) \cap \exte(S_{j^R-1}) \cup S_{j^R} \cup S_{j^R-1}$.

  Finally, if $C^*$ intersects
  $\inte(S_{j^L}) \cap \exte(S_{j^L-1})$, then Lemma~\ref{lem:technical} applies 
  and we have that
  $C^*$ has to intersect a point in the tube. Now observe that the
  only points in the tube in
  $\inter(S_{j^R}) \cup S_{j^R} \cup \outside(S_{j^R-1}) \cup
  S_{j^R-1}$ are the dominating points. Hence, $C^*$ must be
  one of $S_{j^R},S_{j^R-1},S_{j^L},S_{j^L-1}$, and the lemma
  follows.
\end{proof}

\begin{proof}[Proof of Lemma~\ref{lem:finding-cluster-with-center}]
  We can apply the algorithm defined in Lemma~\ref{lem:goodray}
  to obtain a good ray $\R$ with origin $p$. 

  We can also apply a binary search on the landmarks
  to find $i^*$. By Lemma~\ref{lem:aspectratio}, the
  number of landmarks is bounded by $O(n^9)$ -- observe that
  the set of landmarks does not have to be computed explicitly.
  In particular, we apply a binary search procedure that makes $O(\log n)$ calls
  to the algorithm of Section~\ref{S:findcurve} (Lemma~\ref{lem:findcurve}) to find $i^*$.

  We compute $d^R_1,\ldots,d^R_{|T^R|}$ and apply a binary search procedure making
  $O(\log n)$ calls to the algorithm to find $j^R$.
  Proceed similarly to find $d^L_{j^L},d^L_{j^L-1}$, and
  proceed recursively on $S^L_{j^L} \cup \outside(S^L_{j^L})$ and
  $S^L_{j^L-1} \cup \inter(S^L_{j^L-1})$.

  By Lemma~\ref{lem:depthrecurs}, the depth of the recursion is at most $O(\log n)$ 
  and so, since each vertex appears in at most 2 regions of the plane at each 
  recursive call, each vertex contributes to the running time of 
  the algorithm of Section~\ref{S:findcurve} (Lemma~\ref{lem:findcurve}) at most $O(\log n)$ times.
  It follows that the overall complexity is at most $O(n \cdot \log^3 n)$.
  The correctness follows from Lemma~\ref{lem:correctrecurs}.
\end{proof}



\subsection{Proof of Lemma~\ref{lem:finding-cluster-without-center}}
\label{S:findClusterCenter}

Given a weighted set of points $P$, where for each $p \in P$ we have a weight $w_p$,
an \emph{approximate centerpoint} $c$ is a point in the plane such that, for an arbitrary line
containing the point, each of the open halfplanes bounded by the line contains a weight of at least
$\frac{\sum_{p \in P}w_p}{5}$.
Given a tube $T$ (i.e., quadrilateral), let $L_{\max}$ denote its perimeter, and let $L_{\min} \ge \frac{999}{1000} L_{\max}$
denote the 
assumed minimum perimeter of an optimal cluster union within $T$.
Consider the optimal clusters
that are completely contained within the tube, which we denote by $T_1,\ldots,T_m$.  For each
such cluster $T_i$, we choose an arbitrary point $p_i \in T_i$, and assign it a weight of
$w_{p_i} \mydef h(T_i) + \openC$.  Let $W \mydef \sum_{i=1}^m w_{p_i}$.  Moreover,
for any set of points $S \subseteq P$, we let $w(S) \mydef \sum_{p \in S}w_p$.
Since we assume the existence of an optimal cluster union, we assume that $W \geq L_{\min} + \openC$.

In the following, we seek to find an approximate centerpoint efficiently.

\begin{lemma}\label{lem:heavyline}
Let $\ell$ be a line such that the total weight of points on the line
is at least $\frac{L_{\min}+\openC}{100}$. Then there exists a point $p^*$ on $\ell$
that belongs to an optimal cluster union that we can find in $O(n \log n)$ time.
\end{lemma}
\begin{proof}
We let $W_{\ell}$ denote the total weight of points that lie on the line $\ell$,
so that $W_{\ell} \geq \frac{L_{\min}+\openC}{100}$.
We first sort the points on the line by ascending $x$-coordinate.  Then, we can find a weighted
point $p^*$ on the line in this sorted ordering such that the total weight of points
(including the point $p^*$) that appear earlier in the sorted ordering is at least $\frac{W_{\ell}}{2}$,
and similarly the total weight of points appearing later in the sorted ordering is at least $\frac{W_{\ell}}{2}$.
This can essentially be done by computing the prefix sum of the sorted points, and then performing
a binary search.  In all, this process takes $O(n \log n)$ time.

We claim that $p^*$ belongs to an optimal cluster union.  In particular, suppose towards a contradiction
that this is not the case.  For ease of notation, let $P_1$ be the points appearing earlier
than (and including) $p^*$ in the ordering, and similarly define $P_2$ to be those that appear
later (including $p^*$).  Then it must be the case that either all of $P_1$ does not lie in an
optimal cluster union, or all of $P_2$ does not lie in an optimal cluster union.  Otherwise, if there is some point
$p_1 \in P_1$ in an optimal cluster union and some point $p_2 \in P_2$ in an optimal cluster union,
then by convexity the point $p^*$ must also belong to an optimal cluster union.
Hence, the total weight of points outside of an optimal cluster union must be at least
$\frac{W_{\ell}}{2} \geq \frac{L_{\min}+\openC}{200}$.  Now, an optimal cluster union has weight
at least $L_{\min} + \openC$, and hence the total cost of the optimal partition
in the tube $T$ is at least $L_{\min} + \openC + \frac{W_{\ell}}{2} \geq (L_{\min} + \openC)(1 + \frac{1}{200})$.
Since we have that $L_{\min} \geq \frac{999}{1000}L_{\max}$, we have that this quantity is
at least $(1 + \frac{1}{200})(\frac{999}{1000}L_{\max} + \openC) > L_{\max} + \openC$.
This yields a contradiction however, since the supposed optimal partition is not
optimal.
\end{proof}

In order to find an approximate centerpoint, we make use of the notion of weighted ham-sandwich
cuts as defined in~\cite{BL04}.  It is helpful to define the following notation, as developed in~\cite{BL04}.
For a line $\ell$, we denote by $\ell^+$ and $\ell^-$ the two open halfplanes
bounded by $\ell$.  We say that a line $\ell$ \emph{bisects} a weighted
set of points $S$ if $|w(S \cap \ell^+) - w(S \cap \ell^-)| \leq |w(S \cap \ell)|$
(note that, in our setting, $w(S \cap \ell)$ is always non-negative, and hence
$|w(S \cap \ell)| = w(S \cap \ell)$).
\begin{definition}[Weighted Ham-Sandwich Cut~\cite{BL04}]
Let $A,B$ be two sets of weighted points, where
for each point $p \in A \cup B$, we are given a weight $w_p$.  A \emph{weighted
ham-sandwich cut} for $A$ and $B$ is a line that bisects both weighted
sets $A$ and $B$ simultaneously.
\end{definition}

We make use of the following result.
\begin{lemma}[Theorem 1, \cite{BL04}]\label{lem:ham}
Given two sets $A,B \subseteq \mathbb{R}^2$ of weighted points,
with $|A| + |B| = n$, a weighted ham-sandwich cut can be computed
in $O(n \log n)$ time.
\end{lemma}

Using Lemma~\ref{lem:ham}, we show how to obtain 
a point in an optimal cluster union. 

\begin{lemma}\label{lem:centerpoint}
Consider a tube $T$.  There exists an algorithm running in time $O(n \log n)$ that finds a point
belonging to an optimal cluster union contained in $T$, assuming it exists.
\end{lemma}
\begin{proof}
Consider a vertical line $\ell_v$ defined by $x=u$ cutting through the tube $T$
with the property that the total weight of points in $T$ with an $x$-coordinate of at
most $u$ is at least $\frac{W}{2}$, and the total weight of points in $T$ with an $x$-coordinate
of at least $u$ is at least $\frac{W}{2}$.  Note that such a line can be found in
$O(n \log n)$ time by sorting the weighted points in the tube in ascending order of their
$x$-coordinate, computing the prefix sum of the ordered points, and performing a binary search.
By Lemma~\ref{lem:heavyline}, we can assume that the total weight of points on
any such line $\ell_v$ that belong to the tube is at most $\frac{L_{\min} + \openC}{100} \leq \frac{W}{100}$
(as otherwise, we have already found a point in an optimal cluster union).

For convenience, let $P_1$ be the set of weighted points with an $x$-coordinate
strictly less than $u$, and let $P_2$ be the set of weighted points with an $x$-coordinate
strictly more than $u$.  We can apply Lemma~\ref{lem:ham} to the two sets
$P_1,P_2$ to get a weighted ham-sandwich cut in time $O(n \log n)$ (note that
$|P_1| + |P_2| \leq n$).  We denote by $\ell_s$ the weighted ham-sandwich
cut.  Observe that if the total weight of points on the line $\ell_s$ is
at least $\frac{L_{\min}+\openC}{100}$, then by Lemma~\ref{lem:heavyline},
we are done since we can find a point in an optimal cluster union lying on $\ell_s$
in $O(n \log n)$ time.  Note that $\ell_s$ and $\ell_v$ intersect, and hence
partition the plane into four quadrants.  We claim that their point of
intersection, denoted by $p^*$, is an approximate centerpoint.

From now on we assume that the total weight of points lying on $\ell_v$
is at most $\frac{W}{100}$, and similarly the total weight of points on $\ell_s$
is at most $\frac{L_{\min}+\openC}{100} \leq \frac{W}{100}$.  Consider
any line that contains the point $p^*$, and consider the two open halfplanes
induced by such a line.  We argue that each such open halfplane has weight
at least $\frac{W}{5}$, which gives the lemma.

To this end, observe that $w(P_1),w(P_2) \geq \frac{W}{2} - \frac{W}{100}$.
This holds since $\ell_v$ was chosen to ensure that the total weight lying to
the left (and including) $u$ is at least $\frac{W}{2}$.  Since
the weight on line $\ell_v$ is at most $\frac{W}{100}$,
the total weight of points with an $x$-coordinate strictly less than $u$
(which is precisely $w(P_1)$)
must be at least $\frac{W}{2} - \frac{W}{100}$.  A symmetric
argument holds for $w(P_2)$.

Since $\ell_s$ is a weighted ham-sandwich cut,
we have $|w(P_1 \cap \ell_s^+) - w(P_1 \cap \ell_s^-)| \leq w(P_1 \cap \ell_s) \leq \frac{W}{100}$,
and similarly for the set $P_2$.
Consider the open quadrant to the left of $\ell_v$ and above $\ell_s$ given by
$\ell_s^+ \cap \ell_v^-$.  The total weight of points in this open quadrant is
given by $w(P_1 \cap \ell_s^+)$, which we know is at least $w(P_1 \cap \ell_s^-) - \frac{W}{100}$.
On the other hand, we also know $w(P_1 \cap \ell_s^-) + w(P_1 \cap \ell_s^+) \geq \frac{W}{2} - \frac{W}{100}$.
Hence, we get
$w(P_1 \cap \ell_s^+) \geq w(P_1 \cap \ell_s^-) - \frac{W}{100} \geq \frac{W}{2} - \frac{W}{100} - w(P_1 \cap \ell_s^+)$,
implying that $w(P_1 \cap \ell_s^+) \geq \frac{W}{4} - \frac{W}{100} \geq \frac{W}{5}$.
We can symmetrically show that each of the four open quadrants contains a weight of at least~$\frac{W}{5}$.

Finally, the two open halfplanes induced by an arbitrary line containing $p^*$
must fully contain an open quadrant that lies to the left of $\ell_v$, and the other open halfplane
must fully contain an open quadrant that lies to the right of $\ell_v$.
Hence, the weight of each such open halfplane is at least~$\frac{W}{5}$ and
$p^*$ is an approximate centerpoint.

It is now enough to show that the approximate centerpoint $p^*$ lies in an optimal cluster union. As an optimal cluster union is convex,
for any point $p$ that is outside of an optimal cluster union, we can draw a line $\ell_p$ that does not intersect it. As the weight
outside of an optimal cluster union is smaller than $\frac{W}{5}$, on one side of $\ell_p$ we have clusters of weight smaller than $\frac{W}{5}$
and $p$ is not an approximate centerpoint. Therefore, all approximate centerpoints are contained in an optimal cluster union.
\end{proof}

We are now able to prove Lemma~\ref{lem:finding-cluster-without-center}.

\begin{proof}[Proof of Lemma~\ref{lem:finding-cluster-without-center}]
Assume without loss of generality that we are in case (i), and $\Gamma_1$ is at least as high as $\Gamma_2$. 
As for $i \in \{1,2\}$, $\cl$ restricted to the points of $A_{\pol \setminus \Gamma_i}$ is a maximal optimal partition of $A_{\pol \setminus \Gamma_i}$, then for an optimal cluster union $B$ for $\cl$, $H(B)$ has to intersect both $\Gamma_1$ and $\Gamma_2$.

Consider a vertical line segment $I$ constructed as follows. The horizontal distance of $I$ to $\Gamma_1$ is equal to the horizontal distance of $I$ to $\Gamma_2$, the top endpoint of $I$ lies on a line $\ell_t$ connecting the top-right vertex of $\Gamma_1$ with the top-right vertex of $\Gamma_2$, and bottom endpoint of $I$ lies on a line $\ell_b$ connecting the bottom-left vertex of $\Gamma_1$ with the bottom-left vertex of $\Gamma_2$. Then, as $H(B)$ intersects both $\Gamma_1$ and $\Gamma_2$, and $H(B)$ is convex, it must intersect $I$. Let $L$ be the side length of the cells $\Gamma_1$ and $\Gamma_2$. Then, as the horizontal distance between $\Gamma_1$ and $\Gamma_2$ is at least as large as the vertical distance, the length of $I$ is at most $2L$. Similarly, we place vertical intervals $I_\ell$ and $I_r$ as follows. $I_\ell$ contains the left edge of $\Gamma_1$, has its top endpoint at the line connecting the bottom-left corner of $\Gamma_2$ and the top-right corner of $\Gamma_1$, and its bottom endpoint at the line connecting the bottom-left corner of $\Gamma_2$ and the bottom-left corner of $\Gamma_1$. $I_r$ contains the right edge of $\Gamma_2$, has its top endpoint at the line connecting the top-right corner of $\Gamma_1$ and the top-right corner of $\Gamma_2$, and its bottom endpoint at the line connecting the top-right corner of $\Gamma_1$ and the bottom-left corner of $\Gamma_2$. The lengths of $I_\ell$ and $I_r$ are upper bounded by $4L$.

Let $\alpha=100$ be a constant. We place $2 \alpha$ equidistant points $p_1, \ldots, p_\alpha$ on $I$, the first one at the top endpoint and the last one at the bottom endpoint of $I$. Similarly, we place $4 \alpha$ equidistant points at $I_\ell$ and $I_r$, denoted by $p^\ell_j$ and $p^r_j$, respectively. Then, the distance between two consecutive points at each of the segments is at most $L/\alpha$.

We will show that if an optimal cluster union $B$ exists, then either $H(B)$ contains in its interior at least one of the points $p_i$, or it is contained in a long and thin quadrilateral spanned by some four points $p^\ell_j, p^\ell_{j'}, p^r_k, p^r_{k'}$, where $|p^\ell_j  p^\ell_{j'}|, |p^r_k p^r_{k'}| \le 8 L/\alpha$.

Assume that an optimal cluster union $B$ exists, and $H(B)$ does not contain any point $p_i$. Then, $H(B) \cap I$ lies between some two points $p_i,p_{i+1}$. As $H(B)$ is convex, contains a point of both $\Gamma_1$ and $\Gamma_2$, and the horizontal distance from $\Gamma_1$ and $\Gamma_2$ to $I$ is at most $L/2$, the intersection of any vertical line with $H(B)$ is a line segment of width at most $3|p_i p_{i+1}| \le 3 L/\alpha$. Consider a line segment $J$ connecting a leftmost point of $H(B)$ with a rightmost point of $H(B)$. Due to convexity of $H(B)$, this line segment is contained in $H(B)$. Therefore, $H(B)$ is contained in a tube around $J$, of height $6 L/\alpha$. We can extend this tube to a tube of height at most $8 L/\alpha$, with endpoints at the points of $I_\ell$ and~$I_r$.

We start by running the algorithm from Lemma \ref{lem:finding-cluster-with-center} for finding an optimal cluster union for $(\cl,p_i)$ for $i \in \{1,\ldots,\alpha \}$. If an optimal cluster union exists and contains one of the points $p_i$, it will be found by this procedure. Then, we consider the set of $O(\alpha^2)$ tubes $\mathcal{T}$, each tube defined by a pair of points $p^\ell_j, p^r_k$, and of height at most $8 L/\alpha$. For each such tube $T \in \mathcal{T}$, we will find an optimal cluster union contained in $T$, if such a cluster exists. For this, we construct $\alpha$ equidistant vertical lines within $\Gamma_1$, and also $\alpha$ equidistant vertical lines within $\Gamma_2$, and for each such pair of lines we construct a sub-tube $T' \subseteq T$, by cutting of the corresponding left and right parts of $T$. Let $\mathcal{T'}$ be the resulting collection of $O(\alpha^4)$ sub-tubes. For each $T' \in \mathcal{T'}$, we will want to find an optimal cluster union of $A_\pol$, assuming that the cluster is contained in $T'$ and not contained in any shorter tube.

Consider a tube $T' \in \mathcal{T'}$, and assume that an optimal cluster union of $A_\pol$ is contained in $T'$ and not contained in any shorter tube. Let $p(T')$ be the perimeter of $T'$. Clearly, the perimeter of an optimal cluster union is at most $p(T')$. Also, as an optimal cluster union stretches nearly all the way from the left of the tube to the right of the tube, and the height of $T$ is at most $8 L/\alpha$, we have that the perimeter of an optimal cluster union is at least $p(T')-20 L/\alpha$. Additionally, we know that $p(T') \ge 2\alpha$. Therefore, by using the reasoning from Lemma~\ref{lem:centerpoint}, we can find a point in the center of an optimal cluster union. Then, it is enough to run the algorithm from Lemma~\ref{lem:finding-cluster-with-center} for this point.

This algorithm uses a constant number of invocations of the algorithm from Lemma \ref{lem:finding-cluster-with-center}, therefore its running time is $O(n \cdot \polylog n)$.
\end{proof}

\section{The $k$-Cluster Fencing Problem}
In the \fixedkname, we are given a set of $n$ points $S$ in the plane and an integer
$k$ as an input instance $I$, which we denote by the pair $I\mydef (S,k)$.
Our main approach is to use dynamic programming to solve this problem.  Our subproblems are
based on the notion of boxes (i.e., rectangles in the plane), which contain some subset of the input
points that we would like to cluster optimally.
Great care is needed in order to ensure that there are only polynomially many subproblems
we need to consider.  To this end, we prove some structural properties regarding an optimal
solution that enable us to reduce the complexity of the subproblems that we solve.
We first describe some relevant notation and preprocessing that we do
to the input.

\subsection{Preliminaries}
An \emph{edge} $e$ is defined by its distinct head and tail $p,q\in\Reals^2$ and is denoted by $e\mydef pq$. A loop $o$ is defined by its only point $p\in\Reals^2$ and is denoted by $o\mydef pp$.
An \emph{edge set} is defined to be a set of edges or a set containing a single loop.
For a convex polygon $P$, define $\bound P$ to be the boundary of $P$ and $\per(P)$ to be the perimeter of $P$. We use $\E(P)$ to denote the edge set defining $\bound P$ oriented in counterclockwise order. For the special case where $P$ is a single point $p\in \Reals^2$, $\E(P) \mydef \{pp\}$.

For a set of geometric objects $O$, define $\CH(O)$ to be the convex hull of $O$. Specifically, for a set $S$ of points in the plane,
$\CH(S)$ is the convex hull of $S$. With some abuse of notation, we define
$\bound S \mydef \bound \CH(S)$, $\per(S) \mydef\per(\CH(S))$, and $\E(S)\mydef\E(\CH(S))$.
For the special case where $S$ is a single point $p$, define $\E(\{p\}) \mydef \{pp\}$.
We denote by $\G(S)$ the set of the $O(n^2)$ oriented segments defined by any ordered pair of points in $S$.

We define a \emph{$k$-clustering} of $S$ as a partition $\C\mydef \{S_1,\ldots,S_k\}$ of $S$ into $k$ subsets or \emph{clusters} $S_1,\ldots,S_k\subseteq S$.
Let $\Part(S)$ be the set of all possible clusterings of a set $S$.
Let $\Phi(\C)\mydef \sum_{i=1}^k \per(S_i)$ be the cost of the clustering $\C$.
An \emph{optimal} $k$-clustering is a $k$-clustering $\COpt_k(S)$ such that $\Phi(\COpt_k(S))\leq\Phi(\C)$ for any $k$-clustering $\C$.
Here, $\COpt_k(S)$ denotes an arbitrary optimal $k$-clustering of $S$.
We slightly abuse notation and refer to the edges of $\COpt_k(S)$ as the edges induced by the convex hulls of the clusters in $\COpt_k(S)$.
Let $\Opt_k(S)\mydef \Phi(\COpt_k(S))$.

A clustering $\C\mydef \{S_1,\ldots,S_k\}$ is called a \emph{disjoint clustering} if the convex hulls of any two
clusters $S_i,S_j\in\C$  are disjoint, i.e., $\CH(S_i)\cap\CH(S_j)=\emptyset$.

\begin{observation}\label{lem:p-disjoint}
Given a set of points $S$ in the plane, $\COpt_k(S)$ is a disjoint clustering.
\end{observation}

We first ensure that no two points in $S$ have the same $x$- or $y$-coordinate.
This is possible to do in $O(n^2)$ time by computing the slopes of segments between all pairs of points in $S$.
If one of the slopes is vertical or horizontal, we apply a slight rotation to the coordinate system to
eliminate all the horizontal and vertical slopes without introducing new ones.

Let $p_1,\ldots,p_n$ be the points in $S$ sorted by their $x$-coordinates $x_1<\cdots<x_n$.
We construct two vertical lines $v_i^-$ and $v_i^+$ with each $x$-coordinate $x_i$, such that
$v_i^-$ formally is to the left of $p_i$ and $v_i^+$ formally is to the right of $p_i$.
Thus, an edge $p_jp_i$ for $j<i$ intersects $v_i^-$ at $p_i$ but does not intersect $v_i^+$, whereas $p_ip_j$ for $i<j$
intersects $v_i^+$ at $p_i$ but does not intersect $v_i^-$.  The set of all lines $v_i^-,v_i^+$ for all $i\in\{1,\ldots,n\}$
are the \emph{vertical main lines}.  Between any two consecutive vertical main lines $v_i^+,v_{i+1}^-$, we define
$19999$ \emph{vertical help lines} with $x$-coordinates $x_i+\frac{x_{i+1}-x_i}{20000}\cdot j$ for $j\in\{1,\ldots,19999\}$.
That is, these $19999$ vertical help lines induce $20000$ intervals between $x_i$ and $x_{i+1}$, each of which
has the same length given by $\frac{x_{i+1} - x_i}{20000}$.
In a similar way, we define two horizontal main lines with the $y$-coordinate of each input point $p$, one formally below $p$ and the other formally above $p$.
Let $h_1^-,h_1^+,\ldots,h_n^-,h_n^+$ be the horizontal main lines sorted by ascending $y$-coordinate.
We also define $19999$ equidistant horizontal help lines between any two consecutive horizontal main lines.

\paragraph{Boxes.}
Let $\B(S)$ be the set of all rectangles with edges contained in main or help lines.
Note that the size of $\B(S)$ is $O(n^{4})$.
We use $S_B$ as an abbreviation for $B\cap S$ where $B\in \B(S)$.
We denote by $l(B)$ and $r(B)$ the left and right vertical edge of $B$ and by $b(B)$ and $t(B)$ the bottom and top edge of $B$.
We denote by $w(B)$ the \emph{width} of $B$, i.e., the difference in the $x$-coordinates of $r(B)$ and $l(B)$.
Similarly, we denote by $h(B)$ the \emph{height} of $B$, i.e., the difference in the $y$-coordinates of $t(B)$ and $b(B)$.
We define the \emph{length} of a box $B\in\B(S)$ as $\max\{w(B),h(B)\}$ and denote it by $\length(B)$.
Consider an arbitrary $k$-clustering $\C\mydef \{S_1,\ldots,S_k\}$.
For a box $B\in\B(S)$, consider for each $S_i\in\C$ the part of the boundary of $\CH(S_i)$ that is in $B$.
The cost of $\C$ in $B$ is denoted as $\Phi_B(\C)$, and we define it to be the total length of these parts for
all the clusters $S_1,\ldots,S_k$.

For a box $B$, let $\ell$ be a vertical line defined by an $x$-coordinate
lying strictly between the $x$-coordinates of $l(B)$ and $r(B)$.
We say that the vertical line segment $s\mydef \ell\cap B$ is a \emph{vertical separator} of $B$.
A vertical separator of a box $B$ is \emph{good} if it intersects at most two edges in $\COpt_k(S)$.
We also analogously define a \emph{horizontal separator}, along with the notion of a \emph{good horizontal
separator}.  We call a box \emph{elementary} if there are no vertical help or main lines that lie strictly
in between $l(B)$ and $r(B)$, and no horizontal help or main lines that lie strictly in between $b(B)$ and
$t(B)$.

Finally, we define a \emph{vertical strip} of a box $B$ to be a rectangle $T$ contained in $B$ where the bottom edge of $T$ is contained in $b(B)$ and
the top edge of $T$ is contained in $t(B)$ (i.e., the top and bottom edges of $T$ lie on the boundary of $B$).
Similarly, a \emph{horizontal strip} of a box $B$ is a rectangle $H$ contained in $B$ where the left edge
of $H$ is contained in $l(B)$ and the right edge of $H$ is contained in $r(B)$.  For a vertical strip $T$ of a box $B$,
we denote by $w(T)$ the length of the bottom (or top) edge of $T$, and for a horizontal strip $H$ of a box $B$, we denote by
$h(H)$ the length of the left (or right) edge of $H$.

We interpret vertical strips of a box $B$ as the portion of $B$ that lies between two consecutive vertical help lines,
or possibly between a vertical help line and a vertical main line.  We similarly interpret horizontal strips of a box
to be the portion of the box between two consecutive horizontal help lines, or possibly between a horizontal help line and
a horizontal main line.  Note that, with this interpretation, elementary boxes are precisely those that
consist of one vertical strip and one horizontal strip.

\subsection{Structural Properties}
The main structural property we aim to show is that in a subproblem for a box $B\in\B(S)$,
there is a good vertical or horizontal separator contained in a main or help line that divides
$B$ into two smaller boxes.  In particular, in order to ensure subproblems of low complexity,
we aim to maintain the following invariant.

\begin{definition}
A box $B\in\B(S)$ satisfies the \emph{box invariant} if each of its edges is intersected by at most
two edges of $\COpt_k(S)$.
\end{definition}

Hence, we make some useful observations and prove some lemmas that aid us in this goal.
The following observation exploits the general position assumption that no two points in
$S$ have the same $x$- or $y$-coordinate.

\begin{observation}
\label{obs:sumproperty}
Let $B,B_1,B_2\in\B(S)$ be boxes such that $B=B_1\cup B_2$ and $B_1$ and $B_2$ have disjoint interiors.
Let $\C$ be any clustering of $S$.
Then $\Phi_{B}(\C)=\Phi_{B_1}(\C)+\Phi_{B_2}(\C)$.
\end{observation}

We also have the property that no point lies on the boundary of a box $B$.
\begin{observation}
\label{obs:borderpoint}
Due to the way we have defined the main lines, it follows that no box $B \in \B(S)$ has a boundary containing a point from $S$.
Furthermore, if the convex hull $C$ of a cluster of $\COpt_k(S)$ (or an edge $e$ of such a convex hull) intersects the boundary
of a box $B\in\B(S)$, then $C$ (or the edge $e$) is not contained in $B$.
\end{observation}

The following lemma states that, for an optimal solution, the total cost of the portion of its
edges that lie in a box $B$ cannot exceed the perimeter of the box.

\begin{lemma}\label{lemma:costbound}
Let $B\in\B(S)$.
Then $\Phi_{B}(\COpt_k(S))\leq 2 w(B)+2h(B)$.
\end{lemma}

\begin{proof}
Assume towards a contradiction that this is not the case.
Let $S_1,\ldots,S_k$ be the clusters of $\COpt_k(S)$ with convex hulls intersecting $B$.
Let $P\mydef B\cup \bigcup_{i=1}^k \CH(S_i)$.
Now, $P$ has a perimeter strictly smaller than $\sum_{i=1}^k \per(S_i)$.
However, the merged cluster $\bigcup_{i=1}^k S_i$ has a convex hull with a perimeter
at most as large as $P$, which is a contradiction.
\end{proof}

The following lemma is a key ingredient in our algorithm, and it is the main structural property
we show in this section.

\begin{lemma}\label{lemma:box-split}
If $B \in \B(S)$ is an elementary box, then there are at most two edges of $\COpt_k(S)$
intersecting $B$ (in particular, $B$ satisfies the box invariant).  Otherwise, let $B \in \B(S)$ be any box
satisfying the box invariant.  Then $B$ has a good vertical or horizontal separator $s$ contained in a main or help line
such that $s$ divides $B$ into two strictly smaller boxes $B_l$ and $B_r$, both of which satisfy the box invariant.
\end{lemma}

We prove the following sequence of lemmas, which we later show how to combine to yield a proof of
Lemma~\ref{lemma:box-split}.

\begin{lemma}\label{lem:contsep}
Consider any box $B \in \B(S)$.  Let $X$ be an arbitrary partition of $B$ into vertical strips, and let $Y$ be an
arbitrary partition of $B$ into horizontal strips.  Moreover, let $A \mydef \{A_1,\ldots,A_m\},Z = \{Z_1,\ldots,Z_q\}$ be
arbitrary nonempty subsets of $X,Y$, respectively.  Then either there is a vertical separator of $B$ contained in some vertical
strip $A_i \in A$ or a horizontal separator of $B$ contained in some horizontal strip $Z_j \in Z$ that intersects at most
$2\sqrt{2}\frac{w(B) + h(B)}{\sum_{A_i \in A}w(A_i) + \sum_{Z_j \in Z}h(Z_j)}$ edges of $\COpt_k(S)$.
\end{lemma}
\begin{proof}
Define $t_1,t_2$ to be the values such that the edges $l(B)$ and $r(B)$ are contained in the vertical lines $x=x_1$ and $x=x_2$,
respectively.  Similarly, define $y_1,y_2$ to be the values such that the edges $b(B)$ and $t(B)$ are contained in the horizontal
lines $y = y_1$ and $y=y_2$, respectively.  That is, $x_1,x_2,y_1,y_2$ define the coordinates of the bounding box $B$.
In addition, let $x_1 \mydef t_1 < t_2 < \cdots < t_{m+1} \mydef x_2$ be the values such that, for each vertical strip $A_i$,
the left and right edges of $A_i$ are contained in the vertical lines $x = t_{i}$ and $x = t_{i+1}$, respectively.
Similarly, let $y_1 \mydef s_1 < s_2 < \cdots < s_{q+1} \mydef y_2$ be the values such that, for each horizontal strip $Z_i$,
the lower and upper edges of $Z_i$ are contained in the horizontal lines $y = s_i$ and $y = s_{i+1}$, respectively.

Assume towards a contradiction that all vertical separators that are contained in some vertical strip $A_i \in A$ and
all horizontal separators that are contained in some horizontal strip $Z_j \in Z$ intersect
strictly more than $2\sqrt{2}\frac{w(B) + h(B)}{\sum_{A_i \in A}w(A_i) + \sum_{Z_j \in Z}h(Z_j)}$ edges of $\COpt_k(S)$.
Define the function $f(x)$ to be the number of edges of $\COpt_k(S)$ that have at least some portion in $B$
and intersect the vertical separator at $x$ (for each $x_1 \leq x \leq x_2$).  Similarly, define the function $g(y)$
to be the number of edges of $\COpt_k(S)$ that have at least some portion in $B$ and intersect the horizontal separator at $y$
(for each $y_1 \leq y \leq y_2$).

Consider an edge $e$ (or portion of an edge) of $\COpt_k(S)$ that lies in $B$, and suppose it has length $\delta$.
We want to understand its contribution to the sum of integrals $\int_{x_1}^{x_2} f(x)dx + \int_{y_1}^{y_2}g(y)dy$ (e.g.,
a vertical segment contributes $\delta$ to the sum, along with a horizontal segment).  In particular, we want
to understand the maximum contribution possible by such a line segment (i.e., edge).  To this end, consider the right triangle formed
by the two endpoints of the line segment (which is the hypotenuse) along with the point defined by the intersection of the vertical line
passing through the upper endpoint and the horizontal line passing through the lower endpoint.  Let $\delta_1$ denote the length
of one leg and $\delta_2$ denote the length of the other leg, and observe that $\delta_1 + \delta_2$ is precisely the contribution
of edge $e$ to the sum of integrals.  Hence, we wish to maximize $\delta_1 + \delta_2$ subject to the constraint $\delta^2 = \delta_1^2 + \delta_2^2$,
which yields $\delta_1 = \delta_2 = \frac{\delta}{\sqrt{2}}$ (obtained when the triangle is a right isosceles triangle).
This implies that an edge of length $\delta$ contributes at most $\sqrt{2}\delta$ to the sum of integrals.

By Lemma~\ref{lemma:costbound}, we know that $\Phi_{B}(\COpt_k(S))\leq 2w(B)+2h(B)$.  Hence, we obtain a contradiction as follows:
\begin{align*}
2w(B)+2h(B) &\geq \Phi_{B}(\COpt_k(S)) \geq \frac{\int_{x_1}^{x_2} f(x)dx + \int_{y_1}^{y_2} g(y)dy}{\sqrt{2}}\\
&\geq \frac{\sum_{i=1}^m \int_{t_i}^{t_{i+1}}f(x)dx + \sum_{j=1}^{q}\int_{s_j}^{s_{j+1}}g(y)dy}{\sqrt{2}}\\
&> \left(2\frac{w(B) + h(B)}{\sum_{A_i \in A}w(A_i) + \sum_{Z_j \in Z}h(Z_j)}\right)\left(\sum_{A_i \in A}w(A_i) + \sum_{Z_j \in Z}h(Z_j)\right)\\
&= 2w(B)+2h(B).
\end{align*}
Here, the second inequality follows from the maximum amount that an edge (or portion of an edge) of $\COpt_k(S)$ can contribute to the sum of integrals,
the third inequality follows since we are integrating over a smaller domain, and the last inequality follows from our assumption that each vertical
separator that belongs to some vertical strip intersects many edges (and similarly for such horizontal separators).
\end{proof}

We first consider boxes that have many small vertical strips and many small horizontal strips (i.e., many main lines and help lines).
We argue that such boxes have either a good vertical separator or a good horizontal separator in the following lemma.

\begin{lemma}\label{lem:bidensesep}
Let $B$ be any box and consider the vertical strips induced by the vertical help (and main) lines going through $B$,
along with the horizontal strips induced by the horizontal help (and main lines) going through $B$.  Moreover,
suppose that for every such vertical strip $T$ we have $w(T) \leq \frac{1}{200}w(B)$ and for all such horizontal strips
$H$ we have $h(H) \leq \frac{1}{200}h(B)$.  Then there exists a good vertical separator of $B$
that is contained in a vertical help or main line that splits $B$ into two smaller boxes (i.e., the separator
does not lie on $l(B)$ or $r(B)$), or there exists a good horizontal separator of $B$ with similar properties.
\end{lemma}
\begin{proof}
Consider any such box $B$ (i.e., $B$ has many horizontal and vertical lines that are close).  We consider
each portion of $B$ that lies between two consecutive vertical help lines (or between a vertical help line and a vertical
main line) to be a vertical strip of $B$.  We similarly consider the horizontal strips of $B$ induced by the corresponding horizontal
help and main lines.  Throughout the proof, we assume inductively that each of the four edges of $B$ is intersected
by at most two edges of $\COpt_k(S)$ (and hence, there are at most eight such edges in total intersecting the boundary of $B$).

One of the vertical strips has a left edge
that is precisely $l(B)$, and another one of the vertical strips has a right edge that is precisely $r(B)$.
In particular, since we aim to split the box $B$ into two smaller boxes with our separator, we discard these
two vertical strips.  Moreover, we also discard any vertical strip $T$ that has an edge of $\COpt_k(S)$ intersecting
the top edge of $T$ or the bottom edge of $T$.  By our inductive assumption, there are at most four such edges
of $\COpt_k(S)$ (since at most two such edges intersect $t(B)$, and at most two such edges intersect $b(B)$).
Each such edge of $\COpt_k(S)$ can result in discarding at most two vertical strips (in case the edge intersects
the corner of a vertical strip, in which case two strips are affected).  Hence, in all, we discard at most 10
vertical strips.  Similar reasoning applies to horizontal strips, resulting in discarding at most $10$ horizontal
strips.  Note that the total width of vertical strips that are not discarded is at least $(1 - \frac{10}{200})w(B) = \frac{19}{20}w(B)$,
and similarly the total height of horizontal strips that are not discarded is at least $\frac{19}{20}h(B)$.

By Lemma~\ref{lem:contsep}, either there exists a vertical separator of $B$ contained in some vertical
strip that is not discarded or a horizontal separator of $B$ contained in some horizontal strip that is not
discarded that intersects at most $2\sqrt{2}\frac{w(B) + h(B)}{\frac{19}{20}(w(B) + h(B))} < 3$ edges of $\COpt_k(S)$.
Suppose there is some such remaining (i.e., not discarded) vertical strip $T$ containing a good vertical separator
(as the other case is symmetric).  Since $T$
was not discarded, its left edge is not $l(B)$ and its right edge is not $r(B)$.  Moreover, since there
are no edges of $\COpt_k(S)$ intersecting the top or bottom edge of $T$, there can be no edge
of $\COpt_k(S)$ that intersects the left or right edge of $T$ without also intersecting the good vertical
separator contained in $T$.  Hence, either the left or right edge of $T$ serves as a good vertical separator of $B$.
\end{proof}

We now argue that boxes that consist of many vertical strips (induced by the corresponding vertical help or main lines)
and are much wider than they are tall have a good vertical separator.  A similar result holds in the horizontal direction.

\begin{lemma}\label{lem:unbalanceddensesep}
Let $B$ be any box and consider the vertical strips induced by the vertical help (and main) lines going through $B$,
along with the horizontal strips induced by the horizontal help (and main lines) going through $B$.  Moreover,
suppose that $w(B) > 4h(B)$ and all such vertical strips $T$ have $w(T) \leq \frac{1}{200}w(B)$.
Then there exists a good vertical separator of $B$ that is contained in a vertical help or main line that splits
$B$ into two smaller boxes (i.e., the separator does not lie on $l(B)$ or $r(B)$).
Likewise, if $h(B) > 4w(B)$ and all such horizontal strips
$H$ have $h(H) \leq \frac{1}{200}h(B)$, then there exists a good horizontal separator of $B$ with similar properties.
\end{lemma}
\begin{proof}
Consider any such box $B$, and suppose $w(B) > 4h(B)$, as the other case is proved similarly.  We consider
each portion of $B$ that lies between two consecutive vertical help lines (or between a vertical help line and a vertical
main line) to be a vertical strip of $B$.  Assume that all such vertical strips $T$ satisfy $w(T) \leq \frac{1}{200} w(B)$.
Moreover, throughout the proof, we assume inductively that each of the
four edges of $B$ is intersected by at most two edges of $\COpt_k(S)$.  We proceed in a manner
similar to the proof of Lemma~\ref{lem:bidensesep}.

One of these vertical strips has a left edge
that is precisely $l(B)$, and another one of these vertical strips has a right edge that is precisely $r(B)$.
In particular, since we aim to split the box $B$ into two smaller boxes with our separator, we discard these
two vertical strips.  Moreover, we also discard any vertical strip $T$ that has an edge of $\COpt_k(S)$ intersecting
the top edge of $T$ or the bottom edge of $T$.  By our inductive assumption, there are at most four such edges
of $\COpt_k(S)$ (since at most two such edges intersect $t(B)$, and at most two such edges intersect $b(B)$).
Each such edge of $\COpt_k(S)$ can result in discarding at most two vertical strips (in case the edge intersects
the corner of a vertical strip, in which case two strips are affected).  Hence, in all, we discard at most $10$
vertical strips.  Note that the total width of vertical strips that are not discarded is at least $\frac{19}{20}w(B)$.

Now, suppose towards a contradiction that all vertical separators in the remaining vertical strips intersect at least three
edges of $\COpt_k(S)$.  This implies that the total length of such edges is at least $\frac{57}{20}w(B)$.
By Lemma~\ref{lemma:costbound}, we know that $\Phi_{B}(\COpt_k(S))\leq 2w(B)+2h(B)$.  Hence, we obtain a contradiction:
$$\frac{57}{20}w(B) \leq \Phi_{B}(\COpt_k(S)) \leq 2w(B) + 2h(B) \leq 2w(B) + \frac{w(B)}{2} = \frac{5}{2}w(B).$$
Hence, there is some remaining vertical strip $T$ that contains a good vertical separator.  Since $T$
was not discarded, its left edge is not $l(B)$ and its right edge is not $r(B)$.  Moreover, since there
are no edges of $\COpt_k(S)$ intersecting the top or bottom edge of $T$, there can be no edge
of $\COpt_k(S)$ that intersects the left or right edge of $T$ without also intersecting the good vertical
separator contained in $T$.  Hence, either the left or right edge of $T$ serves as a good vertical separator.
\end{proof}

We now consider the case regarding boxes that either have some wide vertical strip and are not much taller than
they are wide, or have some tall horizontal strip and are not much wider than they are tall.

\begin{lemma}\label{lem:balancedsparsesep}
Let $B$ be any box and consider the vertical strips induced by the vertical help (and main) lines going through $B$,
along with the horizontal strips induced by the horizontal help (and main lines) going through~$B$.  Moreover,
suppose that $h(B) \leq 4w(B)$ and there exists some vertical strip $T$ with $w(T) > \frac{1}{200}w(B)$.
Then the left edge of $T$ has at most two edges of $\COpt_k(S)$ intersecting it, and the same holds for the right edge of $T$.
Similarly, if $w(B) \leq 4h(B)$ and there exists some horizontal strip $H$ with $w(H) > \frac{1}{200}h(B)$,
then each of the bottom and top edges of $H$ has at most two edges of $\COpt_k(S)$ intersecting them.
\end{lemma}
\begin{proof}
Consider any such box $B$, and suppose $h(B) \leq 4w(B)$, as the other case is proved similarly.  We consider
each portion of $B$ that lies between two consecutive vertical help lines (or between a vertical help line and a vertical
main line) to be a vertical strip of $B$.  Assume that there exists some vertical strip $T$ with $w(T) > \frac{1}{200} w(B)$.

Now, consider the two points $p_1,p_2$
that give rise to the vertical strip $T$ induced by the corresponding vertical help and main lines, and let
$x_1,x_2$ be the $x$-coordinates of $p_1,p_2$ (respectively), with $x_1 < x_2$.  Observe that
$h(B) \leq 4w(B) < 800w(T) = 800\frac{x_2 - x_1}{20000} = \frac{(x_2-x_1)}{25}$, where the inequalities follow from our assumptions
in the lemma and the equality follows from the fact that we divide the interval $[x_1,x_2]$ into $20000$ segments
of equal length (induced by the vertical help lines).

Suppose towards a contradiction that the left or right edge of $T$ has at least three edges of $\COpt_k(S)$ intersecting it.  Let
$s$ denote this vertical edge, and let $u$ denote the $x$-coordinate induced by the line $s$.  We must have that
at least one of $x_1,x_2$ must be far away from $u$ (note that $x_1 \leq u \leq x_2$).  In particular, either $x_2 \geq u + \frac{x_2 - x_1}{2}$
or $x_1 \leq u - \frac{x_2 - x_1}{2}$.  We consider the former case as the latter is symmetric, so suppose
$x_2 \geq u + \frac{x_2 - x_1}{2}$.

Let $e_1,e_2,e_3$ denote three consecutive edges of $\COpt_k(S)$ that intersect the line segment $s$, sorted in
decreasing order according to the height of the point of intersection with $s$ (e.g., $e_1$ is above $e_2$, which is above $e_3$ at $x=u$).
Note that it cannot be the case that all three edges are contained in the boundary of the convex hull of the same cluster, and hence the three edges lie on the
boundaries of at least two convex hulls (induced by at least two clusters of $\COpt_k(S)$).  Without loss of generality,
assume that there exists an optimal cluster $C$ such that $\CH(C)$ contains both $e_1$ and $e_2$
(the case when $e_2$ and $e_3$ are on the same boundary is symmetric).  Moreover, let $C'$ denote the cluster giving rise
to the boundary on which $e_3$ lies.

We first consider the case when the slope of $e_1$ is strictly more than
that of $e_2$, and the slope of $e_2$ is strictly more than that of $e_3$.
We imagine extending the line segments $e_1,e_2$ to the left until they meet at some point $r_{1,2}$ (they must meet at some point), and denote
by $\theta_{1,2}$ the angle in radians formed as a result of extension.
We can do a similar process for edges $e_2,e_3$ to obtain the angle $\theta_{2,3}$ (again, extending line segments
$e_2$ and $e_3$ to the left must result in an intersection), and also for edges $e_1,e_3$ to obtain an angle~$\theta_{1,3}$.
Observe that $\theta_{1,3} \leq \pi$, implying that either $\theta_{1,2} \leq \frac{\pi}{2}$ or $\theta_{2,3} \leq \frac{\pi}{2}$.

Consider the scenario where $\theta_{1,2} \leq \frac{\pi}{2}$.  We transform the solution as follows.  First, we split the cluster $C$ into two clusters,
where we take all points in $C$ with an $x$-coordinate of at most $u$ to one cluster $C_1$, and all points in $C$ with an $x$-coordinate
of at least $x_2$ to another cluster $C_2$ (note that there are no points with an $x$-coordinate in the open interval
$(u,x_2)$).  Lastly, we merge the clusters $C_1$ and $C'$ by taking the union $C_1 \cup C'$.  We note that the total number of clusters,
after performing the split and merge, remains unchanged.  Hence, we need only argue that the cost after splitting and merging
does not result in an increase in cost.

To this end, observe that the points in $C_1$ all lie inside the polygon $P_1$ obtained by considering the original cluster $C$,
cutting it with the vertical line $x=u$, and taking the left portion (i.e., all points in $\CH(C)$ that have an $x$-coordinate of at most $u$).
By a similar argument, the points in $C_2$ all lie inside the polygon $P_2$ obtained by considering the original cluster
$C$, cutting it with the vertical line $x = x_2$, and taking the right portion (i.e., all points in $\CH(C)$
that have an $x$-coordinate of at least $x_2$).  Now, consider the polygon $Q$ obtained by merging the polygons $P_1$ and $\CH(C')$ as follows:
we consider a vertical line segment at $x=u$ going from edge $e_2$ to edge $e_3$ (note that both edges intersect line $s$ at $x = u$).
Observe that $Q$ contains the points in $C_1 \cup C'$ and $P_2$ contains the points in $C_2$, and hence $\per(C_1 \cup C')$ is at most
the perimeter of $Q$ and $\per(C_2)$ is at most the perimeter of $P_2$.  Hence, we need only bound the perimeter of $Q$ and $P_2$.

Observe that the combined perimeter of $Q$ and $P_2$ is at most $\per(C) + \per(C')$, plus the lengths of the vertical line segments at $x=u$
and $x=x_2$ going from edge $e_1$ to $e_2$, plus twice the length of the vertical line segment at $x=u$ going
from edge $e_2$ to $e_3$, minus the lengths of the portion of edges $e_1$ and $e_2$ with $x$-coordinates in between $x = u$ and $x=x_2$.
For ease of notation, we denote by $b_{1,2}$ the length of the vertical line segment at $x=x_2$ (going from $e_1$ to $e_2$),
by $\ell_1$ the length of $e_1$ in the interval $[u,x_2]$, and by $\ell_2$ the length of $e_2$ in the interval $[u,x_2]$.  First, the vertical
line segments at $x=u$ (which, joined together, go from $e_1$ to $e_3$ at $x = u$) contribute at most $2h(B)$ to the perimeter of $Q$.
Now, we upper bound $b_{1,2}$ by considering the triangle formed by the following three points: $r_{1,2}$, the intersection of $e_1$
with the vertical line $x = x_2$, and the intersection of $e_2$ with the vertical line $x = x_2$.  Let $\beta$ denote the angle formed
at the intersection of $e_1$ with $x=x_2$, and $\gamma$ denote the angle formed at the intersection of $e_2$ with $x=x_2$.  We have
that $b_{1,2} = h(B) + \ell_1 \cos(\beta) + \ell_2 \cos(\gamma)$.  Since $\theta_{1,2} \leq \frac{\pi}{2}$, we know that
$\max\{\beta,\gamma\} \geq \frac{\pi}{4}$.  This implies
$b_{1,2} = h(B) + \ell_1 \cos(\beta) + \ell_2 \cos(\gamma) \leq h(B) + \max\{\ell_1 + \frac{\ell_2}{\sqrt{2}}, \frac{\ell_1}{\sqrt{2}} + \ell_2\}$.
Hence, the new solution's cost, given by the sum of the perimeters of $Q$ and $P_2$, is at most
$\per(C) + \per(C') + 3h(B) + \max\{\ell_1 + \frac{\ell_2}{\sqrt{2}}, \frac{\ell_1}{\sqrt{2}} + \ell_2\} - \ell_1 - \ell_2$.  We bound this expression
as follows:
\begin{align*}
3h(B) &+ \max\left\{\ell_1 + \frac{\ell_2}{\sqrt{2}}, \frac{\ell_1}{\sqrt{2}} + \ell_2 \right\} - \ell_1 - \ell_2\\
&\leq 3h(B) - \left(1 - \frac{1}{\sqrt{2}}\right)\min\{\ell_1,\ell_2\}
\leq 3h(B) - \left(1 - \frac{1}{\sqrt{2}}\right)(x_2 - u),
\end{align*}
where the last inequality follows from the fact that both $\ell_1$ and $\ell_2$ are at least $x_2 - u$.  Hence, as long
as $3h(B) - (1 - \frac{1}{\sqrt{2}})(x_2 - u) < 0$, we have a contradiction.  This holds as $h(B) < \frac{(x_2 - x_1)}{25}$ and
$x_2 - u \geq \frac{x_2-x_1}{2}$, implying $3h(B) - (1 - \frac{1}{\sqrt{2}})(x_2 - u) < 0$.

In the case that $\theta_{1,2} > \frac{\pi}{2}$, then we know $\theta_{2,3} \leq \frac{\pi}{2}$.  In this setting, we
obtain a new solution by merging the two clusters $C$ and $C'$ to get $C \cup C'$, and argue that this new solution is cheaper.
Consider the polygon $P$ obtained by removing the portion of edges $e_2$ and $e_3$ in between $x = u$ and $x = x_2$,
and then joining $C$ and $C'$ on the left by a vertical line segment going from the point at which $e_2$ intersects
$s$ to the point at which $e_3$ intersects $s$ (at $x = u$).  Similarly, we join $C$ and $C'$ on the right by a vertical line segment between
the two points at which edges $e_2$ and $e_3$ intersect at $x=x_2$.  Now, since $P$ contains all points in $C \cup C'$, we have $\per(C \cup C')$
is at most the perimeter of $P$.  Moreover, we have that the perimeter of $P$ is at most $\per(C) + \per(C')$, plus the lengths of the
left and right vertical segments we added, minus $\ell_2$ and $\ell_3$.  The left vertical segment contributes at most $h(B)$ to
the perimeter (since the left vertical line segment is contained in $s$).  By a similar argument as earlier, the right vertical segment
is at most $h(B) + \max\{\ell_2 + \frac{\ell_3}{\sqrt{2}},\frac{\ell_2}{\sqrt{2}} + \ell_3\}$ (using the fact that $\theta_{2,3} \leq \frac{\pi}{2}$ in this case).
Hence, using similar reasoning as before, the perimeter of $P$ is at most
$\per(C) + \per(C') + 2h(B) + \max\left\{\ell_2 + \frac{\ell_3}{\sqrt{2}},\frac{\ell_2}{\sqrt{2}} + \ell_3\right\} - \ell_2 - \ell_3$.
We bound this expression as follows:
\begin{align*}
2h(B) + \max\left\{\ell_2 + \frac{\ell_3}{\sqrt{2}},\frac{\ell_2}{\sqrt{2}} + \ell_3\right\} - \ell_2 - \ell_3
&\leq 2h(B) - \left(1 - \frac{1}{\sqrt{2}}\right)\min\{\ell_2,\ell_3\}\\
&\leq 2h(B) - \left(1 - \frac{1}{\sqrt{2}}\right)(x_2 - u).
\end{align*}
We again have that this quantity is strictly less than zero, a contradiction.

In the case when either the slope of $e_1$ is at most the slope of $e_2$, or the slope of $e_2$ is at most the slope of $e_3$,
the proof is simpler.  In particular, in the former case, we can more easily bound the vertical line segment we add along the line $x = x_2$
by $h(B)$ (instead of some function of $\ell_1,\ell_2$).  This holds since the gap between edges $e_1$ and $e_2$ shrinks when
going from $x=u$ to $x=x_2$.  The same holds in the latter case when the slope of $e_2$ is at most the slope of $e_3$.
\end{proof}

Finally, as a type of base case, we argue that boxes that are not much taller than they are wide and consist of one
vertical strip have at most two edges of $\COpt_k(S)$ intersecting the whole box.  Likewise, boxes that are not much wider than they are tall and consist of
one horizontal strip satisfy the same property.

\begin{lemma}\label{lem:twoedgebox}
Let $B$ be any box with $h(B) \leq 4w(B)$ such that no vertical help or main line has an $x$-coordinate strictly in between the $x$-coordinates
induced by $l(B)$ and $r(B)$.  Then there are at most two edges of $\COpt_k(S)$ intersecting box $B$.
Similarly, if $B$ is any box with $w(B) \leq 4h(B)$ such that no horizontal help or main line has a $y$-coordinate strictly in between the
$y$-coordinates induced by $b(B)$ and $t(B)$, then there are at most two edges of $\COpt_k(S)$ intersecting box $B$.
\end{lemma}
\begin{proof}
Consider any such box $B$ with $h(B) \leq 4w(B)$ (as the proof of the other case is symmetric).  We consider
each portion of $B$ that lies between two consecutive vertical help lines (or between a vertical help line and a vertical
main line) to be a vertical strip of $B$.  In particular, the assumption in the lemma regarding vertical help and
main lines implies that $B$ consists of exactly one vertical strip $T$ (induced by $l(B)$ and $r(B)$), and hence we have
$w(B) = w(T)$.  The following proof uses similar ideas as in the proof of Lemma~\ref{lem:balancedsparsesep}.

Now, consider the two points $p_1,p_2$
that give rise to the vertical strip $T$, and let
$x_1,x_2$ be the $x$-coordinates of $p_1,p_2$ (respectively), with $x_1 < x_2$.  Observe that
$h(B) \leq 4w(B) = 4w(T) = 4\frac{x_2 - x_1}{20000} < \frac{x_2-x_1}{200}$, where the first inequality and the first equality follow
from our assumptions in the lemma, and the second equality follows from the fact that we divide the interval $[x_1,x_2]$ into $20000$ segments
of equal length (induced by the vertical help lines).

Suppose towards a contradiction that there are (at least) three edges of $\COpt_k(S)$ intersecting the interior of $B$.  Let
$u_1,u_2$ denote the $x$-coordinates induced by the line segments $l(B)$ and $r(B)$, respectively.  We must have either
$x_2 \geq \frac{u_1 + u_2}{2} + \frac{x_2 - x_1}{2}$
or $x_1 \leq \frac{u_1 + u_2}{2} - \frac{x_2 - x_1}{2}$.  Otherwise, we get a contradiction:
$x_2 - x_1 < \frac{u_1 + u_2}{2} + \frac{x_2 - x_1}{2} - \frac{u_1 + u_2}{2} + \frac{x_2 - x_1}{2} = x_2 - x_1$.
We consider the former case as the latter is symmetric, so suppose
$x_2 \geq \frac{u_1 + u_2}{2} + \frac{x_2 - x_1}{2}$.  In the former case, we focus on the vertical line $x=u_1$, namely
the line containing $l(B)$ (in the latter case, we focus on the vertical line $x=u_2$, namely the line containing $r(B)$).

Let $e_1,e_2,e_3$ denote three edges of $\COpt_k(S)$ that intersect the interior of box $B$.  Each such edge must also
intersect the vertical lines $x=u_1$ and $x=x_2$.  We consider three such edges that are consecutive, sorted in
decreasing order according to the height of the point of intersection with the line $x=u_1$
(e.g., $e_1$ is above $e_2$, which is above $e_3$ at $x=u_1$).
Note that it cannot be the case that all three edges are contained in the boundary of the convex hull of the same cluster, and hence the three edges lie on the
boundaries of at least two convex hulls (induced by at least two clusters of $\COpt_k(S)$).  Without loss of generality,
assume that there exists an optimal cluster $C$ such that $\CH(C)$ contains both $e_1$ and $e_2$
(the case when $e_2$ and $e_3$ are on the same boundary is symmetric).  We also denote by $C'$ the cluster giving rise
to the boundary on which $e_3$ lies.

For ease of notation, we denote by $\ell_1,\ell_2,\ell_3$ the lengths of $e_1,e_2,e_3$
in the interval $[u_1,x_2]$, and by $g_1,g_2,g_3$ the lengths of $e_1,e_2,e_3$ in an arbitrary interval of length
$\frac{x_2-x_1}{20000}$ in $[u_1,x_2]$.  We also denote by $a_{1,2}$ the length of the vertical line segment at $x=u_1$ connecting $e_1$
and $e_2$, by $a_{2,3}$ the segment at $x=u_1$ connecting $e_2$ and $e_3$, and $a_{1,3} = a_{1,2} + a_{2,3}$ (i.e.,
the length of the vertical line segment at $x = u_1$ going from $e_1$ to $e_3$).  We define analogous lengths
of vertical line segments $b_{1,2}$ and $b_{2,3}$ at $x = x_2$ between edges $e_1,e_2$ and edges $e_2,e_3$, respectively.
We note that $\ell_i \geq 1000 g_i$ for $1 \leq i \leq 3$, since there are at least $1000$ disjoint intervals of length $\frac{x_2-x_1}{20000}$
in the interval $[u_1,x_2]$ (using the fact that $x_2 \geq \frac{u_1 + u_2}{2} + \frac{x_2 - x_1}{2}$), and in each such
interval the length of $e_i$ is precisely $g_i$.

We first argue that $a_{i,j} \leq g_i + g_j + h(B)$ for all $1 \leq i < j \leq 3$.
For any two such edges $e_i,e_j$ we argue this by considering the following three line segments.
We can go along $e_i$ from the intersection of edge $e_i$ with $x= u_1$ to any intersection point of $e_i$
with box $B$, and then to any intersection point of $e_j$ with box $B$, and finally go along $e_j$ back to $x = u_1$.
The cost of connecting $e_i$ to $e_j$ along $x = u_1$ is at most the lengths of the projection of these three
line segments onto $x = u_1$, which sum to at most $g_i + g_j + h(B)$.

Moreover, we show a tighter bound of $a_{i,j} \leq \min\{g_i,g_j\} + h(B)$ for all $1 \leq i < j \leq 3$ satisfying the property that the
slope of $e_i$ is strictly more than that of $e_j$.
If the slopes of $e_i,e_j$ are positive and negative, respectively, then the intersections of $e_i,e_j$ with $x = u_1$ must both lie in $B$, and
hence the cost of connecting them via the vertical line segment at
$x = u_1$ is $a_{i,j} \leq h(B)$.  If the slope of $e_i$ is negative, in which case the slope of $e_j$ is also negative,
then the intersection points of both edges with
$x = u_1$ must lie above $b(B)$, and hence we can obtain the bound via the following two line segments.
We can go from the intersection of edge $e_i$ with $x= u_1$ to any point at which it intersects
box $B$, and then go to the intersection of $b(B)$ with $x = u_1$.  The cost of connecting $e_i$ to $e_j$ along $x = u_1$
is at most the lengths of the projection of these two line segments onto $x = u_1$, which sum to at most $g_i + h(B)$ (note that
$g_i = \min\{g_i,g_j\}$, since both slopes are negative and the slope of $e_i$ is strictly more than that of $e_j$).
A symmetric argument shows that, in the case that $e_j$ is positive (in which case the slope of $e_i$ is also positive),
we have $a_{i,j} \leq g_j + h(B)$ (we symmetrically have $g_j = \min\{g_i,g_j\}$).

We first consider the case when the slope of $e_1$ is strictly more than
that of $e_2$, and the slope of $e_2$ is strictly more than that of $e_3$.
This implies $a_{1,2} \leq \min\{g_1,g_2,g_3\} + h(B)$ and $a_{2,3} \leq \min\{g_1,g_2,g_3\} + h(B)$.
In particular, we know $a_{1,2} \leq \min\{g_1,g_2\} + h(B)$, and we also know $a_{1,2} \leq a_{1,3} \leq \min\{g_1,g_3\} + h(B)$,
implying $a_{1,2} \leq \min\{g_1,g_2,g_3\} + h(B)$.  Symmetrically, we have $a_{2,3} \leq \min\{g_1,g_2,g_3\} + h(B)$.

Now, we imagine extending the line segments $e_1,e_2$ to the left until they meet at some point $r_{1,2}$ (they must meet at some point), and denote
by $\theta_{1,2}$ the angle in radians formed as a result of extension.
We can do a similar process for edges $e_2,e_3$ to obtain the angle $\theta_{2,3}$ (again, extending line segments
$e_2$ and $e_3$ to the left must result in an intersection), and also for edges $e_1,e_3$ to obtain an angle $\theta_{1,3}$.
Observe that $\theta_{1,3} \leq \pi$, implying that either $\theta_{1,2} \leq \frac{\pi}{2}$ or $\theta_{2,3} \leq \frac{\pi}{2}$.

Consider the scenario where $\theta_{1,2} \leq \frac{\pi}{2}$.  We transform the solution as follows.  First, we split the cluster $C$ into two clusters,
where we take all points in $C$ with an $x$-coordinate of at most $u_1$ to one cluster $C_1$, and all points in $C$ with an $x$-coordinate
of at least $x_2$ to another cluster $C_2$ (note that there are no points with an $x$-coordinate in the open interval
$(u_1,x_2)$).  Lastly, we merge the clusters $C_1$ and $C'$ by taking the union $C_1 \cup C'$.  We note that the total number of clusters,
after performing the split and merge, remains unchanged.  Hence, we need only argue that the cost after splitting and merging
does not result in an increase in cost.

To this end, observe that the points in $C_1$ all lie inside the polygon $P_1$ obtained by considering the original cluster $C$,
cutting it with the vertical line $x=u_1$, and taking the left portion (i.e., all points in $\CH(C)$ that have an $x$-coordinate of at most $u_1$).
By a similar argument, the points in $C_2$ all lie inside the polygon $P_2$ obtained by considering the original cluster
$C$, cutting it with the vertical line $x = x_2$, and taking the right portion (i.e., all points in $\CH(C)$
that have an $x$-coordinate of at least $x_2$).  Now, consider the polygon $Q$ obtained by merging the polygons $P_1$ and $\CH(C')$
via a vertical line segment at $x=u_1$ going from edge $e_2$ to edge $e_3$.
Observe that $Q$ contains the points in $C_1 \cup C'$ and $P_2$ contains the points in $C_2$, and hence $\per(C_1 \cup C')$ is at most
the perimeter of $Q$ and $\per(C_2)$ is at most the perimeter of $P_2$.
Moreover, the combined perimeters of $Q$ and $P_2$ is at most $\per(C) + \per(C') + a_{1,2} + 2a_{2,3} + b_{1,2} - \ell_1 - \ell_2$.

Now, we upper bound $b_{1,2}$ by considering the triangle formed by the following three points: $r_{1,2}$, the intersection of $e_1$
with the vertical line $x = x_2$, and the intersection of $e_2$ with the vertical line $x = x_2$.  Let $\beta$ denote the angle formed
at the intersection of $e_1$ with $x=x_2$, and $\gamma$ denote the angle formed at the intersection of $e_2$ with $x=x_2$.  We have
that $b_{1,2} = a_{1,2} + \ell_1 \cos(\beta) + \ell_2 \cos(\gamma)$.  Since $\theta_{1,2} \leq \frac{\pi}{2}$, we know that
$\max\{\beta,\gamma\} \geq \frac{\pi}{4}$.  This implies
$b_{1,2} = a_{1,2} + \ell_1 \cos(\beta) + \ell_2 \cos(\gamma) \leq a_{1,2} + \max\{\ell_1 + \frac{\ell_2}{\sqrt{2}}, \frac{\ell_1}{\sqrt{2}} + \ell_2\}$.
Hence, we obtain
\begin{align*}
a_{1,2} + 2a_{2,3} + b_{1,2} - \ell_1 - \ell_2 &\leq
2(a_{1,2} + a_{2,3}) + \max\left\{\ell_1 + \frac{\ell_2}{\sqrt{2}}, \frac{\ell_1}{\sqrt{2}} + \ell_2\right\} - \ell_1 - \ell_2\\
&\leq 4(\min\{g_1,g_2\} + h(B)) - \left(1 - \frac{1}{\sqrt{2}}\right)\min\{\ell_1,\ell_2\}\\
&\leq 4(\min\{g_1,g_2\} + 4w(B)) - 1000\left(1 - \frac{1}{\sqrt{2}}\right)\min\{g_1,g_2\} < 0,
\end{align*}
where the first inequality follows from substituting our upper bound on $b_{1,2}$,
the second inequality follows from substituting for our upper bounds on $a_{1,2}$ and $a_{2,3}$ and simplifying,
the third inequality follows by the assumption $h(B) \leq 4w(B)$, along with $\ell_1 \geq 1000 g_1$ and $\ell_2 \geq 1000 g_2$,
and the last inequality follows from the observation that $w(B) \leq \min\{g_1,g_2\}$.  Hence,
the perimeter sum of $Q$ and $P_2$ is at most $\per(C) + \per(C') + a_{1,2} + 2a_{2,3} + b_{1,2} - \ell_1 - \ell_2 < \per(C) + \per(C')$,
a contradiction.

In the case that $\theta_{1,2} > \frac{\pi}{2}$, then we know $\theta_{2,3} \leq \frac{\pi}{2}$.  In this setting, we
obtain a new solution by merging the two clusters $C$ and $C'$ to get $C \cup C'$, and argue that this new solution is cheaper.
Consider the polygon $P$ obtained by removing the portion of edges $e_2$ and $e_3$ in between $x = u_1$ and $x = x_2$,
and then joining $C$ and $C'$ on the left by a vertical line segment going from the point at which $e_2$ intersects $x=u_1$
to the point at which $e_3$ intersects $x=u_1$.  Similarly, we join $C$ and $C'$ on the right by a vertical line segment between
the two points at which edges $e_2$ and $e_3$ intersect at $x=x_2$.  Now, since $P$ contains all points in $C \cup C'$, we have $\per(C \cup C')$
is at most the perimeter of $P$.  Moreover, we have that the perimeter of $P$ is at most $\per(C) + \per(C') + a_{2,3} + b_{2,3} - \ell_2 - \ell_3$.

By a similar argument as earlier, we have $b_{2,3} \leq a_{2,3} + \max\{\ell_2 + \frac{\ell_3}{\sqrt{2}},\frac{\ell_2}{\sqrt{2}} + \ell_3\}$
(using the fact that $\theta_{2,3} \leq \frac{\pi}{2}$ in this case).
Hence, using similar reasoning as before, we obtain
\begin{align*}
a_{2,3} + b_{2,3} - \ell_2 - \ell_3 &\leq 2a_{2,3} + \max\left\{\ell_2 + \frac{\ell_3}{\sqrt{2}},\frac{\ell_2}{\sqrt{2}} + \ell_3\right\} - \ell_2 - \ell_3\\
&\leq 2(\min\{g_2,g_3\} + h(B)) - \left(1 - \frac{1}{\sqrt{2}}\right)\min\{\ell_2,\ell_3\}\\
&\leq 2(\min\{g_2,g_3\} + 4w(B)) - 1000\left(1 - \frac{1}{\sqrt{2}}\right)\min\{g_2,g_3\} < 0,
\end{align*}
where the first inequality follows from substituting for $b_{2,3}$, the second inequality follows from substituting for $a_{2,3}$
and simplifying, the third inequality follows from the assumption $h(B) \leq 4w(B)$, along with
$\ell_2 \geq 1000 g_2$ and $\ell_3 \geq 1000 g_3$, and the last inequality follows from $w(B) \leq \min\{g_2,g_3\}$.
Thus, the perimeter of $P$ is at most $\per(C) + \per(C') + a_{2,3} + b_{2,3} - \ell_2 - \ell_3 < \per(C) + \per(C')$, a contradiction.

In the case that the slope of $e_2$ is at most that of $e_3$, then we also argue that merging $C$ and $C'$ to get $C \cup C'$ yields
a cheaper solution.  In particular, in this case, we know that $b_{2,3} \leq a_{2,3}$ (as we go from $x = u_1$ to $x=x_2$,
the vertical gap between the two edges shrinks since the slope of $e_2$ is smaller than the slope of $e_3$).  Moreover, we have
$a_{2,3} \leq g_2 + g_3 + h(B)$.  By similar reasoning as before,
we obtain
\begin{align*}
a_{2,3} + b_{2,3} - \ell_2 - \ell_3 \leq 2a_{2,3} - \ell_2 - \ell_3 &\leq 2(g_2 + g_3 + 4w(B)) - \ell_2 - \ell_3 \\
&\leq 12\max\{g_2,g_3\} - 1000 g_2 - 1000 g_3 < 0,
\end{align*}
which again yields a contradiction.

In the case that the slope of $e_1$ is at most that of $e_2$, but the slope of $e_2$ is strictly more than that of $e_3$,
we do a similar transformation as previously seen by splitting $C$ into two clusters $C_1$ and $C_2$, and then merging
$C_1$ with $C'$ to obtain the cluster $C_1 \cup C'$.  As before, we need to argue $a_{1,2} + 2a_{2,3} + b_{1,2} - \ell_1 - \ell_2 < 0$.
Note that $b_{1,2} \leq a_{1,2}$, since the vertical gap between $e_1$ and $e_2$ shrinks when going from $x = u_1$ to $x = x_2$.
Moreover, we have $a_{2,3} \leq \min\{g_2,g_3\} + h(B) \leq g_2 + h(B)$ and $a_{1,2} \leq \min\{g_1,g_2\} + h(B) \leq g_2 + h(B)$.
Hence, we have
\begin{align*}
a_{1,2} + 2a_{2,3} + b_{1,2} - \ell_1 - \ell_2 &\leq 2(a_{1,2} + a_{2,3}) - \ell_1 - \ell_2 \leq 4(g_2 + h(B)) - \ell_1 - \ell_2\\
&\leq 4(g_2 + 4w(B)) - 1000 g_1 - 1000 g_2 \leq 20g_2 - 1000 g_1 - 1000 g_2 < 0,
\end{align*}
yielding a contradiction.
\end{proof}

We are now ready to prove Lemma~\ref{lemma:box-split}.

\begin{proof}[Proof of Lemma~\ref{lemma:box-split}]
Given any box $B \in B(S)$, we consider each portion of $B$ that lies between two consecutive vertical help lines
(or between a vertical help line and a vertical main line) to be a vertical strip of $B$.
Similarly, we consider the horizontal strips induced by the horizontal help and main lines in $B$.

First, we consider the case when $B$ is an elementary box.
In this case, $B$ is composed of precisely one vertical strip and one horizontal strip.  In particular, we can apply
Lemma~\ref{lem:twoedgebox} to get that at most two edges in $\COpt_k(S)$ intersect the elementary box $B$ (clearly,
we must have either $h(B) \leq w(B) \leq 4w(B)$ or $w(B) \leq h(B) \leq 4h(B)$).

Now, suppose $B$ is a box satisfying the box invariant.
If $B$ is not an elementary box, then it has at least two vertical
strips or at least two horizontal strips.  If all vertical strips $T$ satisfy $w(T) \leq \frac{w(B)}{200}$
and all horizontal strips $H$ satisfy $h(H) \leq \frac{h(B)}{200}$, then we can apply Lemma~\ref{lem:bidensesep}
to get that there exists either a good vertical separator of $B$ that splits it into two smaller boxes
or a good horizontal separator of $B$ that splits it into two smaller boxes (and is contained
in a help or main line).  Since the separator is good,
at most two edges of $\COpt_k(S)$ intersect it, and hence the box invariant continues to be satisfied
for each of the two smaller boxes.

Hence, suppose there is some vertical strip $T$ satisfying $w(T) > \frac{w(B)}{200}$ or there is some
horizontal strip $H$ satisfying $h(H) > \frac{h(B)}{200}$.  If the box $B$ is much wider than it is tall
and has many small vertical strips (i.e., $w(B) > 4h(B)$ and $w(T) \leq \frac{w(B)}{200}$ for all vertical strips $T$),
then we can apply Lemma~\ref{lem:unbalanceddensesep}.  We get that there is a good vertical separator
(contained in a vertical help or main line) that splits $B$ into two smaller boxes such that the box invariant
continues to hold on each of the smaller boxes (since at most two edges of $\COpt_k(S)$ intersect the vertical separator).
Likewise, if the box is much taller than it is wide and has many small horizontal strips (i.e.,
$h(B) > 4w(B)$ and $h(H) \leq \frac{h(B)}{200}$ for all horizontal strips $H$),
then we can again apply Lemma~\ref{lem:unbalanceddensesep}.  We get that there is a good
horizontal separator (contained in a horizontal help or main line)
that splits $B$ into two smaller boxes, each satisfying the box invariant.

The only other case to consider is when either the box is not much taller than it is wide and
there is some wide vertical strip $T$, or the box is not much wider than it is tall
and there is some tall horizontal strip $H$.  That is, either $h(B) \leq 4w(B)$ and there is some
vertical strip $T$ satisfying $w(T) > \frac{w(B)}{200}$, or $w(B) \leq 4h(B)$ and there
is some horizontal strip $H$ satisfying $w(H) > \frac{h(B)}{200}$.  Suppose the former is the case,
as the proof of the latter is symmetric.  Hence, we can apply
Lemma~\ref{lem:balancedsparsesep} to get that either the left edge of $T$ is intersected
by at most two edges of $\COpt_k(S)$ and the right edge of $T$ is intersected by at most
two edges of $\COpt_k(S)$ (in the latter case, we know the bottom and top edges of $H$ satisfy
a similar property).

If either the left edge of $T$ or the right edge of $T$ serves as a good vertical separator of $B$
(i.e., if either edge splits $B$ into two strictly smaller boxes), we are done since we already
know such an edge of $T$ is intersected by at most two edges of $\COpt_k(S)$.  Hence, we need
only consider the case when the box $B$ consists of one vertical strip (i.e., the left and right edges
of $T$ are $l(B)$ and $r(B)$, respectively).  Since $h(B) \leq 4w(B)$, we can apply
Lemma~\ref{lem:twoedgebox} to get that at most two edges of $\COpt_k(S)$ intersect the entire
box $B$.  Since box $B$ is not elementary, there is either a vertical help or main line
that lies strictly in between $l(B)$ and $r(B)$, or there is a horizontal help or main
line that lies strictly in between $b(B)$ and $t(B)$.  Regardless, the separators
induced by such lines are good.

Tying everything together, for boxes $B$ that are not elementary and satisfy the box invariant,
we get that we can always find a good horizontal or vertical separator (contained
in a help or main line) that splits $B$ into two strictly smaller boxes, both of which also satisfy the box invariant.
\end{proof}

\subsection{Coverings and Signatures}

In the rest of this section, we assume we are given a set of points $S$ on the plane, and a number $k$ and we want to solve the $k$-clustering problem for $S$. Moreover, we assume the boxes we work with are in $\B(S)$, the clusterings we work with are clusterings in $\Part(S)$, the clusters we work with are subsets of $S$, and the edges we work with are in $\G(S)$.
For brevity, we define $\tilde{\C} \coloneqq \COpt_k(S)$.

\begin{figure}
\centering
\includegraphics[scale=1]{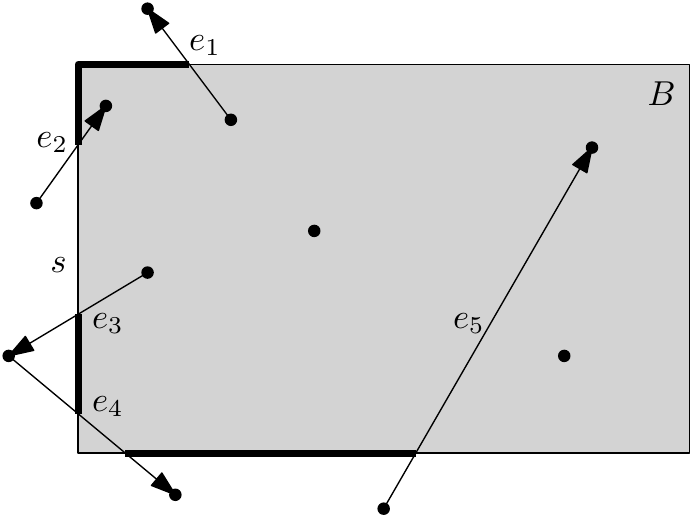}
\caption{The cut number of the set of edges $E\mydef \{e_1,\ldots,e_5\}$ shown in this figure is $2$, i.e.~$\Theta_{B,s}(E)=2$. Note that in this example the border set of $E$ is $E$, i.e.~$\M_B(E) = E$, and it is alternating.}
\label{fig:fixedk:cutnumber}
\end{figure}

Let $B$ be a box and $E$ a set of edges and loops. $E$ is called \emph{interior-disjoint} if the interior of any edge $e\in E$ is not intersected by any other edge in $E$. Indeed, all sets of edges (and loops) we work with in this section are interior-disjoint. The \emph{border set} of $E$ on $B$ is denoted by $\M_B(E)$ and is defined as the set of all edges in $E$ intersecting both the boundary and the interior of $B$. An intersection $q$ of an edge $pp'\in\M_B(E)$ and $\bound B$ is called \emph{entering} if $pq$ does not intersect the interior of $B$, and is called \emph{exiting} if $qp'$ does not intersect the interior of $B$.
Assuming $E$ is interior-disjoint, all the intersections of the edges in its border set and $\bound B$ are distinct and we can define a cyclic order on these intersections by sorting them in counterclockwise order. A border set is called \emph{alternating} if no two consecutive intersections in this cyclic order are both entering or both exiting. Moreover, every two consecutive exiting and entering intersection define an interval on the boundary of $B$. For an edge $s$ of $\bound B$ the \emph{cut number} of $E$ on the edge $s$ is denoted by $\Theta_{B,s}(E)$ and is defined as the number of such intervals (if any) intersecting $s$ (see Figure~\ref{fig:fixedk:cutnumber}).\footnote{To handle the degenerate cases where an edge intersects a corner of a box, we consider $s$ as an open edge by excluding its end points.}
The \emph{border sequence} of $E$ on $B$ is defined as the sequence $e_1,\ldots,e_t$ of all edges in $\M_B(E)$ sorted in counterclockwise order according to their intersection with the boundary of $B$. Note that some edges might appear twice in this sequence. Moreover, this sequence is not unique as the first edge in this sequence can be chosen arbitrarily.

Let $B$ be a box.
A nonempty set of edges $E$ is called a \emph{convex chain} on $B$ of size $t$, where $t\mydef |E|$, if there exists a sequence of distinct points $p_1,\ldots,p_{t+1}$ such that
\begin{itemize}
\item $p_2,\ldots,p_{t}\in S\cap B$ and $p_1,p_{t+1}\in S\setminus B$,
\item \label{def:chain:edges} $E= \{p_1p_2,p_2p_3,\ldots,p_tp_{t+1}\}$,
\item \label{def:chain:simple} $\bigcup_{i=1}^t p_ip_{i+1}$ is a simple open curve intersecting the interior of $B$, and
\item \label{def:chain:convex} for any $i$, where $1\leq i < t$, $p_ip_{i+1}p_{i+2}$ is a left turn.
\end{itemize}
We call this sequence the \emph{vertex sequence} of the convex chain $E$.
The edge $p_1p_2$ is called the \emph{starting} edge and the edge $p_tp_{t+1}$ is called the \emph{ending} edge of the convex chain.
Note that the starting and ending edges of a convex chain are the same when $t=1$.

Let $B$ be a box.
A \emph{cycle} on $B$ is defined to be a nonempty set $E$ with the following properties. We call the set $E$ an \emph{inner} cycle on $B$ of size $t$ if either $t=1$ and $E$ contains a single loop $pp$, where $p\in S\cap B$, or if there exists a sequence of distinct points $\sigma\mydef (p_1,\ldots,p_t)$ such that
\begin{itemize}
\item $p_1,\ldots,p_t\in S\cap B$,
\item \label{def:cycle:edges} $E=\{p_1p_2,p_2p_3,\ldots,p_{t-1}p_t,p_tp_1\}$, and therefore $t=|E|$, and
\item \label{def:cycle:poly} $p_1,\ldots,p_t$ are the vertices of a convex polygon given in counterclockwise direction.
\end{itemize}
An inner cycle consisting of a single loop is called a \emph{loop} cycle on $B$. Moreover, we call the set $E$ a \emph{border} cycle on $B$ if there exists a partition $E=E_1\cup\cdots\cup E_x$, where $1\leq x \leq 4$, such that
\begin{itemize}
\item each $E_i$, where $1\leq i \leq x$, is a convex chain, and
\item $e_1,e_1',\ldots, e_x,e'_x$ is a border sequence of $E$, where $e_i$ and $e'_i$ are the starting and ending edges of the convex chain $E_i$.
\end{itemize}
Note that by the definition of main and helper lines, no point in $S$ can be on the boundary of $B$. It is not hard to verify that the border set of a border cycle is alternating. The \emph{coverage} of a cycle $E$ is denoted by $\R_B(E)$ and is defined either as the set $\{p\}$ if $E$ is a loop cycle $\{pp\}$, or the intersection of all half planes defined by edges in $E$ and the box $B$, i.e.,
$$
\R_B(E)\coloneqq B \cap \bigcap_{e\in E}\hp(e),
$$
where $\hp(e)$ is the half-plane defined by the line passing through $e$ located on the left of $e$.
We say a point $b\in B$ is \emph{covered} by $E$ if it is in the coverage of $E$.

For a given box $B$, an interior-disjoint set $E$ of edges and loops is called a \emph{disjoint covering} or simply a \emph{covering} on $B$ if $E$ satisfies the following properties. The set $E$ is called a \emph{nonempty covering} on $B$ of size $t$ if there exists a partition $E=E_1\cup\cdots\cup E_t$ such that
\begin{itemize}
\item each $E_i$ is a cycle,
\item when $1\leq i<j\leq t$, the coverages of $E_i$ and $E_j$ are disjoint, i.e., $\R_B(E_i)\cap \R_B(E_j) = \emptyset$,
\item $S\cap B \subset \R_B(E_1)\cup \cdots \cup \R_B(E_t)$.
\end{itemize}
It is not hard to verify that the border set of a covering is alternating. We denote the size of $E$ by $\kappa_B(E)$.
By definition, the empty set is a covering on $B$ of size zero if $S\cap B = \emptyset$ and we call it the \emph{empty covering} on $B$.
The following lemma shows that the size of a covering is well defined.
\begin{lemma}
For a nonempty covering $E$ defined on an arbitrary box $B$, the satisfying partition, i.e.~the partition satisfying the conditions for a nonempty covering on $B$, is unique.
\end{lemma}
\begin{proof}
Note that the coverages of each cycle the satisfying partition is a convex polygon containing all parts of its edges in $B$. Therefore as all the coverages are disjoint, no two inner cycles or convex chains that belong to different border cycles intersect inside $B$. Moreover because of the condition on the border sequence of a border cycle, different convex chains of a border cycle cannot intersect in $B$. Hence, any other satisfying partition (if exists) consists of the same convex chains and inner cycles. Additionally, any two convex chain that are in the same border cycle in the satisfying partition must be in the same border cycle in any other satisfying partition as otherwise the coverages of their border cycles intersect. This shows that any other satisfying partition must have the same inner cycles and border cycles which completes the proof.
\end{proof}
It is not hard to verify that for an edge $s$ of $\bound B$ the \emph{cut number} of $E$ on an edge $s$ of $B$, i.e.~$\Theta_{B,s}(E)$, is equal to the number of the coverages of $E$ intersecting $s$.

We say a cluster $C$ \emph{intersects} a box $B$ if its convex hull intersects $B$, and we say it \emph{intersects} the boundary of $B$ or simply intersects $\bound B$ if its convex hull intersects $\bound B$. The \emph{signature} of $C$ on $B$ or its boundary is denoted by $C\sqcap B$ and $C\sqcap \bound B$ respectively and is defined as the set of all edges of its convex hull intersecting $B$ or the boundary of it, i.e.,
\begin{align*}
C\sqcap B&\coloneqq  \{e\in \E(C)\mid e \cap B \neq \emptyset\},&
C\sqcap \bound B&\coloneqq \{e\in \E(C)\mid e \cap \bound B \neq \emptyset\}.
\end{align*}

\begin{figure}
\centering
\includegraphics[scale=1]{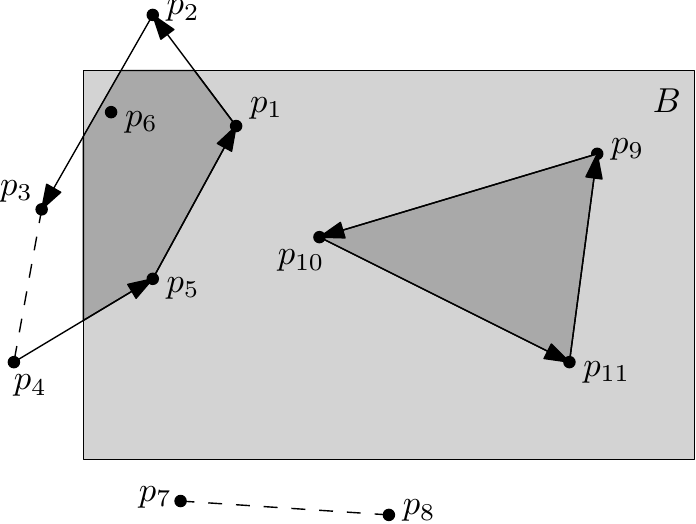}
\caption{The signature of the clustering $\C\mydef \{\{p_1,\ldots,p_6\},\{p_7,p_8\},\{p_9,p_{10},p_{11}\}\}$ on $B$.
The edges of the covering $\C\sqcap B$ are solid and the coverage of each cycle is shown in dark grey.}
\label{fig:fixedk:clusteringcovering}
\end{figure}

With some abuse of notation, the \emph{signature} of a disjoint clustering $\C$ on $B$ or its boundary is denoted by $\C\sqcap B$ and $\C\sqcap \bound B$ respectively and is defined as the set of all edges of convex hulls of its clusters intersecting $B$ or the boundary of it, i.e.,
\begin{align*}
\C\sqcap B&\coloneqq\bigcup_{C\in\C} C\sqcap B,&
\C\sqcap \bound B&\coloneqq\bigcup_{C\in\C} C\sqcap \bound B.
\end{align*}
Note that for any disjoint clustering $\C$, $\C\sqcap B$ is a covering on $B$ (see Figure~\ref{fig:fixedk:clusteringcovering}).

\subsection{Subproblems}

Let $B$ be a box, $M$ be an interior-disjoint set of edges with $|M|\leq 8$, $k'\in\{0,\ldots, n\}$, and $\delta\in\{T,F\}$ a boolean value. The tuple $I\mydef \langle B,M, k',\delta\rangle$ is called a \emph{subproblem} on $B$ or simply a subproblem if it satisfies the following properties.
The tuple $I$ is called a \emph{special} subproblem if
\begin{itemize}
\item $\delta=T$,
\item $M=\emptyset$, and
\item $k'=1$.
\end{itemize}
The empty covering (on $B$) is the only \emph{solution} to a special subproblem $\langle B, \emptyset, 1,T\rangle$.
The subproblem $I$ is called a \emph{normal} subproblem if
\begin{itemize}
\item $\delta=F$,
\item $M$ is equal to its border set, i.e.~$M=\M_B(M)$, and
\item border set of $M$, i.e.~$M$, is alternating.
\end{itemize}
A covering $E$ on $B$ is called a \emph{solution} to a normal subproblem $\langle B, M, k',F\rangle$ if $\M_B(E)=M$ and $\kappa_B(E)=k'$.
The \emph{cost} of a solution is denoted by $\Phi_B(E)$ and is defined as the sum of the length of $e\cap B$ over all edges $e\in E$. A solution to $I$ is called \emph{optimal} if there is no solution to $I$ with smaller cost. We say a point $b\in B$ is \emph{covered} by a solution $E$ if either $b$ is covered by $E$ or $\delta=T$. Note that the coverage status of any point on the boundary of $B$ is the similar for all solutions.

The \emph{cut number} of a subproblem $I\mydef \langle B,M, k',\delta\rangle$ on an edge $s$ of $B$ is denoted by $\Theta_{B,s}(M,\delta)$ and is defined as
\begin{equation*}
\Theta_{B,s}(M,\delta)\coloneqq
\begin{cases}
1, &\delta=T,\\
\Theta_{B,s}(M), & \text{otherwise}.
\end{cases}
\end{equation*}
A disjoint clustering $\C$ \emph{respects} a subproblem $I\mydef \langle B,M, k',\delta\rangle$ if either
\begin{itemize}
\item $\delta=T$ and $B$ is in the convex hull of a cluster in $\C$, or
\item $\delta=F$ and $\C \sqcap B$ is a solution to $I$.
\end{itemize}

\begin{observation}
\label{obs:unique-subproblem}
Given a box $B$, the subproblem $I\mydef \langle B,k',\tCC\sqcap \bound B,\delta\rangle$ is the unique subproblem on $B$ respected by $\tCC$, where $k'$ is the number of clusters in $\tCC$ intersecting $B$ and $\delta$ is $T$ if $B$ is contained in the convex hull of a cluster in $\tCC$ and is $F$ otherwise.
Moreover, $\tCC \sqcap B$ is a solution to $I$.
\end{observation}

Note that the original clustering problem is also a subproblem $I_0\coloneqq\langle B_0,\emptyset,k,F\rangle$, where $B_0$ is defined as the box with edges on $v_1^-,v_n^+,h_1^-,h_n^+$.

\begin{observation}\label{obs:globalSol}
For any solution $E$ to $I_0$, there is a disjoint clustering $\cl\in\Part(S)$ such that $\cl\sqcap B_0=E$.
Likewise, for any disjoint clustering $\cl\in\Part(S)$, it holds that $\cl\sqcap B_0$ is a solution to $I_0$.
\end{observation}

The following property shows that the optimal clustering $\tCC$ is also optimal for subproblems.
\begin{lemma}
\label{lemma:optimality}
Let $I\mydef \langle B,M,k',\delta\rangle$ be a subproblem on a box $B$ respected by $\tCC$. Then, for any solution $E$ to $I$, $\Phi_B(\tCC)\leq \Phi_B(E)$.
\end{lemma}
\begin{proof}
Assume there exist a solution $E$ to $I$ such that
$$
\Phi_B(E) < \Phi_B(\tCC).
$$
It is easy to verify that $E' \mydef ((\tCC\sqcap B_0) \setminus (\tCC\sqcap B)) \cup E$ is a solution to $I_0$ such that $\Phi_{B_0}(E')<\Phi_{B_0}(\tCC)$.
As $E'$ consists of $k$ inner cycles on $B_0$, the set $E'$ is the boundary of a clustering $\cl'$ with cost less than $\tCC$, which is a contradiction.
\end{proof}

\subsection{Split and Merge}

We say that a box $B$ \emph{splits} into $B_l$ and $B_r$ if there is a separator $s$ that separates $B$ to $B_l$ and $B_r$.
Here, the box $B_l$ is the upper box if $s$ is horizontal and the left box if $s$ is vertical.

Consider two boxes $B_1$ and $B_2$ whose boundaries share an edge.
We say two subproblems $\langle B_1,M_1,k'_1,\delta_1\rangle$ and $\langle B_2,M_2,k'_2,\delta_2\rangle$ are \emph{compatible} if their edges match on their shared boundary, i.e.,
\begin{equation*}
\{e\in M_1| e\cap \bound B_2 \neq \emptyset\} = \{e\in M_2| e\cap \bound B_1 \neq \emptyset\}.
\end{equation*}

Let $B$ be a box that splits into $B_l$ and $B_r$ by a separator $s$. We say a subproblem $I \mydef \langle B,M,k',\delta\rangle$ \emph{splits} into $I_l \mydef \langle B_l,M_l,k'_l,\delta_l\rangle$ and $I_r \mydef \langle B_r,M_r,k'_r,\delta_r\rangle$ if $I_l$ and $I_r$ are compatible, $\Theta_{B,s}(M_l,\delta_l)=\Theta_{B,s}(M_r,\delta_r)$, $k'=k'_l+k'_r-\Theta_{B,s}(M_l,\delta_l)$, and $\delta=\delta_l \land \delta_r$.
Equivalently, we say that $I_l$ and $I_r$ \emph{merge} to $I$.

\begin{observation}
\label{obs:sub-split}
Let $B$ be a box that splits into boxes $B_l$ and $B_r$ by a separator $s$, and $\cl$ be a disjoint clustering.
If $\cl$ respects all subproblems $I\mydef \langle B,M,k',\delta\rangle$, $I_l\mydef \langle B_l,M_l,k'_l,\delta_l\rangle$, and $I_r\mydef \langle B_r,M_r,k'_r,\delta_r\rangle$, then the subproblem $I$ splits into $I_l$ and $I_r$.
\end{observation}

The following lemma provides sufficient conditions for a subproblem to have a solution and helps us to find a solution for a subproblem satisfying those conditions.
\begin{lemma}
\label{lemma:join-iff}
Let $I \mydef \langle B,M,k',\delta\rangle$ be a subproblem that splits into subproblems $I_l \mydef \langle B_l,M_l,k'_l,\delta_l\rangle$ and $I_r \mydef \langle B_r,M_r,k'_r,\delta_r\rangle$.
Suppose that $E_l$ and $E_r$ are solutions to $I_l$ and $I_r$, respectively.
Then, $E \mydef E_l\cup E_r$ is a covering on $B$ and a solution to $I$.
Furthermore, $\Phi_B(I)=\Phi_{B_l}(E_l)+\Phi_{B_r}(E_r)$.
\end{lemma}
\begin{proof}
First assume that none of the subproblems is special. As $I_l$ and $I_r$ are compatible and $\Theta_{B,s}(E_l,\delta_l)=\Theta_{B,s}(E_r,\delta_r)$, each border cycle in $E_l$ intersected by $s$ has a corresponding border cycle in $E_r$. The set $E$ is a covering on $B$ because each pair of corresponding border cycles intersected by $s$ merge to a cycle on $B$, and all other cycles from $E_l$ and $E_r$ are also cycles in $B$. Moreover, as the number of the corresponding pairs of border cycles is $\Theta_{B,s}(E_l,\delta_l)$, the number of cycles in $E$ is $\kappa_B(E)=k'_l+k'_r-\Theta_{B,s}(E_l,\delta_l)$. By the definition of a solution, the covering $E$ is a solution to a subproblem $\langle B, \M_B(E), k'_l+k'_r-\Theta_{B,s}(E_l,\delta_l),F\rangle$. However, as $E=E_l\cup E_r$, the border set $\M_B(E)$ is simply all the edges in $M_l$ and $M_r$ that intersect the boundary of $B$. As $I$ splits into $I_l$ and $I_r$, we have $M=\M_B(E)$ and $k'=k'_l+k'_r-\Theta_{B,s}(E_l,\delta_l)$. Moreover, as the edges of $E$ are simply all the edges of $E_l$ and $E_r$ combined, the cost of the covering $E$ satisfies the condition mentioned in the statement of the lemma.

It is easy to verify that in case one or both of the subproblems are special, the claim still holds, which completes the proof of the lemma.
\end{proof}

\subsection{Elementary Subproblems}

We call a subproblem \emph{elementary} if it is defined on an elementary box. In this section, we show that any elementary subproblem $\langle B, M, k',\delta\rangle$ with at most two edges, i.e., $|M|\leq2$, has at most one solution.
Moreover, if such a subproblem has a solution we can recognize it and find its unique solution in constant time.

An elementary box contains either one point from $S$ or is empty. Therefore, if an elementary box $B$ is empty, all cycles on $B$ are border cycles.
If, on the other hand, $B\cap S=\{p\}$, then the only inner cycle on $B$ is the loop cycle $\{pp\}$.
Furthermore, as $B$ is defined by two pairs of consecutive main lines, $v_i^-,v_i^+$ and $h_j^-,h_j^+$, and its width and height is zero, any edge intersecting $B$ contains $p$.
The only covering on $B$ with an inner cycle is $\{pp\}$, and all other coverings consists of border cycles only. Note that if $\{pp\}$ is a solution to a subproblem, there is no other solution, as in that case $M=\emptyset$ and the loop cycle $\{pp\}$ is the only way $p$ can possibly be covered.

Now, consider a solution without any inner cycles to an elementary subproblem. As an elementary box $B$ contains at most one point, the vertex sequence of a convex chain on $B$ has a size at most three, which means it has at most two edges, and both edges must intersect the boundary of $B$. Therefore, for a covering $E$ on an elementary box $B$ that does not have any inner cycles, we have $\M_B(E)=E$.

By considering the possible cases, we make the following observation.

\begin{observation}
\label{obs:elementary-iff}
An elementary subproblem $I \mydef \langle B,M,k',\delta\rangle$ has a solution if and only if either
\begin{itemize}
\item $I=\langle B,\emptyset, 1,F\rangle$ and $S \cap B=\{p\}$,
\item $I=\langle B,\emptyset, 0,F\rangle$ and $S \cap B=\emptyset$,
\item $I=\langle B,\emptyset, 1,T\rangle$, or
\item $I=\langle B,M,k',F\rangle$, where $M$ is a covering on $B$ of size $k'$.
\end{itemize}
Moreover, if a solution exists, it is unique and is equal to $\{pp\}$ in the first case and $M$ otherwise.
\end{observation}
Note that when $|M|\leq 2$, we can decide if a set of edges is a solution to an elementary subproblem in constant time.

\subsection{Dynamic Programming Algorithm}

\begin{definition}
A subproblem $I\mydef \langle B,M,k',\delta\rangle$ is \emph{compact} if for each edge of $B$, at most two edges in $M$ intersect it.
\end{definition}

\begin{observation}
\label{obs:compact-valid}
Let $B$ be a box satisfying the box invariant and $I$ be a subproblem on $B$ such that $\tCC$ respects $I$.
Then $I$ is compact.
\end{observation}

\begin{algorithm}
\LinesNumbered
\DontPrintSemicolon
\SetArgSty{}
\SetKwIF{If}{ElseIf}{Else}{if}{}{else if}{else}{end if}
    \For {$b=1,\ldots,(20000n)^2$} { \label{alg:b-for}
        \For {all boxes $B$ consisting of $b$ elementary boxes} { \label{alg:box-for}
            \For {all compact subproblems $I \mydef \langle B,M,k',\delta\rangle$} { \label{alg:sub-for}
                $\Table(I)\coloneqq\text{Impossible}$\;
                \uIf {$b=1$} {
			\uIf {$M=\emptyset$ and $k'=1$ and $\delta=F$ and $B \cap S = \{p\}$} { \label{alg:noedge-if}
				$\Table(I)\coloneqq \{pp\}$\; \label{alg:point-store}
			} \ElseIf {$|M|\leq 2$ and $M$ is a solution to $I$} { \label{alg:edge-if}
                    	$\Table(I)\coloneqq M$\;\label{alg:edge-store}
			}
                } \Else {
                    \For {all compact subproblems $I_l$ and $I_r$ that merge to $I$} { \label{alg:split-for}
			      \If {$\Table(I_l)\neq\text{Impossible}$ and $\Table(I_r)\neq\text{Impossible}$} {
				    $E\coloneqq\Table(I_l)\cup \Table(I_r)$\; \label{alg:join}
  				     \If  {$\Table(I)=\text{Impossible}$ or $\Phi_B(E) \leq \Phi_B(\Table(I))$} {
                                 $\Table(I)\coloneqq E$\;\label{alg:merge-store}
				     }
			     }
                    }
                }
            }
        }
    }
    \caption{}
\label{ALG1}
\end{algorithm}

Algorithm~\ref{ALG1} solves the minimum perimeter sum problem. For a subproblem $I\mydef \langle B,M,k',\delta\rangle$, the entry $\Table(I)$ either stores a solution to $I$ or the value `Impossible', which means we have not found any solution to $I$, though it does not mean a solution to $I$ does not exist. After the execution of the algorithm, $\Table(\langle B_0,\emptyset, k,F\rangle)$ contains the convex hull of an optimal $k$-clustering.
The algorithm iterates over all compact subproblems sorted ordered by the size of the defining box (i.e., the number of elementary boxes contained in the box) and stores a solution for some of them in $\Table$.
\begin{itemize}
\item For an elementary subproblem $I$ that has a solution, the algorithm stores the unique solution in $\Table(I)$ only if $|M|\leq 2$. The algorithm uses the conditions stated in Observation~\ref{obs:elementary-iff} to find the unique solution to such subproblems. In the proof of Lemma~\ref{lemma:fixed-k}, we show that we do not need to find a solution to elementary subproblem with $|M|>2$.
\item For a subproblem $I$ defined on a non-elementary box, the algorithm finds all possible splits of $I$ to two compact subproblems $I_l$ and $I_r$. For each such split, it combines the coverings (if such exist) in $\Table(I_l)$ and $\Table(I_r)$ to get a covering $E$. Then, it stores $E$ in $\Table(I)$ only if it is a better solution to $I$ than already stored in $\Table(I)$. By Lemma~\ref{lemma:join-iff}, the union of the two solutions to $I_l$ and $I_r$ is a solution to $I$.
\end{itemize}

\begin{lemma}
\label{lemma:fixed-k}
Let $I$ be a subproblem on a box $B$.
Consider the value of $\Table(I)$ as Algorithm~\ref{ALG1} terminates.
If $B$ satisfies the box invariant and $\tCC$ respects $I$, the entry $\Table(I)$ contains a solution to $I$ of cost $\Phi_B(\tCC)$.
In particular, $\Table(I_0)$ contains a solution to $I_0\coloneqq\langle B_0,\emptyset,k,F\rangle$ of cost $\Phi_{B_0}(\tCC)$, and hence, the algorithm solves the \fixedkname.
\end{lemma}

\begin{proof}
The proof is by induction on the number $b$ of elementary boxes contained in $B$.
Consider first the case $b=1$.
If $I$ is respected by $\tCC$, Lemma~\ref{lemma:box-split} yields that $B$ is intersected by at most two edges of the signature of $\tCC$. By Observation~\ref{obs:borderpoint}, there are no points from $S$ on the boundary of $B$,  which shows $|M|\leq 2$ and therefore $I$ is compact and will be generated. Depending on the subproblem, the algorithm stores either $\{pp\}$ or $M$ as a solution to $I$ in $\Table(I)$ only if $M$ is a solution to $I$. As by Observation~\ref{obs:elementary-iff} the solution to $I$ is unique, it must be $\tCC \sqcap B$. If $|M|> 2$, the algorithm stores `Impossible' in $\Table(I)$ which completes the proof, for the base case.

Suppose the claim holds for boxes consisting of up to $b-1$ elementary boxes for some $b\geq 2$ and consider a subproblem $I\mydef \langle B,M,k',\delta)$ where $B$ consists of $b$ elementary boxes. Note that if $\Table(I)$ is not $Impossible$, $\Table(I)$ must contain a value $\Table(I_l)\cup\Table(I_r)$ for some subproblems $I_l$ and $I_r$ that merge to $I$ and are defined on boxes consisting of less than $b$ elementary boxes. By the induction hypothesis, $\Table(I_l)$ and $\Table(I_r)$ contain solutions to $I_l$ and $I_r$ and by Lemma~\ref{lemma:join-iff}, $\Table(I)=\Table(I_l)\cup\Table(I_r)$ is a solution to $I$.

If $\tCC$ respects $I$ and $B$ satisfies the box invariant, by Observation~\ref{obs:compact-valid}, $I$ is compact. By Lemma~\ref{lemma:box-split}, $B$ has a separator $s$ contained in a main or helper line and $s$ separates $B$ into two boxes $B_l$ and $B_r$ both satisfying the box invariant.
By Observation~\ref{obs:unique-subproblem}, there are unique subproblems $I_l$ and $I_r$, on $B_l$ and $B_r$, such that
$\tCC$ respects them. By Observation~\ref{obs:sub-split}, subproblem $I$ splits into $I_l$ and $I_r$.
As $B_l$ and $B_r$ consists of less than $b$ elementary boxes, by the induction hypothesis, $\Table(I_l) = \Phi_{B_l}(\tCC)$ and $\Table(I_r) = \Phi_{B_r}(\tCC)$.  By Lemma~\ref{lemma:join-iff},
\begin{equation}
\label{eq:fixed-k-1}
\Phi(\Table(I))\leq\Phi_{B_l}(\tCC)+\Phi_{B_r}(\tCC)=\Phi_{B}(\tCC),
\end{equation}
where the equality comes from Observation~\ref{obs:sumproperty}. Furthermore, by the above mentioned argument, $\Table(I)$ is a solution to $I$ and by Lemma~\ref{lemma:optimality},
\begin{equation}
\label{eq:fixed-k-2}
\Phi(\Table(I))\geq\Phi_{B}(\tCC),
\end{equation}
which completes the proof.
\end{proof}

\subsection{Run-time Analysis}\label{sec:datastructures}

\paragraph{Data structures.} We implement $\Table$ as an array, so that we can update and insert values at a given entry $\Table(I)$ in time $O(1)$.

We use a special representation for solutions (set of edges and loops) by making a directed acyclic graph (DAG) to support constant time union operation on solutions (sets).
Each solution $E$ to a subproblem $\langle B,M,k',\delta\rangle$ is either the set $M$, where $|M| \leq 2$, or a set containing a single loop for elementary subproblems, or is the union of two solutions $E_l$ and $E_r$ to two subproblems defined on smaller boxes $B_l$ and $B_r$ for non-elementary subproblems. For elementary subproblems we simply make a node and store if the solution is a loop or not, and if it is a loop we store the loop in the node. For non-elementary subproblems we make a node and we store the (at most two) edges in the solution that intersect the shared edge of $B_l$ and $B_r$ but do not intersect $\bound B$. We also store two pointers to $E_l$ and $E_r$ in this node; these pointers form a DAG where the number of nodes visible to each node (solution) is linear and therefore it is possible to find the edges and loops of a solution in linear time by simply adding up all the loops and edges inside the visible nodes to $M$.

For a box $B$ and an edge $s$ of $B$, we can compute the value of $\Theta_{B,s}(M,\delta)$ for a subproblem $I\mydef \langle B, M,k',\delta\rangle$ in $O(1)$ time as follows. For each subproblem, we store the coverage status of the four corners of the defining box. It is easy to compute those values for elementary subproblems and after merging two solutions we can easily merge the coverage status of their corners as well.\footnote{To handle degenerate cases where an edge intersects a corner of a box, we can store two coverage values for each corner. One for the coverage of the corner point and a bit more of the boundary in counterclockwise direction and one for the coverage of the corner point and a bit more of the boundary in clockwise direction.} Let $m_s$ be the number edges in $M$ intersecting $s$.
As we just work with compact subproblems, $0\leq m_s \leq 2$. Having the coverage status of the four corners of $B$, assuming $s_1$ and $s_2$ are the endpoints of $s$, the cut number is given by
\begin{equation*}
\Theta_{B,s}(M,\delta) =
\begin{cases}
2, &  \text{$m_s=2$ and  $s_1$ and $s_2$ are covered},\\
0, &  \text{$m_s =0$ and $s_1$ and $s_2$ are not covered},\\
1, &  \text{otherwise.}
\end{cases}
\end{equation*}

\begin{observation}
\label{obs:subproblems-count}
Given a box $B$,
the number of compact subproblems $I\mydef \langle B,M,k',\delta\rangle$ is at most
$$(n+1)\cdot \binom{O(n^2)}{8}=O(n^{17}).$$
Here, $n+1$ is a bound on the number $k'$ and the factor $\binom{O(n^2)}{8}$ is a bound on the number of ways the edges of $M$ can be selected for a compact subproblem. There is also one special subproblem on box $B$ and we can consider that case as well.
\end{observation}

\begin{observation}
\label{obs:splits-count}
Let $B$ be a box that splits into $B_l$ and $B_r$. The number of ways to split a compact subproblem $I\mydef \langle B,M,k',\delta\rangle$ to two compact subproblems $I_l\mydef \langle B_l,M_l,k'_l,\delta_l\rangle$ and $I_r\mydef \langle B_r,M_r,k'_r,\delta_r\rangle$ is at most
$$n\cdot \binom{O(n^2)}{2}=O(n^{5}).$$
Here $n$ is the number of ways we can choose a value for $k'_l$ and $\binom{O(n^2)}{2}$ is the number subsets of at most two edges intersecting $\bound B_l \cap \bound B_r$.
\end{observation}

We now describe how Algorithm~\ref{ALG1} at line~\ref{alg:split-for} iterates through the $O(n^5)$ pairs of subproblems $(I_l,I_r)$ merging to $I$.
At first, assume $\delta_l$ and $\delta_r$ have the value $F$.
The edges $M$ specify all edges in $M_l$ and $M_r$ except possibly some edges intersecting the edge $s\mydef \partial M_l\cap\partial M_r$ and not other edges of $\partial M_l$ or $\partial M_r$.
As $I_l$ and $I_r$ are compact and compatible, $M_l$ and $M_r$ can have at most two such edges.
Consider now one of these $\binom{O(n^2)}{2}$ choices of edges crossing $s$.
We iterate through each value $k'_l\in\{0,\ldots,k'\}$.
For a given value of $k'_l$, the subproblem $I_l$ is completely specified.
If $I_l$ does not have any solution in $\Table(I_l)$, we can proceed with another choice of $I_l$.
Otherwise, note that $I_r$ is also specified except for the value $k'_r$.
However, we can compute $\Theta_{B_l,s}(E,\delta_l)$, where $E$ is the solution stored as $\Table(I_l)$, in constant time. In order for $I_l$ and $I_r$ to merge to $I$, we now define $k'_r\mydef k'-k'_l+\Theta_{B_l,s}(M_l,\delta_l)$.
There are four different choices for the values $\delta_l$ and $\delta_r$ and it is not hard to consider those cases as well.

\begin{theorem}
\label{thm:fixed-k}
Algorithm~\ref{ALG1} solves the minimum perimeter sum problem in time $O(n^{27})$.
\end{theorem}
\begin{proof}
Note that $B_0$, the box with edges on $v_1^-,v_n^+,h_1^-,h_n^+$, satisfies the box invariant. Since $\tCC$ respects $I_0\mydef\langle B_0,\emptyset,k,F\rangle$, Lemma~\ref{lemma:fixed-k} gives that as the algorithm terminates, $\Table(I_0)$ contains a solution $E$ to $I_0$ of cost $\Phi_{B_0}(\tCC)=\Phi(\tCC)$. Moreover, Observation~\ref{obs:globalSol} states that there is a clustering $\cl$ such that $\cl\sqcap B_0=E$, and hence, $\cl$ is an optimal clustering.
Furthermore, the clustering $\cl$ can easily be computed given $E$.

The algorithm uses the conditions specified in Observation~\ref{obs:elementary-iff} to verify solutions to elementary subproblems $I\mydef \langle B,M,k',\delta\rangle$ satisfying $|M|\leq 2$, which takes constant time and does not affect the asymptotic running time of the algorithm.

There are $O(n^4)$ different boxes in $\B(S)$. By Observation~\ref{obs:subproblems-count}, the number of compact subproblems on a $B$ is $O(n^{17})$. For each separator $s$ of $B$, by Observation~\ref{obs:splits-count}, the number of ways to split $I$ to $I_l$ and $I_l$ is $O(n^5)$. As we showed in the beginning of this section, merging two solutions can be done in constant time. 
The cost of each solution can be stored as a single real number and comparing could be assumed to take only $O(1)$ time.%
\footnote{In the RAM model, assuming comparing two solutions can be done in time $T(n)$, the running time is $O(n^{27} \cdot T(n))$.}
The running time of the algorithm thus is
$$
O(n^4)\cdot O(n^{17})\cdot O(n) \cdot O(n^5) \cdot O(1) = O(n^{27}),
$$
where $O(n^4)$ is the number of different boxes $B\in\B(S)$, the factor $O(n^{17})$ is the number of subproblems on $B$,  $O(n)$ is the number of separators of $B$, $O(n^5)$ is the number of ways to split a subproblem into two subproblems for a given separator $s$, and $O(1)$ is the time needed to compare two solutions.
\end{proof}

\begin{theorem}[The $k$-cluster fencing problem]
\label{thm:fixedk}
There is a polynomial time algorithm that, given any set $S$ of $n$ points in the plane
and an integer $k$, finds a set of at most $k$ closed curves such that
each point in $S$ is enclosed by a curve and the total length of the curves is minimized.
\end{theorem}

\section{Acknowledgments}
The authors are very grateful to Joseph S. B. Mitchell for valuable discussions and bringing some related works to our attention.
The authors also thank the reviewers for their insightful comments and for improving the quality of the paper.

The work of M. Abrahamsen, A. Roytman, and M. Thorup is partially supported by
Thorup's Advanced Grant from the Danish Council for Independent Research under Grant No. DFF-0602-02499B,
and partly by Basic Algorithms Research Copenhagen (BARC), supported by
Thorup's Investigator Grant from the Villum Foundation under Grant No. 16582.
The work of A. Adamaszek is supported by the
Danish Council for Independent Research DFF-MOBILEX mobility grant.
The work of V. Cohen-Addad was done as a Postdoctoral Fellow at the University of Copenhagen and supported by the
European Union's Horizon 2020 research and innovation programme under the Marie Sklodowska-Curie Grant Agreement No. 748094.
V. Cohen-Addad was also part of BARC.
M.~Mehr is supported by the 
Netherlands Organisation for Scientific Research (NWO) under Project No. 022.005.025.

\bibliographystyle{plain}
\bibliography{minpersum}

\end{document}